\documentclass{amsart}
\usepackage{graphicx} 
\usepackage{epsfig}

\usepackage{mathrsfs}

\usepackage{bbm}

\DeclareMathAlphabet      {\mathbfit}{OML}{cmm}{b}{it}
\let\bm=\mathbfit

\let\text=\mbox

\catcode`\@=11
\let\ced=\c
\renewcommand{\a}{\alpha}
\renewcommand{\b}{\beta}
\renewcommand{\d}{\delta}

\newcommand{\q}{\quad}

\newcommand{\s}{\sigma}

\newcommand{\cal}{\mathcal}

\newcommand{\M}{{\cal M}}

\newcommand{\ty}{\infty}
\newcommand{\e}{\varepsilon}

\newcommand{\ov}[1]{\overline{#1}}

\renewcommand{\O}{\Omega}
\newcommand{\pa}{\partial}
\newcommand{\I}{\mathbbm{i}}
\newcommand{\D}{\mathrm{d}}

\newcommand{\stq}{\subseteq}

\newcount\br@j
\newcommand{\udesno}[1]{\unskip\nobreak\hfil\penalty50\hskip1em\hbox{}
             \nobreak\hfil{#1\unskip\ignorespaces}
                 \parfillskip=\z@ \finalhyphendemerits=\z@\par
                 \parfillskip=0pt plus 1fil}
\catcode`\@=11

\newcommand{\eR}{\mathbb{R}}
\newcommand{\eN}{\mathbb{N}}
\newcommand{\Ze}{\mathbb{Z}}

\newcommand{\Ce}{\mathbb{C}}

\newcommand{\re}{\mathop{\mathrm{Re}}}
\newcommand{\im}{\mathop{\mathrm{Im}}}

\newcommand{\po}{{\mathop{\mathcal P}}}

\newcommand{\res}{\operatorname{res}}

\newcommand{\sideremark}[1]{\ifvmode\leavevmode\fi\vadjust{\vbox to0pt{\vss 
      \hbox to 0pt{\hskip\hsize\hskip1em           
 \vbox{\hsize2cm\tiny\raggedright\pretolerance10000
 \noindent #1\hfill}\hss}\vbox to8pt{\vfil}\vss}}}%

\newcommand{\ovb}[1]{\mkern 1.5mu\overline{\mkern-1.5mu#1\mkern-1.5mu}\mkern 1.5mu}

\newcommand{\E}{\mathrm{e}}
\newcommand{\di}{\,\mathrm{d}}

\address[lapidus@math.ucr.edu]{Michel L.\ Lapidus,  University of California, Riverside, Department of Mathematics, 900 University Avenue, Riverside, CA 92521-0135, USA}

\address[goran.radunovic@fer.hr]{Goran Radunovi\'c, University of Zagreb, Faculty of Electrical Engineering and Computing, Department of Applied Mathematics, Unska 3, 10000 Zagreb, Croatia}

\address[darko.zubrinic@fer.hr]{Darko \v Zubrini\'c, University of Zagreb, Faculty of Electrical Engineering and Computing, Department of Applied Mathematics, Unska 3, 10000 Zagreb, Croatia}

\newtheorem{theorem}{Theorem}[section]
\newtheorem{corollary}[theorem]{Corollary}
\newtheorem{lemma}[theorem]{Lemma}
\newtheorem{proposition}[theorem]{Proposition}

\theoremstyle{definition}
\newtheorem{definition}[theorem]{Definition}
\newtheorem{example}[theorem]{Example}
\newtheorem{remark}[theorem]{Remark}

\numberwithin{equation}{section}

\title[Fractal Tube Formulas for Compact Sets and Relative Fractal Drums]{Fractal Tube Formulas for Compact Sets and Relative Fractal Drums: Oscillations, Complex Dimensions and Fractality}

\author[M.\ L.\ Lapidus, G.\ Radunovi\'c and D.\ \v Zubrini\' c]{Michel L.\ Lapidus, Goran Radunovi\'c and Darko \v Zubrini\' c}

\thanks{The work of Michel L.\ Lapidus was partially supported by the US National Science Foundation (NSF) under the research grants DMS-0707524 and DMS-1107750, as well as by the Institut des Hautes Etudes Scientifiques (IHES) in Paris/Bures-sur-Yvette, France, where the first author was a visiting professor in the Spring of 2012 while part of this work was completed.}

\thanks{The research of Goran Radunovi\'c and Darko \v Zubrini\'c was supported by the Croatian Science Foundation under the project IP-2014-09-2285 and by the Franco-Croatian 
PHC-COGITO~project.}

\begin{document}

\begin{abstract}
We establish pointwise and distributional fractal tube formulas for a large class of relative fractal drums in Euclidean spaces of arbitrary dimensions.
A relative fractal drum (or RFD, in short) is an ordered pair $(A,\O)$ of subsets of the Euclidean space (under some mild assumptions) which generalizes the notion of a (compact) subset and that of a fractal string.
By a fractal tube formula for an RFD $(A,\O)$, we mean an explicit expression for the volume of the $t$-neighborhood of $A$ intersected by $\O$ as a sum of residues of a suitable meromorphic function (here, a fractal zeta function) over the complex dimensions of the RFD $(A,\O)$.
The complex dimensions of an RFD are defined as the poles of its meromorphically continued fractal zeta function (namely,  the distance or the tube zeta function), which generalizes the well-known geometric zeta function for fractal strings.
These fractal tube formulas generalize in a significant way to higher dimensions the corresponding ones previously obtained for fractal strings by the first author and van Frankenhuijsen and later on, by the first author, Pearse and Winter in the case of fractal sprays.
They are illustrated by several interesting examples which demonstrate the various phenomena that may occur in the present general situation.
These examples include fractal strings, the Sierpi\'nski gasket and the 3-dimensional carpet, fractal nests and geometric chirps, as well as self-similar fractal sprays.
We also propose a new definition of fractality according to which a bounded set (or RFD) is considered to be fractal if it possesses at least one nonreal complex dimension or if its fractal zeta function possesses a natural boundary.
This definition, which extends to RFDs and arbitrary bounded subsets of $\eR^N$ the previous one introduced in the context of fractal strings, is illustrated by the Cantor graph (or devil's staircase) RFD, which is shown to be `subcritically fractal'. 
\end{abstract}

\bigskip



\maketitle

\tableofcontents

\section{Introduction}

In this paper our main goal is to obtain {\em fractal tube formulas} for a class of {\em relative fractal drums} in Euclidean spaces of arbitrary dimension.
The fractal tube formulas are interesting since, roughly speaking, they describe the ``fractality'' of the set (or relative fractal drum) in a more detailed way than, for instance, the mere Minkowski (or box) dimension which, by definition, only corresponds to a part of the leading term of these tube formulas.
The main results of this paper are the obtained expressions of the fractal tube formulas for relative fractal drums in terms of sums of the residues of the fractal zeta functions, i.e., the {\em distance} and {\em tube zeta functions} associated to these relative fractal drums.
Furthermore, the sum in these expressions will be taken over the set of poles of the fractal zeta function at hand, also called the set of {\em complex dimensions} of a given relative fractal drum.
In short, we will show that (after a suitable meromorphic continuation), the relative distance (or tube) zeta function encodes the information about the inner geometry of a relative fractal drum into the distribution of its poles as well as into the values of the corresponding residues.

The notion of a {\em relative fractal drum} which is defined, roughly, as an ordered pair $(A,\O)$ of subsets of $\eR^N$ (see Definition \ref{drum}), was introduced in \cite{refds} (see also \cite{fzf}) in order to generalize the notion of a bounded subset as well as of a fractal string.
It enables us to study a wide range of fractal phenomena; for instance, the corresponding {\em relative Minkowski $($or box$)$ dimension} (see Equation \eqref{dimrel}) can attain negative values, including $-\infty$.
In particular, we stress that the results of this paper can be applied to bounded subsets of $\eR^N$, with $N\geq 1$ arbitrary.
In fact, the theory developed in this paper, which gives an explicit connection between complex dimensions and fractal tube formulas, very significantly generalizes the analogous theory already developed for fractal strings (i.e., when $N=1$) by the first author and M.\ van Frankenhuijsen (see \cite[Chapter 8]{lapidusfrank12} and also [Lap-vFr1--2]). 
It also broadly extends the later work of the first author and E.\ Pearse [LapPe1--2] 
and that of those two authors and S.\ Winter \cite{lappewi1} on fractal tube formulas for fractal sprays (higher-dimensional generalizations of fractal strings, in the sense of \cite{LapPo2}).

By a fractal tube formula of the relative fractal drum $(A,\O)$, we mean an exact or asymptotic expansion of the relative tube function $t\mapsto|A_t\cap\O|_N$ as $t\to0^+$, where $A_t$ is the $t$-neighborhood of $A$ (i.e., the set of points of $\eR^N$ within a distance less than $t$ from $A$) and $|\cdot|_N$ denotes the $N$-dimensional Lebesgue measure (or volume).
The formulas obtained in this paper will hold pointwise or distributionally (in the sense of Schwartz), depending on the growth properties of the corresponding relative fractal zeta function.
More specifically, the fractal tube formulas which we will establish will be written as sums of residues of the appropriate fractal zeta function evaluated at each (visible) complex dimension, and will either be exact or else, involve an error term; see Equation \eqref{ttt} below.

Furthermore, the relative distance and tube zeta functions which will appear in the fractal tube formulas obtained here are also introduced and extensively studied in [LapRa\v Zu1--7] 
and generalize the theory of complex dimensions and geometric zeta functions for fractal strings studied by the first author and M.\ van Frankenhuijsen in [Lap-vFr1--3], 
as well as in a significant number of research papers by numerous experts in fractal geometry and other areas; see the extensive list of relevant references provided in [Lap-vFr1--3] 
and \cite{fzf}, including [DemDenKo\"U, DemKo\"O\"U, DubSep, ElLapMacRo, EsLi1--2, Fal2, Fr, Gat, HamLap, HeLap, HerLap1--3, KeKom, Kom, KomPeWi, LalLap1--2, Lap1--8, LapL\'eRo, LapLu, LapLu-vFr1--2, LapMa1--2, LapPe1--3, LapPeWi1--2, LapPo1--3, LapRa\v Zu2--8, LapRo, LapRo\v Zu, L\'eMen, MorSep, MorSepVi, Ra1--2, RatWi, Tep1--2, \v Zu].

Moreover, we expect to use the results of this paper in order to establish a connection with earlier tube formulas and their potential interpretation in terms of curvatures or curvature measures in Federer's sense (see \cite{Fed}) associated with integer dimensions.
Namely, in \cite{Fed}, H.\ Federer has unified into a single framework (that of the `sets of positive reach') the Steiner tube formula for compact convex sets in $\eR^N$ and its generalization by Weyl \cite{Wey3} for smooth compact submanifolds of Euclidean spaces (described in \cite{BergGos}, and in the even more general setting of Riemannian manifolds, in \cite{Gra}).
See also the book \cite{Schn} and the articles [Schn1, Z\"a1--3], 
along with \cite{HuLaWe} for a generalization to the case of compact sets in $\eR^N$ and [Wi, WiZ\"a, Z\"a4--5, KeKom] 
for the case of certain self-similar sets and their self-conformal and random generalizations.

As was already mentioned, the first fractal tube formulas were obtained in the books [Lap-vFr1--3] 
in the case of fractal strings via fractal zeta functions (called {\em geometric zeta functions});
see, in particular, \cite[Chapter 8]{lapidusfrank12}.
Furthermore, these fractal tube formulas were generalized to a class of (higher-dimensional) fractal sprays in [LapPe1--2] 
and then in [LapPeWi1--2], 
via {\em tubular zeta functions} for fractal sprays and self-similar tilings which are significantly different from the fractal zeta functions considered in this paper (see \cite[Section 13.1]{lapidusfrank12} for a survey of these results).
Directly inspired by the just mentioned earlier work on fractal tube formulas, further developments were obtained in \cite{DDeKU,DemKoOU,DeKoOzUr}.
Finally, we point out that our general fractal tube formulas for relative fractal drums obtained here can be used to recover the previously obtained results for fractal strings and self-similar sprays, as we show in Subsections \ref{subsec_frstr} and \ref{subsec_self_similar_sp}, respectively.

For further results concerning tube formulas and their various generalizations in a variety of settings, as well as related topics, we mention, in particular, 
[Bla, CheeM\"uSchr1--2, Fu1--2, KlRot, Kow, LapLu3, LapLu--vFr1--2, LapPe3, LapRa\v Zu7--8, Mil, Mink, Mit\v Zu, Ol1--2, Sta, Stein], 
along with the many relevant references therein.

By analogy with the case of fractal strings, the complex dimensions of a relative fractal drum are defined as the poles, or more generally, singularities of its corresponding distance or tube zeta function.
Under mild hypotheses, the Minkowski (or box) dimension of a given relative fractal drum (RFD) will be its unique {\em real} complex dimension with maximal real part.
It follows directly from the definitions that in the case of a Minkowski measurable bounded subset $A$ of $\eR^N$, its Minkowski dimension appears as the co-exponent of the (monotonic) leading term in the asymptotic expansion of its tube function $t\mapsto|A_t|_N$ as $t\to 0^+$.

We will show here that for a given relative fractal drum $(A,\O)$ and under appropriate assumptions, the other (visible) complex dimensions will also appear as (complex) co-exponents, either in the leading term or in the higher order terms, in the asymptotic expansion of the tube function $t\mapsto|A_t\cap\O|_N$ as $t\to 0^+$.
Consequently, if $\omega$ is a nonreal complex dimension of a given relative fractal drum $(A,\O)$ of $\eR^N$, then the corresponding term $t^{N-\omega}$ appearing in the fractal tube formula of the RFD will generate oscillations of order $t^{N-\re\omega}$, to which we refer to as being the {\em inner geometric oscillations} of $(A,\O)$ {\em $($of order $t^{N-\re\omega}$$)$}.
In other words, the amplitude (resp., frequency) of the associated `geometric wave' is determined by the real (resp., imaginary) part of the complex dimension $\omega$.

In particular, for a given relative fractal drum $(A,\O)$, under suitable growth hypotheses on the fractal zeta function of the RFD and the assumption that all of the (visible) complex dimensions (i.e., the poles of its distance zeta function denoted by $\zeta_{A,\O}$, see Definition \ref{zeta_r}) are simple, the fractal tube formula takes the following form:
\begin{equation}\label{ttt}
|A_t\cap\O|_N=\sum_{\omega\in\po({\zeta}_{A,\O},\bm{W})}\frac{t^{N-\omega}}{N-\omega}\res\left({\zeta}_{A,\O},\omega\right)+R^{[0]}_{A,\O}(t).
\end{equation}
Here, $\po({\zeta}_{A,\O},\bm{W})$ denotes the set of {\em visible complex dimensions} through the {\em window} $\bm W$; that is, the poles of $\zeta_{A,\O}$ which are contained in the window $\bm{W}\subseteq\Ce$ (see Definitions \ref{comp_dim_def} and \ref{window_def_5}).
Furthermore, $R^{[0]}_{A,\O}(t)$ is an error term that corresponds to the terms of orders higher than those appearing in the sum in Equation \eqref{ttt}.
Roughly speaking, its estimate is directly connected to the `size' of the window $\bm W$; i.e., to how far to the left of the {\em critical line} (see Definition \ref{window_def_5}) the distance zeta function $\zeta_{A,\O}$ can be meromorphically extended.
The fractal tube formula \eqref{ttt} should be understood pointwise or distributionally, depending on the growth properties of $\zeta_{A,\O}$.
Moreover, if $\zeta_{A,\O}$ can be meromorphically extended to all of $\Ce$, then, under appropriate assumptions, the error term $R^{[0]}_{A,\O}(t)$ disappears; i.e., it is identically equal to zero and therefore, the corresponding fractal tube formula is said to be {\em exact}.

\medskip

In order to illustrate the results of this paper, we now give a sketch of two examples of (fractal) sets and their associated fractal tube formulas.

Let $A$ be the Sierpi\'nski gasket in $\eR^2$, constructed inside an equilateral triangle of side length equal to 1.
Then, its distance zeta function $\zeta_A$ is meromorphic on all of $\Ce$ and given by
\begin{equation}\label{gask_s}
\zeta_A(s)=\frac{6(\sqrt3)^{1-s}2^{-s}}{s(s-1)(2^s-3)}+\frac{2\pi}{s}+\frac{3}{s-1},
\end{equation}
for all $s\in\Ce$.
(See Example \ref{gsk_fract} for details and note that in \eqref{gask_s} above, we have chosen $\d=1$, without loss of generality.)
Clearly, it then follows that the set of complex dimensions of $A$ is given by
\begin{equation}
\po({\zeta}_{A}):=\po(\zeta_A,\Ce)=\{0,1\}\cup\left(\log_23+\frac{2\pi}{\log 2}\I\Ze\right),
\end{equation}
and that all of the complex dimensions are simple (i.e., are simple poles of $\zeta_A$).
The distance zeta function of $A$ is shown to satisfy the appropriate conditions and we can therefore deduce the exact fractal tube formula of $A$ from Equation \eqref{ttt}.
More specifically, we obtain that
\begin{equation}\nonumber
\begin{aligned}
|A_t|&=\sum_{\omega\in\po({\zeta}_{A})}\frac{t^{2-\omega}}{2-\omega}\res\left({\zeta}_A,\omega\right)\\
&=t^{2-\log_23}\,\frac{6\sqrt{3}}{\log 2}\sum_{k=-\ty}^{+\ty}\frac{(4\sqrt{3})^{-\omega_k}t^{-\I k\mathbf{p}}}{(2-\omega_k)(\omega_k-1)\omega_k}+\left(\frac{3\sqrt{3}}{2}+\pi\right)t^2,\\
\end{aligned}
\end{equation}
valid pointwise for all $t\in(0,1/2\sqrt{3})$ and where we have let $\mathbf{p}:=\frac{2\pi}{\log 2}$ and $\omega_k:=\log_23+\I k\mathbf{p}$ for every $k\in\Ze$.
Of course, the above formula coincides with the one obtained earlier in~\cite{lappe2} and \cite{lappewi1}.

Next, we discuss the example of the fractal nest $(A_a,\O)$ in $\eR^2$ generated by the $a$-string, which is given in full detail in Example \ref{ex_nest}.
Here, $a>0$ is a parameter, $A_a$ is a union of concentric circles centered at the origin of radii $j^{-a}$ for every integer $j\geq 1$ and $\O$ is the unit ball in $\eR^2$; see Figure \ref{nest_center} in that example.
In Subsection \ref{subsec_nestch}, we show that the distance zeta function $\zeta_{A_a,\O}$ of the associated RFD has a meromorphic continuation to all of $\Ce$ and that the set of complex dimensions of $(A_a,\O)$ satisfies the inclusion
\begin{equation}
\begin{aligned}
\po({\zeta}_{A_a,\O})&:=\po({\zeta}_{A_a,\O},\Ce)\subseteq\left\{1,\frac{2}{a+1},\frac{1}{a+1}\right\}\cup\left\{-\frac{m}{a+1}:m\in\eN\right\}.
\end{aligned}
\end{equation}
We do not have an equality above since some of the complex dimensions above may vanish, due to zero-pole cancellations.
Furthermore, if $a\neq1$, all of the above (potential) complex dimensions are simple and if $a=1$, the complex dimension $\omega=1$ has multiplicity 2.
On the other hand, we do know that $1$ and $D:=\frac{2}{a+1}$ are never canceled; that is, they always belong to $\po({\zeta}_{A_a,\O})$.
In Subsection \ref{subsec_nestch}, we also show that the distance zeta function satisfies growth conditions which are good enough so as to enable us to obtain the following pointwise fractal tube formula with error term for $(A_a,\O)$, provided $a\neq 1$:
\begin{equation}\label{nest_tub_f}
\begin{aligned}
|(A_a)_t\cap\O|&=\sum_{\omega\in\po({\zeta}_{A,\O},\bm{W})}\frac{t^{N-\omega}}{N-\omega}\res\left({\zeta}_{A,\O},\omega\right)+O(t^{2-\sigma})\\
&=\frac{t^{2-D}}{2-D}\res\left({\zeta}_{A,\O},D\right)+\frac{t^{2-1}}{2-1}\res\left({\zeta}_{A,\O},1\right)+O(t^{2-\sigma})\\
&=\frac{2^{2-D}D\pi}{(2-D)(D-1)}a^{D-1}t^{2-D}+\big(4\pi\zeta(a)-2\pi\big)t+O(t^{2-\sigma}),
\end{aligned}
\end{equation}
as $t\to 0^+$.
Here, $\bm{W}:=\{\re s>\sigma\}$ with $\sigma\in(-1/2(a+1),0)$ arbitrary and $\zeta$ is the Riemann zeta function.
Observe that if $a\in(0,1)$, then $D>1$ and the leading power of $t$ in \eqref{nest_tub_f} above is $2-D$.
In other words, $\dim_B(A_a,\O)=D$ and $(A_a,\O)$ is Minkowski measurable with Minkowski content (see Equation \eqref{minkrel} and the text surrounding it)
$$
\mathcal{M}^D(A,\O)=\frac{2^{2-D}D\pi}{(2-D)(D-1)}a^{D-1}.
$$
On the other hand, if $a>1$, then $D<1$ and the leading power of $t$ in \eqref{nest_tub_f} above is $2-1=1$.
Therefore, we conclude that $\dim_B(A_a,\O)=1$ and $(A_a,\O)$ is Minkowski measurable with Minkowski content
$$
\mathcal{M}^1(A,\O)=\big(4\pi\zeta(a)-2\pi\big).
$$
Finally, when $a=1$, the two distinct complex dimensions $D$ and $1$ `merge' into a single complex dimension of order $2$.
In this case, we also obtain the corresponding pointwise fractal tube formula with error term but we have to use the more general fractal tube formula which is valid in the presence of complex dimensions of higher order; i.e., poles of the associated distance fractal zeta function of higher multiplicities.
In that case, formula \eqref{ttt} must be replaced by the following more general formula:
\begin{equation}\label{tth}
|A_t\cap\O|_N=\sum_{\omega\in\po({\zeta}_{A,\O},\bm{W})}\res\left(\frac{t^{N-s}}{N-s}{\zeta}_{A,\O}(s),\omega\right)+R^{[0]}_{A,\O}(t).
\end{equation}
By choosing a window $\bm{W}:=\{\re s>\sigma\}$, with $\sigma\in(-3/4,-1/2)$, we now obtain the following pointwise fractal tube formula with error term for $(A_1,\O)$:
\begin{equation}
\begin{aligned}
|(A_1)_t\cap\O|&=\sum_{\omega\in\po({\zeta}_{A,\O},\bm{W})}\res\left(\frac{t^{N-s}}{N-s}{\zeta}_{A,\O}(s),\omega\right)+O(t^{2-\sigma})\\
&=\res\left(\frac{t^{2-s}}{2-s}\zeta_{A_1,\O}(s),1\right)+\frac{2}{3}{\res\left(\zeta_{A_1,\O},\frac{1}{2}\right)t^{\frac{3}{2}}}\\
&\phantom{=}+\frac{2}{5}{\res\left(\zeta_{A_1,\O},-\frac{1}{2}\right)t^{\frac{5}{2}}}+O(t^{2-\sigma})\ \textrm{as}\ t\to0^+.\\
\end{aligned}
\end{equation}
In order to calculate the residue at $1$, we expand the function $t^{2-s}/(2-s)$ into a Taylor series around $s=1$ and we multiply this with the Laurent expansion of $\zeta_{A_1,\O}$ around $s=1$, which then yields
\begin{equation}
\res\left(\frac{t^{2-s}}{2-s}\zeta_{A_1,\O}(s),1\right)=2\pi t\log t^{-1}+\mathrm{const}\cdot t;
\end{equation}
so that (still pointwise)
\begin{equation}
\begin{aligned}
|(A_1)_t\cap\O|&=2\pi t\log t^{-1}+\mathrm{const}\cdot t+o(t)\quad \textrm{as}\ t\to0^+.
\end{aligned}
\end{equation}
We point out that the above tube formula is in agreement with the fact $\dim_B(A_1,\O)=1$ and that $(A_1,\O)$ is Minkowski degenerate, i.e., $\mathcal{M}^1(A_1,\O)=+\ty$.

In general, and under suitable assumptions, a complex dimension $\omega$ of an RFD $(A,\O)$ which is of order $m\geq 1$ will generate terms of the type $t^{N-\omega}(\log t^{-1})^{k-1}$, for $k=1,2,\ldots,m$, in the corresponding fractal tube formula.
We note that RFDs having (even) principal complex dimensions of arbitrary orders exist and are relatively easy to construct, as was done in \cite[Section 4.4]{refds} and also in \cite[Subsection 4.2.2]{fzf}.
Furthermore, also in the just mentioned references, RFDs with principal complex dimensions of infinite order (i.e., with principal complex dimensions that are essential singularities of the associated fractal zeta function) have been constructed; see \cite[Section 4.4]{refds} and \cite[Subsection 4.2.2]{fzf}.
We stress that the theory of the present paper can also be applied if we allow complex dimensions of infinite order.
Indeed, all of the statements and proofs of the relevant theorems are also valid almost verbatim\footnote{One needs to appropriately replace, for instance, the phrase ``meromorphic extension'' by ``meromorphic extension with possible isolated essential singularities'', etc.} if we allow complex dimensions to be also essential singularities (alongside poles) of the associated fractal zeta functions.

\medskip

In our forthcoming paper \cite{ftf_b}, we will apply the results of this paper in order to obtain a Minkowski measurability criterion for relative fractal drums formulated in terms of the nonexistence of nonreal complex dimensions with maximal real part; this criterion generalizes the corresponding Minkowski measurability criterion for fractal strings obtained in [Lap-vFr1--3]. 
(See \cite[Section 8.3]{lapidusfrank12})
More precisely, under appropriate hypotheses, we will show in \cite{ftf_b} that the relative fractal drum is Minkowski measurable if and only if the only complex dimension with maximal real part is the Minkowski dimension itself, and it is simple.
The distributional fractal tube formulas obtained in the present paper, alongside a Tauberian theorem due to Wiener and Pitt, will play a crucial role in establishing the aforementioned criterion.

We point out that in this paper, we work with four kinds of fractal zeta functions for relative fractal drums, including, the already mentioned {\em distance} and {\em tube} zeta functions (see Definitions \ref{zeta_r} and \ref{tube_zeta_deff}, respectively).
These latter two zeta functions are connected by a relatively simple functional equation (see Equation \eqref{equ_tilde}), which implies that for a given relative fractal drum $(A,\O)$ in $\eR^N$, they generate the same complex dimensions, provided the upper Minkowski dimension of $(A,\O)$ is strictly less than $N$.
Of these two fractal zeta functions, the tube zeta function has a more theoretical value, whereas the distance zeta function is more practical since it is often easier to compute in concrete examples.
In light of this, we first obtain the fractal tube formulas expressed in terms of the tube zeta function and then translate them in terms of the distance zeta function by introducing a new (intermediate) fractal zeta function, called the {\em shell zeta function} (see Definition \ref{shell_defn}). 
The reason for introducing this new fractal zeta function is of a technical nature since the shell zeta function satisfies an even more direct functional equation connecting it with the distance zeta function; see Equation \eqref{shell_dist}.

Finally, the last zeta function we introduce in this work is the {\em relative Mellin zeta function}, in Subsection \ref{subsec_mellin}.
The reason for introducing this zeta function is also of a technical nature since it will be crucial in order to extend our distributional tube formulas to a larger class of test functions.
This greater generality will be needed in our aforementioned forthcoming paper \cite{ftf_b}, where we use, in particular, our distributional tube formula to derive a Minkowski measurability criterion for RFDs.  

The results of this paper justify in a natural way the notion of complex dimensions and enable us to propose a new definition of fractality which (roughly) states that a relative fractal drum (or a bounded subset) of $\eR^N$ is considered to be `fractal' if it possesses a nonreal complex dimension.
This definition of fractality was already given in [Lap-vFr1--3] 
(see, e.g., \cite[Sections 12.1 and 12.2]{lapidusfrank12}) but we now have to our disposal a general theory of fractal zeta functions and of associated fractal tube formulas valid in any dimension $N$ (with $N\geq 1$) for any bounded subset (and relative fractal drum) of $\eR^N$.
We will demonstrate how this proposed definition `recognizes' the fractality of a number of subsets which would not be fractal in the classical sense but which everyone `feels' that they should nevertheless be considered fractal, just by looking at them.
For instance, this is the case with the relative fractal drum generated by the `devil's staircase', i.e., the Cantor function graph (see Example \ref{stair}), as well as with the examples of the $1/2$-square and the $1/3$-square fractals (see Examples \ref{1/2-tube_formula} and \ref{1/3-tube_formula}).

Finally, we refer the interested reader to our monograph \cite{fzf} for a complete and detailed exposition of the higher-dimensional theory of complex dimensions and fractal zeta functions.

\medskip

In closing this introduction, it may be helpful to the readers to point out the relationship between the results of our present work and the classic Steiner tube formula \cite{Stein}, as generalized in various ways by many authors (including Minkowski \cite{minkow}, Weyl \cite{Wey3} and later, Federer [Fed1--2]) 
and as stated in the original case of compact convex sets in \cite[Theorem 4.2.1]{Schn}.\footnote{Our exposition of this material closely follows part of \cite[Subsection 13.1.3]{lapidusfrank12}; see also [LapPe2--3] and \cite{lappewi1}.}

Let $A$ be a compact convex subset of $\eR^N$ $(N\geq 1)$ and let $B^k$ denote the $k$-dimensional unit ball of $\eR^k$ (for any integer $k\geq 0$), with $k$-dimensional volume (or Lebesgue measure) denoted by $|B^k|_k$.
(For $k=0$, we let $|B^0|_0:=1$.)
Note that for $t\geq 0$, the $t$-neighborhood (or $t$-parallel body) of $A$ can be written as $A_t=A+tB^N$.
Then its volume $V_A(t):=|A_t|_N$ can be expressed as a polynomial of degree $\leq N$ (exactly $N$ if $|A|_N>0$, e.g., if $A$ has nonempty interior) in the variable $t$:
\begin{equation}\label{xxx}
V_A(t)=\sum_{k=0}^N\mu_k(A)|B^{N-k}|_{N-k}t^{N-k},
\end{equation}
where for $k=0,1,\ldots,N$, $\mu_k(A)$ denotes the $k$-th {\em intrinsic volume} of $A$.

Up to some suitable normalizing and multiplicative constant (depending on $k$), for each $k\in\{0,1,\ldots,N\}$, the $k$-th intrinsic volume $\mu_k(A)$ coincides with the $k$-th {\em total curvature} of $A$ or the so-called $(N-k)$-th {\em Quermassintegral} of $A$.
Moreover, still for $k\in\{0,1,\ldots,N\}$, $\mu_k(A)$ can be interpreted either combinatorially and algebraically in terms of appropriate valuations (see \cite{KlRot}) or (in a closely related context) within the framework of integral geometry, as the average measure of orthogonal projections to $(N-k)$-dimensional subspaces of Euclidean space $\eR^N$; see, e.g., \cite{Schn} and \cite[Chapter 7]{KlRot}.
(This latter interpretation was already implicit in Steiner's original work \cite{Stein} and that of his immediate successors, where $N=2$ or 3.)

To make a long and beautiful story short, let us simply mention here that (up to a suitable normalizing multiplicative constant) $\mu_0$ corresponds to the {\em Euler characteristic},\footnote{In the present case of compact convex sets, $\mu_0$ is always identically equal to one. However, in the more general setting of sets of positive reach or of finite unions of such sets, it is $\Ze$-valued; see, e.g., \cite[Section 3.4]{Schn} and \cite{Fed,Za2}.} $\mu_1$ to the so-called {\em mean width}, $\mu_{N-1}$ the {\em surface area}, and $\mu_N$ the $N$-dimensional {\em volume} of $A$ (i.e., $\mu_N(A)=|A_N|=|A|$, in our notation).

Finally, let us point out that the intrinsic volumes $\mu_k$ have the following algebraic and geometric properties (for every $k=0,1,\ldots,N$):

\medskip

$(i)$ Each $\mu_k$ is homogeneous of degree $k$, i.e., for all $\lambda>0$,
\begin{equation}\label{x}
\mu_k(\lambda A)=\lambda^k\mu_k(A),
\end{equation}
and

\medskip

$(ii)$ each $\mu_k$ is rigid motion invariant; more specifically, for any (affine) isometry $R$ of $\eR^N$, we have that
\begin{equation}\label{xx}
\mu_k(R(A))=\mu_k(A).
\end{equation}

Remarkably, for any (visible) complex dimension $\omega$ of a bounded subset $A$ of $\eR^N$ (or, more generally, of an RFD $(A,\O)$ of $\eR^N$), the corresponding coefficient of our fractal tube formula (assuming that we are in the case of simple poles), that is, essentially, the residue of the fractal zeta function at $s=\omega$ (see Equation \eqref{ttt} above), satisfies entirely analogous homogeneous and geometric invariance properties (with $k$ replaced by $\omega$ in Equations \eqref{x} and \eqref{xx}).\footnote{The analog of Equation \eqref{xx} follows easily from the definitions, while the counterpart of Equation \eqref{x} follows from the scaling property of the fractal zeta function and hence of its residues (see \cite[Section 2.2]{mefzf} or \cite[Theorem 4.1.38]{fzf}).}
Furthermore, of course, the resulting fractal tube formula is no longer a polynomial of degree $N$ in the variable $t$ but involves a typically infinite sum ranging over all of the underlying visible complex dimensions of $A$ (or of the RFD $(A,\O)$).
Moreover, as we shall see in many examples, the coefficients of the fractal tube formula that correspond to the set of (visible) complex dimensions can frequently be naturally decomposed as a set of {\em integer dimensions} (say, $\omega=k\in\{0,1,\ldots,N\}$) and of {\em scaling dimensions} (say, $\omega\in \mathcal{D}_{\mathfrak{S}}$).\footnote{Of course, if $\mathcal{D}_{\mathfrak{S}}$ happens to be empty (which is certainly the case if $A$ is a compact convex set), then $V_A(t)$ reduces to a polynomial expression of degree $\leq N$ in $t$ and the corresponding tube formula is Steiner-like, much as in Equation \eqref{xxx} above.}
(See, especially, the discussion of the Sierpi\'nski gasket and of the $3$-carpet in Section \ref{subsec_sier}, along with that of self-similar sprays in Section \ref{subsec_self_similar_sp}; such a situation already arose in the very special but important case of fractal sprays studied in [LapPe1--2] and [LapPeWi1--2].)

We leave to a later work a further and much more detailed exploration of the possible geometric, algebraic and combinatorial interpretations of our fractal tube formulas (as well as potential local versions thereof), in the spirit of the above discussion and particularly, the work of Stein \cite{Stein}, Minkowski \cite{minkow} (see also \cite{Schn}), Weyl \cite{Wey3} (see also \cite{BergGos} and \cite{Gra}), Federer (\cite{federer} and, especially in \cite{Fed}, his work on {\em local tube formulas} and {\em curvature measures}), Klain and Rota \cite{KlRot}, and many other authors; see, e.g., the books \cite{Bla}, \cite{Schn}, \cite{Gra}, [Lap-vFr1--3], along with the articles [Fu1--2], \cite{HuLaWe}, \cite{kombrink}, \cite{Kom}, \cite{Kow}, [LapPe1--3], [LapPeWi1--2], \cite{Mil}, [Ol1--2], \cite{winter}, \cite{Schn1}, \cite{stacho}, \cite{Wi}, \cite{WiZa}, [Z\"a1--5], and the many relevant references therein.

\medskip

The rest of this paper is organized as follows: In Section \ref{prelm}, we provide some basic definitions (concerning the Minkowski dimension and the distance zeta functions of RFDs) and technical preliminaries about the Mellin transform.
In Section \ref{sec_point}, we establish the pointwise fractal tube formula, with or without error term, and expressed in terms of the relative tube zeta function.
In Section \ref{sec_distr}, we then use this pointwise tube formula in order to derive the distributional fractal tube formula, with or without error term and still expressed in terms of the tube zeta function.
In Section \ref{distance_tube}, we establish the pointwise and distributional fractal tube formulas, with or without error term, but now formulated via the (relative) distance zeta function.
In the process, we introduce the notion of shell zeta function (as well as that of Mellin zeta function) which enables us, in particular, to use the results of Sections \ref{sec_point} and \ref{sec_distr} formulated via the tube zeta function.
Finally, in Section \ref{exp_app}, we illustrate our results by obtaining fractal tube formulas in a variety of concrete examples, including fractal strings (Subsection \ref{subsec_frstr}), the Sierpi\'nski gasket and the 3-dimensional carpet (Subsection \ref{subsec_sier}), the Cantor graph RFD (Subsection \ref{subsec_devil}), fractal nests and (unbounded) geometric chirps (Subsection \ref{subsec_nestch}), as well as fractal sprays, and more specifically, self-similar sprays (Subsection \ref{subsec_self_similar_sp}).

\section{Preliminaries}\label{prelm}

We begin this section by stating some definitions and results from [LapRa\v Zu2--5] 
(see also the research monograph~\cite{fzf})
which will be needed here,
such as the definition of a relative fractal drum in $\eR^N$ and its associated relative distance and tube zeta functions.
In order to exclude dealing with trivial cases and shorten the statements of the results, we will always assume throughout this paper that all the sets $A$ and $\O$ are nonempty.

First of all, given a subset $A$ of $\eR^N$, we denote its {\em $\delta$-neighborhood} (or {\em $\d$-parallel set}) by
\begin{equation} 
A_\d:=\{x\in\eR^N:d(x,A)<\d\}.
\end{equation}
Here, $d(x,A):=\inf\{|x-y|:y\in A\}$ is the Euclidean distance between the point $x$ and the set $A\subseteq\eR^N$.

\begin{definition}[{[LapRa\v Zu1,4]}]\label{zeta_r}\label{drum}
Let $\Omega$ be a Lebesgue measurable subset of $\eR^N$, not necessarily bounded, but of finite $N$-dimensional Lebesgue measure (or ``volume'').
Furthermore, let $A\subseteq\eR^N$, also possibly unbounded, be such that $\Omega$ is contained in $A_\delta$ for some $\delta>0$.
The {\em distance zeta function $\zeta_{A,\O}$ of $A$ relative to $\Omega$} (or the {\em relative distance 
zeta function})
 is defined by the following Lebesgue integral:
\begin{equation}\label{rel_dist_zeta}
\zeta_{A,\O}(s):=\int_{\Omega} d(x,A)^{s-N}\D x,
\end{equation}
for all $s\in\Ce$ with $\re s$ sufficiently large.
The ordered pair $(A,\Omega)$, appearing in Definition~\ref{zeta_r} is called   
a {\em relative fractal drum} or RFD in short. 
In light of this, we will also use the phrase {\em zeta functions of relative fractal drums} instead of relative zeta functions.
\end{definition}

\medskip

\begin{remark}\label{holo_diff}
If we replace the domain of integration $\O$ in Equation \eqref{rel_dist_zeta} with $A_\d\cap\O$ for some fixed $\d>0$, that is, if we let
\begin{equation}\label{rel_dist_zeta_d}
\zeta_{A,\O}(s;\d):=\int_{A_\d\cap\Omega} d(x,A)^{s-N}\D x,
\end{equation}
then the difference $\zeta_{A,\O}(s)-\zeta_{A,\O}(s;\d)$ is an entire function (see~[LapRa\v Zu1--5]).
Therefore, we can alternatively define the relative distance function of $(A,\O)$ by~\eqref{rel_dist_zeta_d},
since in the theory of complex dimensions, we are mostly interested in the poles (or, more generally, in the singularities) of meromorphic extensions of (various) fractal zeta functions.
Then, in light of the principle of analytic continuation, the dependence of $\zeta_{A,\O}(\,\cdot\,;\d)$ on $\d$ is inessential from the point of view of the complex dimensions (defined in Definition \ref{comp_dim_def} just below).
  
The condition that $\O\subseteq A_\d$ for some $\d>0$ is of a technical nature and ensures that the function $x\mapsto d(x,A)$ is bounded for $x\in\O$.
If $\O$ does not satisfy this condition, we can still use the alternative definition given by Equation~\eqref{rel_dist_zeta_d}.\footnote{Since then, $\O\setminus A_\d$ and $A$ are a positive distance apart, this replacement will not affect the relative box dimension of $(A,\O)$ introduced just below or any other fractal properties of $(A,\O)$ that will be introduced later on.}

\end{remark}

\begin{remark}
As was already stated in the introduction, the notion of a relative fractal drum generalizes the notion of a bounded subset of $\eR^N$.
Namely, in order to apply the results of this paper to an arbitrary bounded subset $A$ of $\eR^N$, one chooses any bounded open set $\O$ containing $A_\delta$ for some $\delta>0$ (for instance, $A_\d$ itself) and applies the theory to the RFD $(A,\O)$.
\end{remark}

Analogous comments also hold for the relative tube zeta function, which we now introduce.

\begin{definition}[{[LapRa\v Zu1,4]}]\label{tube_zeta_deff}
Let $(A,\O)$ be an RFD in $\eR^N$ and fix $\d>0$. We define the {\em tube zeta function} $\widetilde{\zeta}_{A,\O}(s;\d)$ {\em of $A$ relative to $\O$} (or the {\em relative tube zeta function} by
\begin{equation}\label{401/2}
\widetilde{\zeta}_{A,\O}(s;\d):=\int_0^{\delta}t^{s-N-1}|A_t\cap\O|\,\D t,
\end{equation}
for all $s\in\Ce$ with $\re s$ sufficiently large, where the integral is taken in the Lebesgue sense and $|A_t\cap\O|:=|A_t\cap\O|_N$ denotes the $N$-dimensional volume of $A_t\cap\O\subseteq\eR^N$.
\end{definition} 

The distance and tube zeta functions of relative fractal drums are a special case of Dirichlet-type integrals (or, in short, DTIs), and as such, have a well-defined {\em abscissa of $($absolute$)$ convergence}.
The abscissa of convergence of a DTI $\zeta\colon E\to\Ce$, where $E\subseteq\Ce$ is a domain, is defined as the infimum of all the real numbers $\a$ for which the integral $\zeta(\a)$ is absolutely convergent and we denote it by $D(\zeta)$.\footnote{For a precise definition of a DTI, as well as for the results mentioned here concerning them (and their generalizations), we refer the interested reader to~\cite[Appendix A]{dtzf} and for more details, to \cite[Appendix~A]{fzf}.} 

In short, a DTI is given by
\begin{equation}\label{2.3.1/2}
\zeta_{E,\varphi,\nu}(s):=\int_{E}\varphi(x)^s\di\nu(x),
\end{equation}
for all $s\in\Ce$ with $\re s$ sufficiently large, where $E$ is a locally compact and metrizable topological space (e.g., $E:=\O$, $E:=A_\delta\cap\O$ or $E:=[0,\d]$, in Equation \eqref{rel_dist_zeta}, \eqref{rel_dist_zeta_d} or \eqref{401/2}, respectively), $\nu$ is a (positive or complex) local measure with total variation measure denoted by $|\nu|$, and $\varphi\colon E\to\eR$ satisfies $\varphi\geq 0$ $|\nu|$-a.e.\ on $E$ and is tamed (i.e., there exists $C<\ty$ such that $\varphi\leq C$ $|\nu|$-a.e.\ on $E$).

A general result about a DTI $\zeta$ is the fact that it is a holomorphic function in the open half-plane to the right of its abscissa of convergence; that is, on the {\em half-plane of $($absolute$)$ convergence} $\Pi(\zeta):=\{\re s>D(\zeta)\}$.\footnote{Here and thereafter, subsets of $\Ce$ of the type $\{s\in\Ce:\re s< \a\}$, $\{s\in\Ce:\re s> \a\}$ and $\{s\in\Ce:\re s= \a\}$ are denoted by $\{\re s<\a\}$, $\{\re s>\a\}$ and $\{\re s=\a\}$, respectively.}
In the sequel, the vertical line $\{\re s=D(\zeta)\}$ is often referred to as the {\em critical line} (for $\zeta$).
Furthermore, the relative distance and tube zeta functions are connected by the functional equation
\begin{equation}\label{equ_tilde}
\zeta_{A,\O}(s;\d)=\delta^{s-N}|A_\delta\cap\O|+(N-s)\widetilde\zeta_{A,\O}(s;\d),
\end{equation}
which is valid on any open connected subset $U$ of $\Ce$ to which any of these two zeta functions has a meromorphic continuation (see~\cite{refds} or \cite{fzf}).
This result is very useful since in many concrete examples, the distance zeta function is much easier to calculate than the tube zeta function.
On the other hand, the tube zeta function has an important theoretical value and many results in [LapRa\v Zu1--5] 
are proven in terms of the tube zeta function and then reformulated in terms of the distance zeta function.
This will also be the case in the present paper.

A key technical observation underlying some of the methods used in this paper is that the tube zeta function coincides with the Mellin transform of a modified tube function $t\mapsto|A_t\cap\O|$.
More specifically, as we will see in a moment, one has that for all $s\in\Ce$ with $\re s$ sufficiently large,
\begin{equation}\label{mellin_tube}
\begin{aligned}
\widetilde{\zeta}_{A,\O}(s;\d)&=\int_0^{+\ty}t^{s-1}\left(\chi_{(0,\d)}(t)t^{-N}|A_t\cap\O|\right)\D t,
\end{aligned}
\end{equation}
where $\chi_{(0,\d)}$ denotes the characteristic function of the set $(0,\d)$.
Recall that the Mellin transform of a function $f\colon\eR\to\eR$ is defined by
\begin{equation}\label{mell_trans_def}
\{\mathfrak{M}f\}(s):=\int_0^{+\ty}t^{s-1}f(t)\di t,
\end{equation}
where $s$ is a complex number with large enough real part.
Then, by letting 
\begin{equation}
f_{\d}(t):=\chi_{(0,\d)}(t)t^{-N}|A_t\cap\O|,
\end{equation}
we have that
\begin{equation}
\widetilde{\zeta}_{A,\O}(s;\d)=\{\mathfrak{M}f_{\d}\}(s),
\end{equation}
for all $s\in\Ce$ with $\re s$ sufficiently large.
This will enable us (in Theorem \ref{tube_inversion} below) to recover the tube function $t\mapsto|A_t\cap\O|$ from the relative tube zeta function $\widetilde{\zeta}_{A,\O}$ by using the Mellin inversion theorem (recalled in Theorem \ref{mellin_inv}).
More interestingly, the functional equation \eqref{equ_tilde} will then enable us to use the distance zeta function ${\zeta}_{A,\O}$ instead of the tube zeta function $\widetilde{\zeta}_{A,\O}$ for operating this recovery.

\begin{remark}\label{2.4.1/2}
The important special case of a bounded set $A\subset\eR^N$ is obtained by considering the RFD $(A,A_{\delta})$ (i.e., by letting $\O=A_{\delta}$, for some $\delta>0$) in Equation \eqref{rel_dist_zeta} and Equation \eqref{401/2} in order to obtain the distance zeta function $\zeta_A$ and the tube zeta function $\widetilde{\zeta}_A$ of $A$, respectively.
(See \cite{dtzf} and \cite[Chapters 2 and 3]{fzf}.)
We note that the notion of distance zeta function $\zeta_A$ was first introduced by the first author in 2009.

An entirely analogous comment could be made for the (upper, lower) Minkowski dimension and (upper, lower) Minkowski content of a bounded subset $A$ of $\eR^N$.
Namely, in the discussion just below, it would suffice, for example, to consider the RFD $(A,A_\delta)$ in Equation \eqref{minkrel} or Equation \eqref{dimrel} in order to recover $\mathcal{M}^{*r}(A)$ or $\ov{\dim}_BA$, respectively.
\end{remark}

We now proceed by introducing the notions of Minkowski content and Minkowski (or box) dimension of a relative fractal drum (RFD) and relating them to the distance and tube zeta functions of this RFD.
For any {\em real} number $r$, we define the {\em upper $r$-dimensional Minkowski content of $A$ relative to $\Omega$}
(or {\em the upper relative Minkowski content}, or {\em the upper Minkowski content of the relative fractal drum $(A,\Omega)$}) by
\begin{equation}\label{minkrel}
{{\mathcal{M}}}^{*r}(A,\Omega):=\limsup_{t\to0^+}\frac{|A_t\cap\Omega|}{t^{N-r}}, 
\end{equation}
and we then proceed in the usual way in order to define $\ov{\dim}_B(A,\O)$:
\begin{equation}\label{dimrel}
\begin{aligned}
\ov\dim_B(A,\Omega)&=\inf\{r\in\eR:{{\mathcal{M}}}^{*r}(A,\Omega)=0\} \\
&=\sup\{r\in\eR:{{\mathcal{M}}}^{*r}(A,\Omega)=+\ty\}.
\end{aligned}
\end{equation}
We call it the {\em relative upper box dimension}
 $($\rm{or} {\em relative Minkowski dimension}$)$ of $A$ with respect to $\Omega$ (or else the {\em relative upper box dimension of $(A,\Omega)$}).
Note that $\ov\dim_B(A,\Omega)\in[-\ty,N]$.
We stress that the values of $\ov{\dim}_B(A,\O)$ can indeed be negative, even equal to $-\ty$; 
see \cite{refds} or \cite[Chapter 4]{fzf}.\footnote{However, in the important special case of a bounded set $A\subset\eR^N$ discussed in Remark \ref{2.4.1/2}, we always have that $\dim_BA\in[0,N]$; in particular, $\ov{\dim}_BA\geq 0$.}
Also note that for these definitions to make sense, it suffices that $|A_\d\cap\O|<\ty$ for some $\d>0.$

The value ${\mathcal{M}}_*^{r}(A,\Omega)$ of the {\em lower $r$-dimensional Minkowski content} of $(A,\Omega)$, is defined as in \eqref{minkrel}, except for a lower instead of an upper limit.
Analogously as in \eqref{dimrel}, we then define the {\em relative lower box $(${\rm or} Minkowski$)$ dimension} of $(A,\Omega)$ by using the lower $r$-dimensional Minkowski content of $(A,\Omega)$ instead of the upper.
Furthermore, in the case when $\underline\dim_B(A,\Omega)=\ov\dim_B(A,\Omega)$, we denote by
$
\dim_B(A,\Omega)
$ 
this common value and call it the {\em relative box $(${\rm or} Minkowski$)$ dimension\label{rel_box_dim}}.
Moreover, if $0<{\mathcal{M}}_*^D(A,\Omega)\le{\mathcal{M}}^{*D}(A,\Omega)<\ty$, we say that the relative 
fractal drum $(A,\O)$ is {\em Minkowski nondegenerate}.\label{nondeg_rel}
It then follows that $\dim_B(A,\O)$ exists and is equal to $D$.

Finally, if ${\mathcal{M}}_*^D(A,\Omega)={\mathcal{M}}^{*D}(A,\Omega)$, we denote this common value by $\mathcal{M}^D(A,\Omega)$ and call it the {\em relative 
Minkowski content} of $(A,\O)$.
If $\mathcal{M}^D(A,\Omega)$ exists and is different from $0$ and $\ty$ (in which case $\dim_B(A,\Omega)$ exists and then necessarily $D=\dim_B(A,\Omega)$), we say that the relative fractal drum $(A,\Omega)$ is {\em Minkowski 
measurable}.
Many examples and properties of the relative box dimension can be found in [Lap1--3], [LapPo1--3], \cite{lapidushe}, [Lap-vFr1--3], \cite{rae}, [LaPe2--3], [LapPeWi1--2], in various special cases, and in [LapRa\v Zu1--7] in the present general setting of RFDs.

In the following three theorems, we recall some basic results about zeta functions of relative fractal drums.
(See [LapRa\v Zu1--2] 
for the special case of bounded subsets of $\eR^N$.)

\begin{theorem}[{[LapRa\v Zu1,4]}]\label{an_rel} 
Let $(A,\O)$ be a relative fractal drum in $\eR^N$. Then the following properties hold$:$

\bigskip 

$(a)$ The relative distance zeta function $\zeta_{A,\O}$ is holomorphic in the half-plane $\{\re s>\overline{\dim}_B(A,\Omega)\}$.
More precisely,
\begin{equation}
D(\zeta_{A,\O})=\overline{\dim}_B(A,\Omega).
\end{equation}

\bigskip

$(b)$ If the relative box $($or Minkowski$)$ dimension $D:=\dim_B(A,\O)$ exists, $D<N$ and ${\M}_*^{D}(A,\O)>0$, then $\zeta_{A,\O}(s)\to+\ty$ as $s\in\eR$
converges to $D$ from the right.   
\end{theorem}

\begin{remark}\label{tube_holo}
For a general relative fractal drum $(A,\O)$ in $\eR^N$, the right half-plane $\{\re s>\overline{\dim}_B(A,\Omega)\}$ is not necessarily the maximal open right half-plane to which its relative distance zeta function can be holomorphically continued. 
For instance, this is the case with the line segment $I:=[0,1]\subset\eR$, understood as a relative fractal drum $(I,I_\d)$; see Subsection \ref{subsec_lisp}.\footnote{We would like to thank E.\ P.\ J.\ Pearse for this example.}
However, this situation cannot occur if $(A,\O)$ satisfies the hypotheses of part $(b)$ of Theorem~\ref{an_rel}. 


Furthermore, if $\ovb{\dim}_B(A,\O)<N$, in light of the functional equation~\eqref{equ_tilde}, Theorem~\ref{an_rel} is also valid if we replace the relative distance zeta function by the relative tube zeta function in its statement.
Moreover, it can be shown directly (i.e.,  without the use of the functional equation) that in the case of the tube zeta function, Theorem~\ref{an_rel} is also valid in the special case when $\ovb{\dim}_B(A,\O)=N$.
\end{remark}

\begin{theorem}[{[LapRa\v Zu1,4]}]\label{pole1} 
Assume that $(A,\O)$ is a nondegenerate RFD in $\eR^N$, 
that is, $0<{\M}_*^{D}(A,\O)\le{\M}^{*D}(A,\O)<\ty$ $($in particular, $\dim_B(A,\O)=D)$, 
and $D<N$.
If $\zeta_{A,\O}=\zeta_{A,\O}(\,\cdot\,,\d)$ has a meromorphic extension to a connected open neighborhood of $s=D$,
then $D$ is necessarily a simple pole of $\zeta_{A,\O}$, the residue $\res(\zeta_{A,\O},D)$ is independent of $\d$ and 
\begin{equation}\label{res}
(N-D){\M}_*^{D}(A,\O)\le\res(\zeta_{A,\O},D)\le(N-D){\M}^{*D}(A,\O).
\end{equation}
Furthermore, if $(A,\O)$ is Minkowski measurable, then 
\begin{equation}\label{pole1minkg1=}
\res(\zeta_{A,\O}, D)=(N-D)\M^D(A,\O).
\end{equation}
\end{theorem}

The above theorem can be reformulated in terms of the relative tube zeta function and in that case, we can remove the condition $\dim_B(A,\O)<N$.

\begin{theorem}[{[LapRa\v Zu1,4]}]\label{pole1mink_tilde} 
Assume that $(A,\O)$ is a nondegenerate RFD in $\eR^N$ $($so that $D:=\dim_B(A,\O)$ exists$)$, 
and that for some $\d>0$ there exists a meromorphic extension of $\widetilde\zeta_{A,\O}=\widetilde\zeta_{A,\O}(\,\cdot\,;\d)$ to a connected open neighborhood of $D$.
Then,  $D$ is a simple  pole of $\widetilde\zeta_{A,\O}$
and $\res(\widetilde{\zeta}_{A,\O},D)$ is independent of $\delta$.
Furthermore, we have
\begin{equation}\label{zeta_tilde_M}
{\M}_*^{D}(A,\O)\le\res(\widetilde\zeta_{A,\O}, D)\le {\M}^{*D}(A,\O).
\end{equation}
In particular, if $(A,\O)$ is Minkowski measurable, then 
\begin{equation}\label{zeta_tilde_Mm}
\res(\widetilde\zeta_{A,\O}, D)=\M^D(A,\O).
\end{equation}
\end{theorem}

We continue by stating some of the definitions that were already introduced in [Lap-vFr1--3] 
in the setting of generalized fractal strings and adapt them to the present context of relative fractal drums in $\eR^N$.
(See, e.g., \cite[Chapter 5]{lapidusfrank12}.) 

\begin{definition}\label{window_def_5}
The {\em screen} $\bm S$ is the graph of a bounded, real-valued Lipschitz continuous function $S(\tau)$, with the horizontal and vertical axes interchanged:
\begin{equation}\label{g_screen}
\bm{S}:=\{S(\tau)+\I \tau\,:\,\tau\in\eR\}.
\end{equation}
The Lipschitz constant of $\bm S$ is denoted by $\|S\|_{\mathrm{Lip}}$; so that
$$
|S(x)-S(y)|\leq\|S\|_{\mathrm{Lip}}|x-y|,\quad\textrm{ for all }x,y,\in\eR.
$$
Furthermore, the following quantities are associated to the screen:
\begin{equation}\nonumber
\inf S:=\inf_{\tau\in\eR}S(\tau)\quad\textrm{ and }\quad\sup S:=\sup_{\tau\in\eR}S(\tau).
\end{equation}
\end{definition}

As before, given an RFD $(A,\O)$ in $\eR^N$, we denote its upper relative box dimension by $\ov{D}:=\ov{\dim}_B(A,\O)$; recall that $\ov{D}\leq N$.
We always assume, additionally, that $\ov{D}>-\ty$ and the screen $\bm S$ lies always to the left of the {\em critical line} $\{\re s=\ov{D}\}$, i.e., that $\sup S\leq\ov{D}$.
Also, in the sequel, we assume that $\inf S>-\ty$ (see, however, Remark \ref{5.1.1.1/2R} below); hence, we have that
\begin{equation}\label{5.1.1.1/2-}
-\ty<\inf S\leq\sup S\leq\ov{D}.
\end{equation}

Moreover, the {\em window} $\bm W$ is defined as the part of the complex plane to the right of $\bm S$; that is,
\begin{equation}
\bm W:=\{s\in\Ce:\re s\geq S(\im s)\}.
\end{equation}
We say that the relative fractal drum $(A,\O)$ is {\em admissible} if its relative tube (or distance) zeta function can be meromorphically extended (necessarily uniquely) to an open connected neighborhood of some window $\bm W$, defined as above.

\begin{remark}\label{5.1.1.1/2R}
Occasionally, in the strongly languid case, in the sense of Definition \ref{str_languid} or Definition \ref{d_lang} below (and hence, in particular, when the fractal zeta function involved is meromorphic on all of $\Ce$), it is convenient to implicitly move the screen $\bm{S}$ to $-$infinity (i.e., to let $S\equiv -\ty$) and thus to choose $\bm{W}:=\Ce$. 
\end{remark}

The next definition adapts~\cite[Definition~5.2]{lapidusfrank12} to the case of relative fractal drums in $\eR^N$ (and, in particular, to the case of bounded subsets of $\eR^N$).

\begin{definition}[{\em Languidity}, adapted from~\cite{lapidusfrank12}]\label{languid}
An admissible relative fractal drum $(A,\O)$ in $\eR^N$ is said to be {\em languid} if for some fixed $\d>0$, its tube zeta function $\widetilde{\zeta}_{A,\O}(\,\cdot\,;\d)$ satisfies the following growth conditions:

\medskip

There exists a real constant $\kappa$ and a two-sided sequence $(T_n)_{n\in\Ze}$ of real numbers such that $T_{-n}<0<T_n$ for all $n\geq 1$ and
\begin{equation}\label{seq_T_n}
\lim_{n\to\ty}T_n=+\ty,\quad\lim_{n\to\ty}T_{-n}=-\ty
\end{equation}
satisfying the following two hypotheses, {\bf L1} and {\bf L2}:\footnote{Here, unlike in the definition given in~\cite{lapidusfrank12}, we do not need to assume that $\lim_{n\to +\ty}T_n/|T_{-n}|=1$.}
\newline

{\bf L1}\ \ For a fixed real constant $c>\ov{\dim}_B(A,\O)$, there exists a positive constant $C>0$ such that for all $n\in\Ze$ and all $\s\in (S(T_n),c)$,\footnote{This is a slight modification of the original definition of languidity from~\cite{lapidusfrank12}, where $c$ was replaced by $+\ty$; compare with \cite[Definition~5.2, pp.\ 146--147]{lapidusfrank12}.
Furthermore, it is clear that if condition {\bf L1} is satisfied for some $c>\ov{\dim}_B(A,\O)$, then it is also satisfied for any $c_1$ such that $\ov{\dim}_B(A,\O)<c_1<c$.}
\begin{equation}\label{L1}
|\widetilde{\zeta}_{A,\O}(\s+\I T_n;\d)|\leq C(|T_n|+1)^{\kappa}.
\end{equation}

{\bf L2}\ \ For all $\tau\in\eR$, $|\tau|\geq 1$,
\begin{equation}\label{L2}
|\widetilde{\zeta}_{A,\O}(S(\tau)+\I \tau;\d)|\leq C|\tau|^{\kappa},
\end{equation}
where $C$ is a positive constant which can be chosen to be the same one as in condition~{\bf L1}.
\end{definition}

Note that hypothesis {\bf L1} is a polynomial growth condition along horizontal segments (necessarily not passing through any singularities of $\widetilde{\zeta}_{A,\O}(\,\cdot\,;\d)$), while hypothesis {\bf L2} is a polynomial growth condition along the vertical direction of the screen.
These hypotheses will be needed to establish the pointwise and distributional tube formulas with error term.

In order to obtain the pointwise and distributional tube formulas without error terms (that is, {\em exact} tube formulas), we will need a stronger notion of languidity.
Accordingly, we introduce the following definition, which adapts to our current more general situation the definition of strong laguidity given in \cite[Definition~5.3]{lapidusfrank12}.

\begin{definition}[{\em Strong languidity}, adapted from~\cite{lapidusfrank12}]\label{str_languid}
We say that an admissible relative fractal drum $(A,\O)$ in $\eR^N$ is {\em strongly languid} if $(i)$ for some $\d>0$, its tube zeta function satisfies condition {\bf L1} with $S(T_n)\equiv -\ty$ (that is, with $S(T_n)$ replaced by $-\ty$) in~\eqref{L1}, i.e., estimate \eqref{L1} holds for every $\sigma<c$; and, additionally, $(ii)$ there exists a sequence of screens $\bm{S}_m\colon \tau\mapsto S_m(\tau)+\I \tau$ for $m\geq 1$, $\tau\in\eR$ with $\sup S_m\to -\ty$ as $m\to\ty$ and with a uniform Lipschitz bound, $\sup_{m\geq 1}\|S_m\|_{\mathrm{Lip}}<\ty$, such that the following condition holds:
\newline

{\bf L2'} There exist constants $B,C>0$ such that  for all $\tau\in\eR$ and $m\geq 1$,
\begin{equation}\label{L2'}
|\widetilde{\zeta}_{A,\O}(S_m(\tau)+\I \tau;\d)|\leq CB^{|S_m(\tau)|}(|\tau|+1)^{\kappa}.
\end{equation} 
\end{definition}

One immediately sees that hypothesis {\bf L2'} implies hypothesis {\bf L2}; so that a strongly languid relative fractal drum is also languid.
We also note that if a relative fractal drum is languid for some $\kappa\in\eR$, then it is also languid for any $\kappa_1>\kappa$.
Observe that for $\widetilde{\zeta}_{A,\O}:=\widetilde{\zeta}_{A,\O}$ (or, equivalently, the RFD $(A,\O)$) to be strongly languid, $\widetilde{\zeta}_{A,\O}$ must admit a meromorphic continuation to all of $\Ce$; see also Remark \ref{5.1.1.1/2R} above.

We will also use the notion of languid (or else, strongly languid) relative tube zeta function, in the obvious sense.

\medskip

As we shall see, most of the geometrically interesting examples of RFDs (and, in particular, of bounded sets) in $\eR^N$ considered here are either languid (relative to a suitable screen), in the sense of Definition \ref{languid} above (or of its counterpart for the distance zeta function, in Definition \ref{d_lang} below) or else, strongly languid, in the sense of Definition \ref{str_languid} just above (or, again, in the sense of Definition \ref{d_lang}).

\begin{proposition}\label{propA}
Let $(A,\O)$ be a relative fractal drum in $\eR^N$.
If the relative tube zeta function $\widetilde{\zeta}_{A,\O}(\,\cdot\,;\d)$ satisfies the languidity conditions {\bf L1} and {\bf L2} for some $\d>0$ and $\kappa\in\eR$, then so does
$\widetilde{\zeta}_{A,\O}(\,\cdot\,;\d_1)$ for any $\d_1>0$ and with $\kappa_{\d_1}:=\max\{\kappa,0\}$.

Furthermore, the analogous statement is also true in the case when $\widetilde{\zeta}_{A,\O}(\,\cdot\,;\d)$ is strongly languid, under the additional assumption that $\d\geq 1$ and $\d_1\geq 1$.
\end{proposition}

\begin{proof}
Without loss of generality, we may assume that $\d<\d_1$.
Then the conclusion follows from the fact that $\widetilde{\zeta}_{A,\O}(\,\cdot\,;\d_1)=\widetilde{\zeta}_{A,\O}(\,\cdot\,;\d)+f(s)$, where $f$ is entire and
\begin{equation}\label{f_s_upper}
|f(s)|\leq\int_{\d}^{\d_1}t^{\re s-N-1}|A_t\cap\O|\di t\leq\begin{cases}
|\O|\frac{{\d_1^{\re s-N}-\d^{\re s-N}}}{\re s-N},& \re s\neq N\\
|\O|(\log\delta_1-\log\delta),&\re s= N.
\end{cases}
\end{equation}
Since, clearly, the upper bound on $|f(s)|$ does not depend on $\im s$, we conclude that $f$ satisfies the languidity conditions {\bf L1} and {\bf L2} with the languidity exponent $\kappa_f:=0$ and for any given window $\bm W$.
This observation implies that then, $\widetilde{\zeta}_{A,\O}(\,\cdot\,;\d_1)$ is languid for $\kappa_{\d_1}:=\max\{\kappa,0\}$ and with the same window as for $\widetilde{\zeta}_{A,\O}(\,\cdot\,;\d)$.

The additional assumption for strong languidity is needed since {\bf L1} must then be satisfied for all $\sigma\in(-\ty,c)$, in the notation of Definition \ref{languid}, and for this to be achieved we need that $\d_1>\d\geq 1$ in \eqref{f_s_upper} since otherwise, we do not have an upper bound on $|f(s)|$ when $\re s\to -\ty$.
\end{proof}

Let us now introduce the notion of complex dimensions of a relative fractal drum.

\begin{definition}[Complex dimensions of an RFD~{[LapRa\v Zu1,4]}]\label{comp_dim_def} 
Let $(A,\O)$ be a relative fractal drum in $\eR^N$. 
Assume that the relative tube zeta function $\widetilde{\zeta}_{A,\O}(\,\cdot\,;\d)$ has a meromorphic extension to a connected neighborhood $U$ of the critical line $\{\re s=\ov{\dim}_B(A,\O)\}$. Then, the set of {\em visible complex dimensions of $(A,\O)$ $($with respect to $U)$} is the set of poles of $\widetilde{\zeta}_{A,\O}(\,\cdot\,;\d)$ that belong to $U$ and we denote it by
\begin{equation}
\po\big(\widetilde{\zeta}_{A,\O}(\,\cdot\,;\d),U\big):=\big\{\omega\in U:\omega\textrm{ is a pole of }\widetilde{\zeta}_{A,\O}(\,\cdot\,;\d)\big\}.
\end{equation}
If $U:=\Ce$, we say that $\po(\widetilde{\zeta}_{A,\O}(\,\cdot\,;\d),\Ce)$ is the set of {\em complex dimensions} of $(A,\O)$ and denote it by $\dim_{\Ce}(A,\O)$.

Furthermore, we call the set of poles located on the critical line $\{\re s=\ov{\dim}_B(A,\O)\}$ the set of {\em principal complex dimensions of $(A,\O)$} and denote it by
\begin{equation}
\dim_{PC}(A,\O):=\big\{\omega\in \po(\widetilde{\zeta}_{A,\O}(\,\cdot\,;\d),U):\re \omega=\ov{\dim}_B(A,\O)\big\}.
\end{equation} 
\end{definition}

\begin{remark}
In light of the functional equation~\eqref{equ_tilde} and the relevant discussion concerning it, the above definition can also be made in terms of the relative distance zeta function; that is, we always have
$\po(\widetilde{\zeta}_{A,\O}(\,\cdot\,;\d),U)=\po({\zeta}_{A,\O}(\,\cdot\,;\d),U)$
whenever one of the above zeta functions has a meromorphic extension to the domain $U$ containing the critical line $\{\re s=\ov{\dim}_B(A,\O)\}$ and if $N\notin\po(\widetilde{\zeta}_{A,\O}(\,\cdot\,;\d),U)$.\footnote{In that case, the other zeta function also has a meromorphic continuation to $U$.}
Furthermore, according to Remark~\ref{holo_diff}, the set of (visible) complex dimensions $\po(\widetilde{\zeta}_{A,\O}(\,\cdot\,;\d),U)$ of a relative fractal drum $(A,\O)$ does not depend on $\d$. 
\end{remark}

In order to obtain the relative tube formula expressed in terms of the complex dimensions of the relative fractal drum $(A,\O)$, we will need to work (for each $k\in\eN$) with the $k$-th primitive (or $k$-th anti-derivative) function, $V^{[k]}=V^{[k]}(t)$, of the relative tube function $V=V(t)$ vanishing along with its first $(k-1)$ derivatives at $t=0$.
Therefore, we let
\begin{equation}\label{V_tube}
V(t)=V_{A,\O}(t)=V^{[0]}(t):=|A_t\cap\O|
\end{equation}
and
\begin{equation}
V^{[k]}(t)=V^{[k]}_{A,\O}(t):=\int_0^tV^{[k-1]}(\tau)\di\tau,\quad\textrm{ for each }k\in\eN.
\end{equation}
(Here and thereafter, we let $\eN:=\{1,2,3,\ldots\}$ and $\eN_0:=\eN\cup\{0\}$.)
In the case of a bounded subset $A\subset\eR^N$, we use the analogous notation $V^{[k]}(t)=V_A^{[k]}(t)$ for the $k$-th primitive function of the tube function $V(t)=V_A(t):=|A_t|$, where $k\in\eN_0$.
Furthermore, we recall that for any $s\in\Ce$, the {\em Pochammer} symbol is defined by
\begin{equation}\label{pochamer}
(s)_0:=1,\quad (s)_k:=s(s+1)\cdots(s+k-1),
\end{equation}
for any nonnegative integer $k$ and, more generally, for the purpose of Section~\ref{sec_distr}, for every $k\in\Ze$ by
\begin{equation}\label{pochamer_gamma}
(s)_k:=\frac{\mathrm{\Gamma}(s+k)}{\mathrm{\Gamma}(s)},
\end{equation}
where $\mathrm{\Gamma}$ denotes the gamma function.

\begin{remark}\label{2.16.1/2}
One may legitimately wonder why we work with the $k$-th primitive, for any $k\geq 0$ rather than simply for $k=0$.
There are several reasons for that, one of them being that the larger $k$, the weaker our assumptions in the statement of our pointwise tube formulas.
Furthermore, in proving the distributional tube formulas, we will essentially use our corresponding pointwise tube formula at level $k$, with $k$ sufficiently large, and then distributionally differentiate the resulting formula in order to obtain a distributional formula valid at any level $l\in\Ze$ (rather than for $l\in\eN_0$), the case when $l=-1$ being the most fundamental one in that distributional situation.\footnote{This is analogous to the way periodic distributions are shown to have a distributionally convergent Fourier series (under rather weak hypotheses), by integrating sufficiently many times and then using the classic pointwise result about the uniform convergence of Fourier series; see \cite[Section VII, I, esp., p.\ 226]{Schw}.}
(See, e.g, the proof of Theorem \ref{dist_error} in Section \ref{subs_distr}.)
When $N=1$ (i.e., in the case of fractal strings), this same method was already used in [Lap-vFr1--3] in order to deduce the distributional explicit formula from its pointwise counterpart; see Remark 5.20 along with the first proof of Theorem 5.18 in \cite{lapidusfrank12}.
There will, however, be several technical differences in the execution of the method, which we will not necessarily point out.
\end{remark}

Before stating the main relationship connecting $V^{[k]}=V^{[k]}_{A,\O}$ and the tube zeta function $\widetilde{\zeta}_{A,\O}$ of the RFD $(A,\O)$, valid for any integer $k\geq 0$, we begin by considering the key special case when $k=0$ (so that $V^{[0]}=V=V_{A,\O}$).
In order to proceed, we need to briefly provide some basic information about the Mellin transform and its inverse transform.
Recall that the Mellin transform of a function $f\colon\eR\to\eR$ is defined by Equation \eqref{mell_trans_def}.
Furthermore, the Mellin inversion theorem, which we recall here for the sake of completeness, together with 
Equation \eqref{mellin_tube} yields an integral expression for the tube function of a given relative fractal drum.

\begin{theorem}[Mellin's inversion theorem, {\rm cited from~{\cite[Theorem 28]{titch}}}]\label{mellin_inv}
Let $f\colon(0,+\ty)\to\eR$ be such that for a given $y>0$, $f(t)$ is of bounded variation in a neighborhood of the point $t=y$.
Furthermore, assume that $t\mapsto t^{c-1}f(t)$ belongs to $L^1(0,+\ty)$, where $c$ is a real number, and define
\begin{equation}\label{ddeff_mell}
\{\mathfrak{M}f\}(s):=\int_{0}^{+\ty}t^{s-1}f(t)\di t
\end{equation}
for all $s\in\Ce$ such that $\re s=c$.
Then, for the above value of $y$, the following inversion formula holds$:$
\begin{equation}\label{ddeff_inv}
\frac{1}{2}\big(f(y+0)+f(y-0)\big)=\frac{1}{2\pi\I}\int_{c-\I\ty}^{c+\I\ty}y^{-s}\{\mathfrak{M} f\}(s)\di s,
\end{equation}
where $f(y+0)$ and $f(y-0)$ denote, respectively, the right and left limits of $f$ at $y$.
Here, on the right-hand side of \eqref{ddeff_inv}, the contour integral is taken over the vertical line $\{\re s=c\}$.
%
\end{theorem}

We can now state the announced integral formula connecting the relative tube function of the RFD $(A,\O)$ and the tube zeta function $\widetilde{\zeta}_{A,\O}:=\widetilde{\zeta}_{A,\O}(\,\cdot\,;\d)$. 

\begin{theorem}\label{tube_inversion}
Let $(A,\O)$ be a relative fractal drum in $\eR^N$ and fix $\d>0$.
Then, for any fixed $c>\ov{\dim}_{B}(A,\O)$ and for every $t\in(0,\d)$, we have
\begin{equation}\label{tube_inversion_formula}
|A_t\cap\O|=\frac{1}{2\pi\I}\int_{c-\I\ty}^{c+\I\ty}t^{N-s}\widetilde{\zeta}_{A,\O}(s;\d)\di s.
\end{equation}
\end{theorem}
\begin{proof}
Let $f(t):=\chi_{(0,\d)}(t)t^{-N}|A_t\cap\O|$ and observe that $t\mapsto|A_t\cap\O|$ is nondecreasing, and hence, is locally of bounded variation on $(0,+\ty)$.
Since the product of two functions of locally bounded variation is also a function of locally bounded variation, we conclude that $f$ is also locally of bounded variation on $(0,+\ty)$.
Furthermore, we deduce from Theorem~\ref{an_rel} and from the functional equation~\eqref{equ_tilde} (see also the end of Remark \ref{tube_holo}) that the integral defining the tube zeta function $\widetilde{\zeta}_{A,\O}$ in Equation \eqref{mellin_tube} is absolutely convergent (and hence, convergent) for all $s\in\Ce$ such that $\re s>\ov{\dim}_B(A,\O)$ or, in other words, $t\mapsto t^{\re s-1}f(t)$ belongs to $L^1(0,+\ty)$ for such $s$.
Consequently, the Mellin transform $\{\mathfrak{M}f\}(s)$ of $f$ is well defined by Equation \eqref{ddeff_mell} and coincides with $\widetilde{\zeta}_{A,\O}(s;\d)$ for $c=\re s>\ov{\dim}_B(A,\O)$; that is, Equation \eqref{mellin_tube} holds for all $s\in\Ce$ such that $\re s>\ov{\dim}_B(A,\O)$, as was claimed above.
Therefore, by Theorem~\ref{mellin_inv}, we can recover the relative tube function from the relative tube zeta function and for positive $y\neq\d$, we have
\begin{equation}
\chi_{(0,\d)}(y)y^{-N}|A_y\cap\O|=\frac{1}{2\pi\I}\int_{c-\I\ty}^{c+\I\ty}y^{-s}\widetilde{\zeta}_{A,\O}(s;\d)\di s,
\end{equation}
where $c>\ov{\dim}_B(A,\O)$ is arbitrary; that is, \eqref{tube_inversion_formula} is valid for all $t\in(0,\d)$, as desired.
\end{proof}

One of our main goals in this paper will be to express formula \eqref{tube_inversion_formula} in a more useful and applicable way.
More specifically, we will express the right-hand side of \eqref{tube_inversion_formula} in terms of the relative distance zeta function and as a sum (interpreted in a suitable way) of residues over the complex dimensions of the given relative fractal drum.
The resulting identity will be called a ``fractal tube formula'' (as in \cite{lapidusfrank12}) or simply, a tube formula.

\medskip

A priori, one would naively expect that Equation \eqref{ddeff_mell} and hence also, Equation \eqref{ddeff_inv}, only holds for $c\geq N$.
(Indeed, since $f(t)=0$ for all $t\geq\d$ and $|A_t\cap\O|\leq|\O|$, we easily see that $t\mapsto t^{c-1}f(t)$ belongs to $L^1(0,+\ty)$ for $c\geq N$.)
The stronger conclusion obtained in Theorem \ref{tube_inversion} requires the aforementioned results obtained in [LapRa\v Zu1--2] 
and \cite{refds}.

\medskip

The following result is really a corollary of Theorem~\ref{tube_inversion} but given its importance for the rest of this section, we state it as a separate proposition.

\begin{proposition}\label{kth_prim}
Let $(A,\O)$ be a relative fractal drum in $\eR^N$ and let $\d>0$ be fixed.
Then for every $t\in(0,\d)$ and $k\in\eN_0$, we have
\begin{equation}\label{kth_prim_eq}
V^{[k]}_{A,\O}(t)=\frac{1}{2\pi\I}\int_{c-\I\ty}^{c+\I\ty}\frac{t^{N-s+k}}{(N\!-\! s\! +\! 1)_{k}}\widetilde{\zeta}_{A,\O}(s;\d)\di s,
\end{equation}
where $c\in(\ov{\dim}_{B}(A,\O),N+1)$ is arbitrary.
\end{proposition}

\begin{proof}
By Theorem~\ref{tube_inversion}, we have the following equalities, valid (pointwise) for all $t\in(0,\d)$:
$$
\begin{aligned}
V^{[1]}_{A,\O}(t)=\int_{0}^tV_{A,\O}(\tau)\di\tau&=\frac{1}{2\pi\I}\int_{0}^t\int_{c-\I\ty}^{c+\I\ty}\tau^{N-s}\widetilde{\zeta}_{A,\O}(s;\d)\di s\di\tau\\
&=\frac{1}{2\pi\I}\int_{c-\I\ty}^{c+\I\ty}\frac{t^{N-s+1}}{N\!-\! s\!+\! 1}\widetilde{\zeta}_{A,\O}(s;\d)\di s,\\
\end{aligned}
$$
since $N-c+1>0$.
The change of the order of integration is justified by combining Lebesgue's dominated convergence theorem and the Fubini--Tonelli theorem. Iterating this calculation $k-1$ times, we prove the statement of the proposition.
\end{proof}

We adapt the following definition of the truncated screen and window from Section 5.3 of \cite{lapidusfrank12}, where it was stated for languid generalized fractal strings and can now be used in the same form in the case of relative fractal drums in $\eR^N$.
 
\begin{definition}[The truncated screen and window]\label{trunc_screen_window}
Given an integer $n\geq 1$ and a languid relative fractal drum in $\eR^N$, the {\em truncated screen} $\bm{S}_{|n}$ is the part of the screen $\bm S$ restricted to the interval $[T_{-n},T_n]$, and the {\em truncated window} $\bm{W}_{|n}$ is the window $\bm W$ intersected with the horizontal strip between $T_{-n}$ and $T_n$; i.e.,
\begin{equation}\label{trunc_window}
\bm W_{|n}:=\bm W\cap\{s\in\Ce\,:\, T_{-n}\leq\im s\leq T_n\}.
\end{equation}

We then call $\po(\widetilde{\zeta}_{A,\O},\bm W_{|n})$ the set of {\em truncated visible complex dimensions},
i.e., it is the set of visible complex dimensions of $(A,\O)$ relative to the window $\bm W$ and with imaginary parts between $T_{-n}$ and $T_n$.
Note that since by assumption, there are no poles of $\widetilde{\zeta}_{A,\O}$ along the screen $\bm S$, we could replace $\bm{W}_{|n}$ by its interior $\mathring{\bm W}_{|n}$, in the aforementioned notation:
\begin{equation}
\po\left(\widetilde{\zeta}_{A,\O},\bm W_{|n}\right)=\po\left(\widetilde{\zeta}_{A,\O},\mathring{\bm W}_{|n}\right).
\end{equation}
\end{definition}

\section{Pointwise Fractal Tube Formula}\label{sec_point}

In this section, our main goal is to obtain fractal tube formulas via the tube zeta function which are valid pointwise.
Furthermore, depending on the growth properties of the corresponding tube zeta function, these fractal tube formulas will be either exact or else approximate with a pointwise error term.

\subsection{Pointwise Tube Formula with Error Term}\label{subsec_point_error}

From now on, the phrase ``let $(A,\O)$ be a languid (or strongly languid) relative fractal drum'' will implicitly mean that $(A,\O)$ is admissible for some window $\bm W$ and for some $\d>0$, the relative tube zeta function $\widetilde{\zeta}_{A,\O}(s;\d)$ of $(A,\O)$ satisfies the languidity conditions of Definition~\ref{languid} (or Definition~\ref{str_languid}, respectively).
We will first obtain a `truncated pointwise tube formula' (Lemma \ref{trunc_point}), from which the main result (Theorem~\ref{pointwise_formula}) will follow.
(Note that Lemma \ref{trunc_point} is the counterpart, now valid for any $N\geq 1$, of \cite[Lemma 5.9]{lapidusfrank12}.)
Recall from the end of Section \ref{prelm} that for each integer $n\geq 1$, the truncated screen $\bm{S}_{|n}$ and associated truncated window $\bm{W}_{|n}$ were defined in Definition \ref{trunc_screen_window}.

\begin{lemma}[Truncated pointwise tube formula]\label{trunc_point}
Let $k\geq 0$ be an integer and $(A,\O)$ a languid RFD in $\eR^N$ for a fixed $\d>0$.
Furthermore, fix a constant $c\in(\ov{\dim}_B(A,\O),N+1)$.
Then, for all $t\in(0,\d)$ and all integers $n\geq 1$, we have
\begin{equation}\label{trunc_point_formula}
\begin{aligned}
I_n&:=\frac{1}{2\pi\I}\int_{c+\I T_{-n}}^{c+\I T_n}\frac{t^{N-s+k}}{(N\!-\! s\!+\! 1)_k}\widetilde{\zeta}_{A,\O}(s;\d)\di s\\
&\phantom{:}=\sum_{\omega\in\po(\widetilde{\zeta}_{A,\O},\bm W_{|n})}\res\left(\frac{t^{N-s+k}}{(N\!-\! s\!+\! 1)_k}\widetilde{\zeta}_{A,\O}(s;\d),\omega\right)\\
&\phantom{:=}+\frac{1}{2\pi\I}\int_{\bm{S}_{|n}}\frac{t^{N-s+k}}{(N\!-\! s\!+\! 1)_k}\widetilde{\zeta}_{A,\O}(s;\d)\di s+E_n(t).
\end{aligned}
\end{equation}
Moreover, assuming that hypothesis {\bf L1} is fulfilled, we have the following pointwise remainder estimate, valid for all $t\in(0,\d)$$:$
\begin{equation}\label{Eest0}
\begin{aligned}
|E_n(t)|\leq{t^{N+k}}K_{\kappa}\max\big\{T_n^{\kappa-k},|T_{-n}|^{\kappa-k}\big\}(c-\inf S)\max\big\{t^{-c},t^{-\inf S}\big\},\\
\end{aligned}
\end{equation}
where $K_{\kappa}$ is a positive constant depending only on $\kappa$.\footnote{More precisely, $K_{\kappa}$ depends only on $\kappa$ and the constant $C$ occurring in hypothesis {\bf L1}.}

Finally, for each point $s=S(\tau)+\I\tau$, where $\tau\in\eR$ is such that $|\tau|>1$, and for all $t\in(0,\d)$, the integrand over the truncated screen $\bm{S}_{|n}$ appearing in~\eqref{trunc_point_formula} is bounded in absolute value by
\begin{equation}\label{integrand1}
Ct^{N+k}\max\big\{t^{-\sup S},t^{-\inf S}\big\}|\tau|^{\kappa-k},
\end{equation}
when hypothesis {\bf L2} holds, and by
\begin{equation}\label{integrand2}
C_{\kappa}t^{N+k}\max\big\{B^{|\inf S|},B^{|\sup S|}\big\}\max\big\{t^{-\sup S},t^{-\inf S}\big\}|\tau|^{\kappa-k},
\end{equation}
when hypothesis {\bf L2'} holds, with the constant $C_\kappa$ depending only on $\kappa$.\footnote{Here, the constant $C_\kappa$ actually depends only on $\kappa$ and on the constant $C$ appearing in hypothesis~{\bf L1}.}
\end{lemma}

\begin{proof}
We let $\ov{D}:=\ov{\dim}_B(A,\O)$ and for the sake of brevity, write $\widetilde{\zeta}_{A,\O}(s)$ instead of $\widetilde{\zeta}_{A,\O}(s;\d)$ throughout the proof.
Next, we replace the integral over the segment $[c+\I T_{-n},c+\I T_n]$ with the integral over the contour $\Gamma$ consisting of this segment, the truncated screen $\bm{S}_{|n}$ and the two horizontal segments joining $S(T_{\pm n})+\I T_{\pm n}$ and $c+\I T_{\pm n}$ (see Figure~\ref{screen_fig}).
In other words, we have
$$
\begin{aligned}
I_n=\frac{1}{2\pi\I}\int_{c+\I T_{-n}}^{c+\I T_n}\!\!\frac{t^{N-s+k}}{(N\!-\! s\!+\! 1)_k}\widetilde{\zeta}_{A,\O}(s)\di s&=\frac{1}{2\pi\I}\oint_{\Gamma}\frac{t^{N-s+k}}{(N\!-\! s\!+\! 1)_k}\widetilde{\zeta}_{A,\O}(s)\di s\\
&\phantom{=}+\frac{1}{2\pi\I}\int_{\bm{S}_{|n}}\!\!\frac{t^{N-s+k}}{(N\!-\! s\!+\! 1)_k}\widetilde{\zeta}_{A,\O}(s)\di s\!+\! E_n(t),
\end{aligned}
$$
where
$$
\begin{aligned}
E_n(t)&:=\frac{1}{2\pi\I}\int_{\Gamma_L\cup\Gamma_U}\frac{t^{N-s+k}}{(N\!-\! s\!+\! 1)_k}\widetilde{\zeta}_{A,\O}(s)\di s.\\
\end{aligned}
$$
\begin{figure}[t]
\begin{center}
\includegraphics[trim=0cm 0cm 0cm 0cm,clip=true,width=9cm]{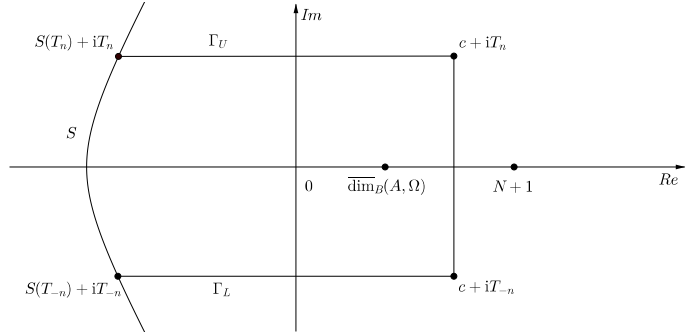}
\end{center}
\caption{The truncated window $\bm{W}_{|n}$ and the contour $\Gamma$ which we use to estimate the integral $I_n$ in the proof of Lemma~\ref{trunc_point}.}
\label{screen_fig}
\end{figure}

Furthermore, note that the integrand appearing above is meromorphic on the bounded domain having $\Gamma$ as its boundary and its poles are exactly the poles of the relative tube zeta function since $c\in(\ov{\dim}_B(A,\O),N+1)$ ensures that there are no zeros of $(N\!-\!s \!+\!1)_k$ inside of $\Gamma$.
Consequently, we deduce from the residue theorem that
$$
\begin{aligned}
I_n&=\!\!\!\!\sum_{\omega\in\po(\widetilde{\zeta}_{A,\O},W_{|n})}\!\!\!\!\res\left(\frac{t^{N-s+k}}{(N\!-\! s\!+\! 1)_k}\widetilde{\zeta}_{A,\O}(s),\omega\right)+\frac{1}{2\pi\I}\int_{\bm{S}_{|n}}\!\!\frac{t^{N-s+k}\,\widetilde{\zeta}_{A,\O}(s)}{(N\!-\! s\!+\! 1)_k}\di s+E_n(t).
\end{aligned}
$$
To obtain the upper bound on $|E_n(t)|$, we first observe that for $s=\sigma+\I T_n$ we have $|(N-s+1)_k|\geq T_n^k$ and then, under hypothesis {\bf L1}, we estimate the integrals over the upper segment $\Gamma_U$ and the lower segment $\Gamma_L$ as follows:
$$
\begin{aligned}
\left|\int_{\Gamma_U}\frac{t^{N-s+k}}{(N\!-\!s \!+\!1)_k}\widetilde{\zeta}_{A,\O}(s)\di s\right|&=\left|\int_{\bm{S}(T_n)}^{c}\frac{t^{N+k-\sigma-\I T_n}}{(N\!+\! 1\!-\! (\sigma\! +\! \I T_n))_k}\widetilde{\zeta}_{A,\O}(\sigma+\I T_n)\di\sigma\right|\\
&\leq  t^{N+k}C(T_n+1)^{\kappa}T_n^{-k}\int_{\bm{S}(T_n)}^{c}t^{-\sigma}\di\sigma\\
&\leq  t^{N+k}K_\kappa T_n^{\kappa-k}\big(c-S(T_n)\big)\max\big\{t^{-c},t^{-S(T_n)}\big\},
\end{aligned}
$$
where $K_{\kappa}$ is a positive constant such that $C(|T_n|+1)^{\kappa}\leq K_\kappa|T_n|^{\kappa}$ for all $n\in\Ze$.
Furthermore, since $\inf S\leq S(\tau)$ for all $\tau\in\eR$, we have
\begin{equation}\label{E3}
\begin{aligned}
\left|\int_{\Gamma_U}\frac{t^{N-s+k}\widetilde{\zeta}_{A,\O}(s)\di s}{(N\!-\! s\!+\! 1)_k}\right|\leq t^{N+k}K_\kappa T_n^{\kappa-k}(c-\inf S)\max\big\{t^{-c},t^{-\inf S}\big\}.
\end{aligned}
\end{equation}
Analogously, we bound the integral over the lower line segment by
\begin{equation}\label{E4}
\begin{aligned}
t^{N+k}K_\kappa |T_{-n}|^{\kappa-k}(c-\inf S)\max\big\{t^{-c},t^{-\inf S}\big\}.
\end{aligned}
\end{equation}
Therefore, putting~\eqref{E3} and \eqref{E4} together, we obtain the upper bound~\eqref{Eest0}.\footnote{The constant $K_\kappa$ in~\eqref{Eest0} is actually equal to the present constant $K_\kappa$ divided by $\pi$.}

To estimate the integrand over the truncated screen $\bm{S}_{|n}$, we observe that for $s=S(\tau)+\I\tau$ with $|\tau|>1$, we have
\begin{equation}
\begin{aligned}
\left|\frac{t^{N-s+k}}{(N\!-\! s\!+\! 1)_k}\widetilde{\zeta}_{A,\O}(s)\right|&\leq Ct^{N-S(\tau)+k}|\tau|^{\kappa-k}\\
&\leq Ct^{N+k}\max\big\{t^{-\sup S},t^{-\inf S}\big\}|\tau|^{\kappa-k},
\end{aligned}
\end{equation}
under hypothesis {\bf L2} and similarly, under hypothesis {\bf L2'}, which completes the proof of the lemma.
\end{proof}

Next, we state and prove the main result of this subsection.

\begin{theorem}[Pointwise fractal tube formula with error term, via $\widetilde{\zeta}_{A,\O}$]\label{pointwise_formula}
Let $(A,\O)$ be a relative fractal drum in $\eR^N$ which is languid for some fixed $\d>0$ and some fixed exponent $\kappa\in\eR$.
Furthermore, let $k>\kappa+1$ be a nonnegative integer.
Then, the following pointwise fractal tube formula with error term, expressed in terms of the tube zeta function $\widetilde{\zeta}_{A,\O}:=\widetilde{\zeta}_{A,\O}(\,\cdot\,;\d)$, is valid for every $t\in(0,\d)$$:$
\begin{equation}\label{point_form}
V^{[k]}_{A,\O}(t)=\sum_{\omega\in\po(\widetilde{\zeta}_{A,\O},\bm{W})}\res\left(\frac{t^{N-s+k}}{(N\!-\! s\!+\! 1)_k}\widetilde{\zeta}_{A,\O}(s),\omega\right)+\widetilde{R}^{[k]}_{A,\O}(t).
\end{equation}
Here, for every $t\in(0,\d)$, the $($pointwise$)$ error term $\widetilde{R}^{[k]}_{A,\O}$ is given by the absolutely convergent $($and hence, convergent$)$ integral
\begin{equation}\label{error_term}
\widetilde{R}^{[k]}_{A,\O}(t)=\frac{1}{2\pi\I}\int_{\bm S}\frac{t^{N-s+k}}{(N\!-\! s\!+\! 1)_k}\widetilde{\zeta}_{A,\O}(s)\di s.
\end{equation}
Furthermore, we have the following pointwise error estimate, valid for all $t\in(0,\d)$$:$
\begin{equation}\label{R_upper_bound}
\big|\widetilde{R}^{[k]}_{A,\O}(t)\big|\leq t^{N+k}\max\{t^{-\sup S},t^{-\inf S}\}\left(\frac{C\big(1+\|S\|_{\mathrm{Lip}}\big)}{2\pi(k-\kappa-1)}+C'\right),
\end{equation}
where $C$ is the positive constant appearing in {\bf L1} and {\bf L2} and $C'$ is some suitable positive constant.
These constants depend only on the relative fractal drum $(A,\O)$ and the screen, but not on the value of the nonnegative integer $k$.

In particular, we have the following pointwise error estimate$:$
\begin{equation}\label{estm}
\widetilde{R}^{[k]}_{A,\O}(t)=O(t^{N-\sup S+k})\quad\mathrm{\ as\ }\quad t\to 0^+.
\end{equation}
Moreover, if $S(\tau)<\sup S$ for all $\tau\in\eR$ $($i.e., if the screen $\bm S$ lies strictly to the left of the vertical line $\{\re s=\sup S\}$$)$, then we have the following stronger pointwise estimate$:$
\begin{equation}\label{S(t)<sup}  
\widetilde{R}^{[k]}_{A,\O}(t)=o(t^{N-\sup S+k})\quad\mathrm{\ as\ }\quad t\to 0^+.
\end{equation}
\end{theorem}

Before proving Theorem \ref{pointwise_formula}, we make the following two comments (in parts $(a)$ and $(b)$ of Remark \ref{a_b_rem}), which will help the reader to understand the statement of the theorem.
Furthermore, we also point out that comments similar to those in Remark \ref{a_b_rem} also apply to all other theorems stated below, in which a (typically infinite) sum over the $($visible$)$ complex dimensions appears, either in reference to a pointwise or distributional fractal tube formula.
Of course, in the case of the distributional fractal tube formula the (potentially infinite) sum has to be interpreted as a distributional (rather than pointwise) limit of the partial sums.

\begin{remark}\label{a_b_rem}
$(a)$\ \ The (potentially infinite) sum appearing in~\eqref{point_form} in the above theorem (Theorem \ref{pointwise_formula}) should be interpreted as the limit
\begin{equation}\label{5.1.13.1/2.1/2}
\lim_{n\to\ty}\sum_{\omega\in\po(\widetilde{\zeta}_{A,\O},\bm{W}_{|n})}\res\left(\frac{t^{N-s+k}}{(N\!-\! s\!+\! 1)_k}\widetilde{\zeta}_{A,\O}(s),\omega\right),
\end{equation}
where $\bm{W}_{|n}$ is the truncated window (see Definition~\ref{trunc_screen_window}); that is, as the pointwise limit of the partial sums over the (visible) truncated complex dimensions, i.e., the poles of $\widetilde{\zeta}_{A,\O}$ located in $\bm{W}_{|n}$.
More specifically, the existence of this limit follows from the proof of the theorem in which we show that the series in \eqref{point_form} converges pointwise and conditionally.
On the other hand, we point out that Theorem~\ref{pointwise_formula} does not give any information about the possible absolute convergence of the series in \eqref{point_form}.
This situation is similar to the one which occurs in \cite[Chapters~5 and~8]{lapidusfrank12} and, in fact, also in Riemann's original explicit formula for the counting function of the prime numbers (see, e.g., \cite{Edw}).

\medskip

$(b)$\ \ The sum over the set $\po(\widetilde{\zeta}_{A,\O},\bm{W})$ in Equation \eqref{point_form} of Theorem~\ref{pointwise_formula} is independent of the parameter $\d$ since changing $\d$ has no effect on the residues appearing in~\eqref{point_form}.
This follows directly from the fact that the principal parts of a meromorphic extension of the relative tube zeta function around any of its poles do not depend on $\d$ (see [LapRa\v Zu1--4]). 
In other words, when applying Theorem~\ref{pointwise_formula}, one has to determine that $(A,\O)$ is languid for some $\d>0$, but when calculating the sum, one can take any $\d>0$; that is, in practice, the most convenient one.
\end{remark}

\begin{proof}[Proof of Theorem \ref{pointwise_formula}]
Without loss of generality, let $c\in(\ov{\dim}_B(A,\O),N+1)$ be the constant from the languidity condition {\bf L1} of Definition \ref{languid}.
We will prove the theorem by using Lemma~\ref{trunc_point} in order to obtain~\eqref{trunc_point_formula} for all $n\geq 1$ and then, by letting $n\to\ty$.
We note that $E_n(t)$ tends to zero for $k>\kappa$ at the rate of some negative power of $\min\{T_n, |T_{-n}|\}$.
Furthermore, for $k>\kappa+1$, the error term $\widetilde{R}_{A,\O}^{[k]}(t)$ is absolutely convergent.
Indeed, note that, since the function $\tau\mapsto S(\tau)$ is Lipschitz continuous, it is differentiable almost everywhere and, consequently, the derivative of the map $\tau\mapsto S(\tau)+\I\tau$ is bounded by $(1+\|S\|_{\mathrm{Lip}})$ for almost all $\tau\in\eR$.
Moreover, since
$
\int_{1}^{+\ty}\tau^{\kappa-k}\di\tau=(k-\kappa-1)^{-1}
$
for $k>\kappa+1$, the upper bound \eqref{R_upper_bound} on the error term $\widetilde{R}^{[k]}_{A,\O}(t)$ now follows from~\eqref{integrand1}.
The positive constant $C'$ in \eqref{R_upper_bound} is the constant which corresponds to the integral over the part of the screen for which $|\tau|<1$; i.e.,
$$
C':=\frac{1}{2\pi}\int_{\bm{S}\cap\{|\im S|<1\}}\frac{|\tilde{\zeta}_{A,\O}(s)|}{|(N\!-\! s\!+\! 1)_k|}\,|\D s|.
$$

In the case when the screen stays strictly to the left of the line $\{\re s=\sup S\}$, we can obtain the better estimate~\eqref{S(t)<sup} by using a well-known method; see, e.g.,~\cite[pp. 33--34]{In}.
Namely, for any given $\e>0$, we have to show that~\eqref{error_term} is bounded by $\e t^{N-\sup S+k}$.
For a given $T>0$, we can split the integral~\eqref{error_term} into two parts; namely, the integral over the part of the screen for which $|\im S|>T$ and the integral over the part of the screen for which $|\im S|\leq T$.
Since the first integral is absolutely convergent, we can choose $T$ sufficiently large so that it is bounded by $\frac{1}{2}\e t^{N-\sup S+k}$.
For the second integral, we observe that the maximum of $S(\tau)$ for $\tau\in[-T,T]$ is strictly less than $\sup S$; i.e., we can choose $\a>0$ such that $S(\tau)<\sup S-\a$ for all $\tau\in[-T,T]$.
This implies that the integral over the part of the screen for which $|\im S|\le T$ is of order $O(t^{N-\sup S+k+\a})$ as $t\to 0^+$.\footnote{Observe that since the screen $\bm S$ avoids the poles of the relative tube zeta function, we have that $\widetilde{\zeta}_{A,\O}(s)$ is bounded for all $s\in\Ce$ in the part of the screen $\bm S$ for which $|\im S|\leq T$.}
Hence, for all sufficiently small $t>0$ it is bounded by $\frac{1}{2}\e t^{N-\sup S+k}$.
This proves that $\widetilde{R}^{[k]}_{A,\O}(t)=o(t^{N-\sup S+k})$ as $t\to0^+$, as desired, and therefore completes the proof of the theorem.   
\end{proof}

\subsection{Exact Pointwise Tube Formula}\label{subs_exact}

In this subsection, we show that in the case of a strongly languid relative fractal drum, we are able to obtain an exact pointwise tube formula; that is, a pointwise formula without an error term.
This is the main content of the following theorem.

\begin{theorem}[Exact pointwise fractal tube formula via $\widetilde{\zeta}_{A,\O}$]\label{str_pointwise_formula}
Let $(A,\O)$ be a relative fractal drum in $\eR^N$ which is strongly languid for some fixed $\d>0$ and some fixed exponent $\kappa\in\eR$.
Furthermore, let $k>\kappa$ be a nonnegative integer.
Then, the following exact pointwise fractal tube formula, expressed in terms of the tube zeta function $\widetilde{\zeta}_{A,\O}:=\widetilde{\zeta}_{A,\O}(\,\cdot\,;\d)$, holds for all $t\in(0,\min\{1,\d,B^{-1}\})$$:$
\begin{equation}\label{point_form_w}
V^{[k]}_{A,\O}(t)=\sum_{\omega\in\po(\widetilde{\zeta}_{A,\O},\Ce)}\res\left(\frac{t^{N-s+k}}{(N\!-\! s\!+\! 1)_k}\widetilde{\zeta}_{A,\O}(s),\omega\right).
\end{equation}
Here, $B$ is the positive constant appearing in hypothesis {\bf L2'}.
\end{theorem}

\begin{proof}
We begin by fixing an integer $n\geq 1$ and applying Lemma~\ref{trunc_point} with the screen $\bm{S}_m$ given by hypothesis {\bf L2'}.
Next, we proceed by letting $m\to\ty$ while keeping $n$ fixed.
The fact that the screens $\bm{S}_m$ have a uniform Lipschitz bound implies that if we take $t<\min\{1,B^{-1}\}$, then the sequence of integrals over the truncated screens $(\bm{S}_{m|n})_{m\geq 1}$ converges to $0$ as $m\to\ty$.
(Here and throughout this proof, the truncated screen $\bm{S}_{m|n}$ denotes the $n$-th restriction of the screen $\bm{S}_m$, in the notation of Definition \ref{trunc_screen_window}.)
Indeed, to see this, let us take $m_0$ large enough so that $\sup S_m<0$ for every $m\geq m_0$.
This is possible since $\sup S_m\to -\ty$ as $m\to\ty$ by hypothesis {\bf L2'} of Definition \ref{str_languid}.

Furthermore, note that for every $m\geq 1$ and $n\geq 1$, the integral over the truncated screen $\bm{S}_{m|n}$ is given by
\begin{equation}
I_{n,m}:=\frac{1}{2\pi\I}\int_{\bm{S}_{m|n}}\frac{t^{N-s+k}}{(N\!-\! s\!+\! 1)_k}\widetilde{\zeta}_{A,\O}(s)\di s
\end{equation}
and, similarly as in the proof of Lemma~\ref{trunc_point}, we have that the integrand is bounded in absolute value by
\begin{equation}
C_{\kappa}\max\big\{B^{|\inf S_{m|n}|},B^{|\sup S_{m|n}|}\big\}{t^{N+|\sup S_{m|n}|+k}},
\end{equation}
where $C_{\kappa}$ is a suitable constant depending only on $\kappa$.
Here, we use the notation 
\begin{equation}
\inf S_{m|n}:=\inf_{\tau\in[T_{-n},T_n]}S_m(\tau)\ \ \textrm{and}\ \ \sup S_{m|n}:=\sup_{\tau\in[T_{-n},T_n]}S_m(\tau).
\end{equation}
We now let $L:=\sup_{m\geq 1}\|S_m\|$ be the uniform Lipschitz bound for the sequence of screens $\bm{S}_m$.
Then, the derivative of $\tau\mapsto S_m(\tau)+\I\tau$ is bounded for almost every $\tau\in[T_{-n},T_n]$ by $(1+L)$.

We must next consider the following two cases: firstly, if $B<1$, we then have that
$$
|I_{n,m}|\leq\frac{C_{\kappa}(1+L)B^{|\sup S_{m|n}|}}{2\pi}(T_n-T_{-n}){t^{N+|\sup S_{m|n}|+k}},
$$
and, since $t<1$, we have that $I_{n,m}\to 0$ as $m\to\ty$.
Secondly, if $B\geq 1$, we deduce from the Lipschitz condition on $S_m$ that we have
$$
\sup S_{m|n}-\inf S_{m|n}\leq L(T_n-T_{-n});
$$
i.e.,
$$
|\inf S_{m|n}|\leq |\sup S_{m|n}|+L(T_n-T_{-n}),
$$
from which we deduce the estimate
$$
|I_{n,m}|\leq\frac{C_{\kappa}(1+L)B^{L(T_n-T_{-n})}}{2\pi}(T_n-T_{-n})(Bt)^{|\sup S_{m|n}|}{t^{N+k}}.
$$
Therefore, $I_{n,m}\to 0$ as $m\to\ty$ since $Bt<1$.

We now let $E_{n,m}(t)$ be the error function appearing in~\eqref{trunc_point_formula} for the truncated screen $\bm{S}_{m|n}$ and we finalize the proof by showing that its iterated limit converges to zero pointwise.
Namely, for $c\in({\ov{\dim}_B(A,\O)},N+1)$ and since $0<t<1$, we have, much as in the proof of Lemma~\ref{trunc_point}, that
\begin{equation}
\begin{aligned}
\left|\int_{\Gamma_{U_m}}\frac{t^{N-s+k}\widetilde{\zeta}_{A,\O}(s)\di s}{(N\!-\! s\!+\! 1)_k}\right|&\leq  t^{N+k}C(T_n+1)^{\kappa}T_n^{-k}\int_{-\ty}^{c}t^{-\sigma}\di\sigma\leq  \frac{t^{N-c+k}K_\kappa T_n^{\kappa-k}}{\log t^{-1}}.
\end{aligned}
\vspace{6pt} 
\end{equation}
Here, $\Gamma_{U_m}$ is the segment connecting $S_m(T_n)+\I T_n$ and $c+\I T_n$.
A similar reasoning for the corresponding integral over the lower segment gives us the following upper bound on $|E_{n,m}(t)|$, independent of $m$:
$$
|E_{n,m}(t)|\leq\frac{t^{N-c+k}}{\pi\log t^{-1}}K_\kappa \max\{T_n^{\kappa-k},|T_{-n}|^{\kappa-k}\}.
$$
Finally, this inequality, which is valid for all $m\geq 1$ and all $n\geq 1$, implies that for a fixed $k>\kappa$, the iterated limit of $E_{n,m}(t)$ tends to $0$ when $m\to\ty$ and then $n\to\ty$; i.e., we have
$
\lim_{n\to\ty}\big(\lim_{m\to\ty}E_{n,m}(t)\big)=0.
$
This concludes the proof of the theorem.
\end{proof}

Of course, Theorems~\ref{pointwise_formula} and~\ref{str_pointwise_formula} are of most interest in the case when $k=0$, i.e., when we obtain a pointwise formula for the volume of the relative $t$-neighborhood $|A_t\cap\O|$ in terms of the complex dimensions of $(A,\O)$.
In that case, the sum over the (visible) complex dimensions of $(A,\O)$ takes the simpler form
\begin{equation}\label{3.17.1/4}
\sum_{\omega\in\po(\widetilde{\zeta}_{A,\O},\bm{W})}\res\left(t^{N-s}\widetilde{\zeta}_{A,\O}(s),\omega\right)
\end{equation}
or
\begin{equation}\label{3.17.1/2}
\sum_{\omega\in\po(\widetilde{\zeta}_{A,\O},\Ce)}\res\left(t^{N-s}\widetilde{\zeta}_{A,\O}(s),\omega\right)
\end{equation}
in Equation \eqref{point_form} of Theorem \ref{pointwise_formula} or in Equation \eqref{point_form_w} of Theorem \ref{str_pointwise_formula}, respectively.
Observe that in the important special case considered in Remarks \ref{5.1.13.1/4} and \ref{5.1.13.3/4} below when all of the (visible) complex dimensions of the RFD are simple, then Equation \eqref{3.17.1/4} becomes
\begin{equation}\label{3.17.3/4}
\sum_{\omega\in\po(\widetilde{\zeta}_{A,\O},\bm{W})}\res\left(\widetilde{\zeta}_{A,\O},\omega\right)t^{N-\omega}
\end{equation}
(much as in the Equation \eqref{ttt}, where $\zeta_{A,\O}$ is used instead of $\widetilde{\zeta}_{A,\O}$), while Equation \eqref{3.17.1/2} naturally becomes
\begin{equation}\label{3.17.4/5}
\sum_{\omega\in\po(\widetilde{\zeta}_{A,\O},\Ce)}\res\left(\widetilde{\zeta}_{A,\O},\omega\right)t^{N-\omega}.
\end{equation}
An analogous comment applies to all the fractal tube formulas obtained in Sections \ref{sec_distr} and \ref{distance_tube} below (see, especially, Theorems \ref{dist_error}, \ref{dist_no_error}, \ref{pointwise_formula_rd} and \ref{str_pointwise_formula_d_bez}).
In the case of the distance zeta function $\zeta_{A,\O}$ (instead of $\widetilde{\zeta}_{A,\O}$), this is so provided $\ov{\dim}_B(A,\O)<N$; furthermore, in that case, $t^{N-s}$ (resp., $t^{N-\omega}$) should be replaced by $\frac{t^{N-s}}{N-s}$ (resp., $\frac{t^{N-\omega}}{N-\omega}$) in the counterpart for $\zeta_{A,\O}$ of Equations \eqref{3.17.1/4} and \eqref{3.17.1/2} (resp., \eqref{3.17.3/4} and \eqref{3.17.4/5}).

%
%

\begin{remark}\label{5.1.13.1/4}
We point out that in the applications, the common situation is when all of the visible complex dimensions are simple.
More specifically, if we assume that all of the poles of $\widetilde{\zeta}_{A,\O}$ visible through the window $\bm W$ (i.e., lying in $\bm W$) are simple, then in the statement of Theorem~\ref{pointwise_formula} (when $k=0$ in the statement of that theorem), the sum over the visible complex dimensions appearing in Equation \eqref{point_form} reduces to the following expression:
\begin{equation}\label{5.1.40.1/4}
\sum_{\omega\in\po(\widetilde{\zeta}_{A,\O},\bm{W})}\widetilde{c}_{\omega}t^{N-\omega},
\end{equation}
where for each $\omega\in\po(\widetilde{\zeta}_{A,\O},\bm{W})$, we have
$
\widetilde{c}_{\omega}:=\res(\widetilde{\zeta}_{A,\O},\omega).
$
\end{remark}

\begin{remark}\label{5.1.13.3/4}
We also note that, in light of Theorem \ref{pointwise_formula} and Theorem \ref{str_pointwise_formula}, 
the counterpart of Remark \ref{5.1.13.1/4} 
holds for any level $k$ (satisfying the assumptions of the relevant result).
For example, provided that all of the complex dimensions visible through $\bm W$ are simple, the exact pointwise fractal tube formula \eqref{point_form_w} of Theorem \ref{str_pointwise_formula} becomes (for all $t\in(0,\min\{1,\d,B^{-1}\})$)
\begin{equation}\label{V_k_simple}
V_{A,,\O}^{[k]}(t)=\sum_{\omega\in\po(\widetilde{\zeta}_{A,\O},\Ce)}\res\left(\widetilde{\zeta}_{A,\O}(s),\omega\right)\frac{t^{N-\omega+k}}{(N-\omega+1)_k},
\end{equation}
and similarly for the pointwise fractal tube formula with error term given in \eqref{point_form} of Theorem \ref{pointwise_formula}.

Note that in light of \eqref{pochamer} and for each $k\in\eN_0$, we have (with the obvious convention if $k=0$)
\begin{equation}\label{N-s+1}
(N-s+1)_k=(N-s+1)(N-s+2)\cdots(N-s+k)
\end{equation}
and hence, the zeros of the polynomial function $s\mapsto(N-s+1)_k$ are simple and occur precisely at 
$
s=N+1,N+2,\ldots,N+k.
$
(Clearly, since $(N-s+1)_0=1$, \eqref{N-s+1} does not have any zeros if $k=0$.)
Consequently, since $\ov{\dim}_B(A,\O)\leq N$ and $k$ is nonnegative (i.e., $k\in\eN_0$) in the present case of pointwise tube formulas, the complex number $(N-\omega+1)_k$ is never equal to zero for $\omega\in\po(\widetilde{\zeta}_{A,\O}):=\po(\widetilde{\zeta}_{A,\O},\Ce)$ (or else for $\omega\in\po(\widetilde{\zeta}_{A,\O},\bm{W})$, in the case of a pointwise tube formula with error term).
Moreover, if we work with a distributional tube formula (as will be case in Section \ref{sec_distr} and part of Section \ref{distance_tube}, for example), the level $k$ is allowed to be negative (i.e., $k\in\Ze$).
However, in the case of a negative integer $k$, the function $s\mapsto(N-s+1)_k$ does not have any zeros, but only simple poles located precisely at 
$
s=N+1+k,N+2+k,\ldots,N;
$
so that its reciprocal has simple zeros precisely at those same points.
%
Therefore, we note that in the distributional case, it may happen that $\omega$ is a zero of $s\mapsto(N-s+1)_k^{-1}$, which in that case will cancel out the term corresponding to $t^{N-\omega}$ in Equation \eqref{V_k_simple}.
\end{remark}

\begin{remark}\label{5.1.17}
As was alluded to above, the obvious counterpart of Remark \ref{5.1.13.1/4} 
and Remark \ref{5.1.13.3/4} holds for all of the fractal tube formulas considered in this paper, whether they are pointwise or distributional formulas, with or without error term, as well as expressed in terms of either $\zeta_{A,\O}$ or $\widetilde{\zeta}_{A,\O}$ or (with the notation of Subsection \ref{subsec_shell}) $\breve{\zeta}_{A,\O}$.
In the case of $\zeta_{A,\O}$ and $\breve{\zeta}_{A,\O}$, one must assume, in addition, that $\ov{D}:=\dim_B(A,\O)<N$.
\end{remark}

\section{Distributional Fractal Tube Formula}\label{sec_distr}

In this section, our goal is to weaken the languidity conditions imposed on the tube zeta function and still obtain a fractal tube formula expressed in terms of $\widetilde{\zeta}_{A,\O}$.
More precisely, if we want to relax the condition on the languidity exponent $\kappa$, we will still obtain a fractal tube formula but only in the sense of Schwartz distributions.
In other words, we will establish the distributional analogs of Theorems~\ref{pointwise_formula} and~\ref{str_pointwise_formula} in order to derive a distributional fractal tube formula for $V^{[k]}_{A,\O}(t)$, valid for any integer $k\in\Ze$ and still expressed in terms of the (visible) poles of the tube zeta function $\widetilde{\zeta}_{A,\O}$.
This will provide us with asymptotic information (in the sense of Schwartz distributions or generalized functions) about the tube function of a relative fractal drum $(A,\O)$, independently of for which exponent $\kappa\in\eR$ the relative fractal drum $(A,\O)$ is languid.
(See Definition~\ref{languid}.)  
More precisely, let $\delta>0$ and define $\mathcal{D}(0,\d):=C_c^{\ty}(0,\delta)$\label{Dspace} to be the space of infinitely differentiable (complex-valued) test functions with compact support contained in $(0,\d)$.
Actually, let us introduce a larger space of test functions for which the formulas obtained here will be valid.
Namely, let $\mathcal{K}(0,\delta)$\label{Kspace} be the set of test functions $\varphi$ in the class $C^{\ty}(0,\d)$, such that for all $m\in\Ze$ and $q\in\eN$, we have $t^m\varphi^{(q)}(t)\to 0$, as $t\to 0^+$ and also that $(t-\d)^m\varphi^{(q)}(t)\to 0$ as $t\to\d^-$, where $\varphi^{(q)}$ denotes the $q$-th derivative of $\varphi$.

Note that $\mathcal{D}(0,\d)\subseteq\mathcal{K}(0,\d)$.\label{D_K}
Hence, we have the following (reverse) inclusion between the corresponding spaces of distributions (i.e., the dual spaces):
\begin{equation}\label{5.2.1/2}
\mathcal{K}'(0,\d)\subseteq\mathcal{D}'(0,\d).
\end{equation}

General information about the theory of distributions (or generalized functions) can be found in \cite{Schw,Bre,folland,ho2,johlap,johlapni,resi1}.

\begin{definition}\label{Vdist}
Let $(A,\O)$ be a relative fractal drum in $\eR^N$ and let $k\in\Ze$ be an arbitrary integer.
We define the distribution $\mathcal{V}^{[k]}=\mathcal{V}_{A,\O}^{[k]}$ on $\mathcal{K}(0,\d)$ to be the $|k|$-th distributional derivative of $V(t)=|A_t\cap\O|$ in case $k<0$ and the $k$-th primitive (or $k$-th anti-derivative) function (considered as a regular distribution in $\mathcal{K}'(0,\d)$) of $V(t)$ if $k>0$.
For $k=0$, this is the (regular) distribution generated by the locally integrable function $V(t)$.
(Note that the local integrability of $V=V(t)$ on $(0,+\ty)$ follows from its continuity.)
More specifically, for any test function $\varphi\in\mathcal{K}(0,\d)$, we have 
\begin{equation}\label{V_k}
\langle\mathcal{V}^{[k]},\varphi\rangle:=\int_0^{+\ty}V^{[k]}(t)\varphi(t) \di t,\quad\textrm{ for } k\geq 0,
\end{equation}
and
\begin{equation}\label{V_-k}
\langle\mathcal{V}^{[k]},\varphi\rangle:=(-1)^{|k|}\int_0^{+\ty}V(t)\varphi^{(|k|)}(t) \di t,\quad\textrm{ for } k<0.
\end{equation}
Here and thereafter, for convenience, we always extend the test function $\varphi\in\mathcal{K}(0,\d)$ to the interval $[\delta,+\ty)$ by letting $\varphi_{|[\d,+\ty)}\equiv 0$.
\end{definition}


Let now $\varphi\in\mathcal{K}(0,\d)$ be a test function.
The decay conditions on $\varphi$ imply that $t^s\varphi(t)$ is integrable on $(0,\d)$ for every $s\in\Ce$ and that its Mellin transform $\{\mathfrak{M}\varphi\}(s)$ is an entire function.
This follows directly from a general result about the holomorphicity of an integral depending analytically on a parameter (see, e.g.,~\cite[Theorem~31]{titch} or \cite[Theorem 2.1.46]{fzf}).

Furthermore, let $g(s)$ be a meromorphic function.
Then, the residue $\res(g(s),\omega)$ vanishes unless $\omega$ is a pole of $g$.
Moreover, for all $k\in\Ze$, $N\in\eN$ and by choosing a suitable closed contour $\Gamma$ around $\omega$, we have
$$
\begin{aligned}
\int_0^{+\ty}\varphi(t)\res(t^{N-s+k}g(s),\omega)\di t&=\int_0^{+\ty}\varphi(t)\frac{1}{2\pi\I}\oint_{\Gamma}t^{N-s+k}g(s)\di s\di t\\
&=\frac{1}{2\pi\I}\oint_{\Gamma}g(s)\int_0^{+\ty}t^{N-s+k}\varphi(t)\di t\di s\\
&=\res\big(\{\mathfrak{M}\varphi\}(N\!-\! s\!+\! k\!+\! 1)\,g(s),\omega\big). 
\end{aligned}
$$
The change of the order of integration is justified by the Fubini--Tonelli theorem since the last integral above is absolutely convergent.
In short, for every $\varphi\in\mathcal{K}(0,\d)$, we have
\begin{equation}\label{res_res}
\Big\langle\res\big(t^{N-s+k}g(s),\omega\big),\varphi\Big\rangle=\res\Big(\{\mathfrak{M}\varphi\}(N\!-\! s\!+\! 1\!+\! k)\,g(s),\omega\Big),
\end{equation}
where $g(s)$ is a meromorphic function on a connected open neighborhood of $\omega\in\Ce$ and where $k\in\Ze$ and $N\in\eN$.

\begin{remark}\label{4.1.1/2}
We refer to our earlier Remark \ref{2.16.1/2} for an explanation of the usefulness (both conceptually and technically) of working with any $k\in\eN\cup\{0\}$, in the pointwise case, and any $k\in\Ze$, in the present distributional case.

As was already alluded to in that remark, particular attention should be paid to the case when $k=-1$.
Indeed, observe that for $k=-1$, the distribution $\mathcal{V}^{[-1]}=\mathcal{V}_{A,\O}^{[-1]}$ can be viewed as a (positive) {\em measure} on $(0,+\ty)$; indeed, it is the distributional derivative of the nondecreasing and locally integrable function $t\mapsto V(t)=V_{A,\O}(t)$ on $(0,+\ty)$.
By analogy with the special case of fractal strings discussed in \cite[Subsection 6.3.1]{lapidusfrank12}, we call it the {\em geometric density of $($volume$)$ states} of the RFD $(A,\O)$.
(Compare with \cite{lapidusfrank12} and the relevant references therein about the mathematical and theoretical physics literature on spectral theory, semiclassical approximation and quantum mechanics.)
From a fundamental point of view, this measure $\mathcal{V}^{[-1]}=\mathcal{V}_{A,\O}^{[-1]}$ is the most important `distributional tube function' and the corresponding fractal tube formulas the most useful distributional tube formulas.

We leave it as a simple exercise for the reader to write explicitly the corresponding distributional tube formula at level $k=-1$ as a mere corollary of Theorem \ref{dist_error} (in Section \ref{subs_distr}) and Theorem \ref{dist_no_error} (in Section \ref{sec_ex_dis}) below, as well as of other distributional tube formulas obtained in this paper, and to also compare their general form with the corresponding results in \cite[Subsection 6.3.1]{lapidusfrank12} obtained for the geometry and spectra of fractal strings (i.e., when $N=1$).
\end{remark}

\subsection{Distributional Tube Formula with Error Term}\label{subs_distr}

After having introduced the necessary preliminaries just above, we are now able to state the distributional analog of Theorem~\ref{pointwise_formula}; that is, the distributional tube formula with an error term.

\begin{theorem}[Distributional fractal tube formula with error term, via $\widetilde{\zeta}_{A,\O}$]\label{dist_error}
Let $(A,\O)$ be a relative fractal drum in $\eR^N$ which is languid for some $\kappa\in\eR$ and $\d>0$.
Then, for every $k\in\Ze$, the distribution $\mathcal V^{[k]}_{A,\O}$ in $\mathcal{K}'(0,\d)$ $($and hence, also in $\mathcal{D}'(0,\d)$$)$ is given by the following distributional fractal tube formula, with error term and expressed in terms of the tube zeta function $\widetilde{\zeta}_{A,\O}:=\widetilde{\zeta}_{A,\O}(\,\cdot\,;\d)$$:$
\begin{equation}\label{dist_form_error}
\mathcal V^{[k]}_{A,\O}(t)=\sum_{\omega\in\po(\widetilde{\zeta}_{A,\O},\bm{W})}\res\left(\frac{t^{N-s+k}}{(N\!-\! s\!+\! 1)_k}\widetilde{\zeta}_{A,\O}(s),\omega\right)+\widetilde{\mathcal R}^{[k]}_{A,\O}(t).
\end{equation}
That is, the action of $\mathcal V^{[k]}_{A,\O}$ on an arbitrary test function $\varphi\in\mathcal{K}(0,\d)$ is given by
\begin{equation}\label{error_action}
\begin{aligned}
\big\langle\mathcal V^{[k]}_{A,\O},\varphi\big\rangle&=\sum_{\omega\in\po(\widetilde{\zeta}_{A,\O},\bm{W})}\res\left(\frac{\{\mathfrak{M}\varphi\}(N\!-\! s\!+\! 1\!+\! k)\,\widetilde{\zeta}_{A,\O}(s)}{(N\!-\! s\!+\! 1)_k},\omega\right)+\big\langle\widetilde{\mathcal R}^{[k]}_{A,\O},\varphi\big\rangle.
\end{aligned}
\end{equation}
Here, the $($distributional$)$ error term $\widetilde{\mathcal R}^{[k]}_{A,\O}$ is given by the distribution in $\mathcal{K}'(0,\d)$ defined for all test functions $\varphi\in\mathcal{K}(0,\d)$ by
\begin{equation}\label{R_distr}
\big\langle\widetilde{\mathcal R}^{[k]}_{A,\O},\varphi\big\rangle=\frac{1}{2\pi\I}\int_{\bm S}\frac{\{\mathfrak{M}\varphi\}(N\!-\! s\!+\! 1\!+\! k)\,\widetilde{\zeta}_{A,\O}(s)}{(N\!-\! s\!+\! 1)_k}\di s.
\end{equation}
$($The corresponding distributional error estimate for $\widetilde{\mathcal R}^{[k]}_{A,\O}$ will be given in Theorem \ref{dist_estimate} of Subsection \ref{subsec_error} below.$)$
\end{theorem}

\begin{proof}
We begin the proof by fixing $k\in\eN_0$ such that $k>\kappa+1$ and a constant $c\in(\ov{\dim}_B(A,\O),N+1)$.
Note that by fixing $c\in(\ov{\dim}_B(A,\O),N+1)$, we have ensured that none of the poles of $(N\!-\! s\!+\! 1)_k^{-1}$ is located in the window $\bm W$.
Indeed, according to the discussion provided in Remark \ref{5.1.13.3/4}, the set of poles of $(N\!-\! s\!+\! 1)_k^{-1}$
is a subset of $\{N+n:n\in\eN\}$.
Then, for every test function $\varphi\in\mathcal{K}(0,\d)$, we have successively:
\begin{equation}\label{5.2.8.1/2}
\begin{aligned}
\big\langle V^{[k]}_{A,\O},\varphi\big\rangle&=\int_0^{+\ty}V^{[k]}_{A,\O}(t)\varphi(t)\di t=\frac{1}{2\pi\I}\int_{c-\I\ty}^{c+\I\ty}\int_0^{+\ty}\varphi(t)\frac{t^{N-s+k}\,\widetilde{\zeta}_{A,\O}(s)}{(N\!-\! s\!+\! 1)_{k}}\di t\di s\\
&=\frac{1}{2\pi\I}\int_{c-\I\ty}^{c+\I\ty}\frac{\{\mathfrak{M}\varphi\}(N\!-\! s\!+\! 1\!+\! k)\,\widetilde{\zeta}_{A,\O}(s)}{(N\!-\! s\!+\! 1)_{k}}\di s.\\
\end{aligned}
\vspace{6pt}
\end{equation}
Here, the change of the order of integration in the second equality of \eqref{5.2.8.1/2} is justified by the Fubini--Tonelli theorem since the first integral above is absolutely convergent.
(It is easy to see that $|V^{[k]}(t)|\leq|A_t|t^{k}$, for all $t\in(0,+\ty)$ and $k\geq 0$.)
One can now approximate the last integral in \eqref{5.2.8.1/2} above in the same way as in Lemma~\ref{trunc_point}; that is, we approximate it by the following expression:
\begin{equation}
\begin{aligned}
&\sum_{\omega\in\po(\widetilde{\zeta}_{A,\O},{\bm W}_{|n})}\res\left(\frac{\{\mathfrak{M}\varphi\}(N\!-\! s\!+\! 1\!+\! k)\,\widetilde{\zeta}_{A,\O}(s)}{(N\!-\! s\!+\! 1)_{k}}\right)\\
&+\frac{1}{2\pi\I}\int_{\bm{S}_{|n}}\frac{\{\mathfrak{M}\varphi\}(N\!-\! s\!+\! 1\!+\! k)\,\widetilde{\zeta}_{A,\O}(s)}{(N\!-\! s\!+\! 1)_{k}}\di s\\
&+\frac{1}{2\pi\I}\int_{\Gamma_L\cup\Gamma_U}\frac{\{\mathfrak{M}\varphi\}(N\!-\! s\!+\! 1\!+\! k)\,\widetilde{\zeta}_{A,\O}(s)}{(N\!-\! s\!+\! 1)_{k}}\di s.\\
\end{aligned}
\end{equation}
Furthermore, in light of~\eqref{res_res}, the latter expression is equal to
\begin{equation}
\begin{aligned}
&\sum_{\omega\in\po(\widetilde{\zeta}_{A,\O},W_{|n})}\left\langle\res\left(\frac{t^{N-s+k}}{(N\!-\! s\!+\! 1)_{k}}\widetilde{\zeta}_{A,\O}(s),\omega\right),\varphi\right\rangle\\
&+\frac{1}{2\pi\I}\int_{\bm{S}_{|n}}\frac{\{\mathfrak{M}\varphi\}(N\!-\! s\!+\! 1\!+\! k)\,\widetilde{\zeta}_{A,\O}(s)}{(N\!-\! s\!+\! 1)_{k}}\di s+\int_0^{+\ty}E_n(t)\varphi(t)\di t,\\
\end{aligned}
\end{equation}
where the error term $E_n(t)$ is given as in Lemma \ref{trunc_point} and its proof.

Next, by letting $n\to\ty$, we deduce by the same argument as in Theorem~\ref{pointwise_formula} that the integral $\int_0^{+\ty}E_n(t)\varphi(t)\di t$ tends to zero and, similarly, we show that the above integral over the truncated screen $\bm{S}_{|n}$ converges absolutely.
Thus, we deduce that
\begin{equation}\label{V_dist}
\begin{aligned}
\big\langle V^{[k]}_{A,\O},\varphi\big\rangle&=\!\!\!\sum_{\omega\in\po(\widetilde{\zeta}_{A,\O},\bm{W})}\!\!\!\res\left(\frac{\{\mathfrak{M}\varphi\}(N\!-\! s\!+\! 1\!+\! k)\,\widetilde{\zeta}_{A,\O}(s)}{(N\!-\! s\!+\! 1)_{k}},\omega\right)+\langle\widetilde{\mathcal R}^{[k]}_{A,\O},\varphi\rangle,
\end{aligned}
\vspace{6pt}
\end{equation}
where $\widetilde{\mathcal R}^{[k]}_{A,\O}$ is given by its action on test functions as shown in Equation \eqref{R_distr}.

Moreover, observe that the expression on the right-hand side of \eqref{V_dist} defines a distribution in $\mathcal{K}'(0,\d)$ (since $V^{[k]}_{A,\O}$ is locally integrable).
This concludes the proof of the theorem in the case when $k>\max\{-1,\kappa+1\}$.

In the case when $k\leq\kappa+1$ and $k\in\Ze$, we choose an integer $q$ such that $k+q>\max\{\kappa+1,-1\}$ and note that by the definition of the distributional derivative (or alternatively, in light of Equations \eqref{V_k} and \eqref{V_-k} defining $\mathcal{V}_{A,\O}^{[k]}$), we have that
\begin{equation}\label{dist_der_V}
\big\langle\mathcal V^{[k]}_{A,\O},\varphi\big\rangle=(-1)^q\big\langle\mathcal V^{[k+q]}_{A,\O},\varphi^{(q)}\big\rangle.
\end{equation}
Finally, in order to complete the proof, we use identity \eqref{dist_der_V} together with~\eqref{V_dist} applied at level $k+q$, 
along with the following well-known (and easy to verify) fact about the Mellin transform (see Equation \eqref{ddeff_mell} defining $\{\mathfrak{M}f\}(s)$):
\begin{equation}
\{\mathfrak{M}\varphi\}(s)=\frac{(-1)^q}{(s)_q}\big\{\mathfrak{M}\varphi^{(q)}\big\}(s+q),
\end{equation}
for all $s\in\Ce$ and $q\in\Ze$.
We therefore deduce that \eqref{error_action} holds, with $\big\langle\widetilde{\mathcal R}^{[k]}_{A,\O},\varphi\big\rangle$ given by \eqref{R_distr}, as required.
This concludes the proof of the theorem.
\end{proof}

\begin{remark}
Note the above proof of Theorem~\ref{dist_error} establishes the fact that the sum over the (visible) complex dimensions appearing in~\eqref{dist_form_error} defines a distribution in $\mathcal{K}'(0,\d)$ (since it is a difference of two distributions in $\mathcal{K}'(0,\d)$) and hence, according to the inclusion \eqref{5.2.1/2}, also in $\mathcal{D}'(0,\d)$.
In turn, this fact implies that both terms on the right-hand side of~\eqref{dist_form_error} are, on their own, distributions in $\mathcal{K}'(0,\d)$.
Namely, this is a consequence of a well-known fact about the convergence of distributions, which itself follows from a suitable generalization of the Hahn--Banach theorem to locally convex topological spaces (see, for example, \cite[Theorem~2.1.8, p.~39]{ho2}):
%
%
%
%

An entirely analogous comment applies to Theorem~\ref{dist_no_error} below, with the space of test functions now coinciding with $\mathcal{D}(0,\d_0)$ and thus the associated space of distributions being equal to $\mathcal{D}'(0,\d_0)$.
\end{remark}

\subsection{Exact Distributional Tube Formula}\label{sec_ex_dis}

The main result of this subsection is a distributional analog of the pointwise tube formula without error term stated in Theorem \ref{str_pointwise_formula} of Subsection \ref{subs_exact}; see Theorem \ref{dist_no_error} just below.
The resulting fractal tube formula will be an asymptotic distributional formula meaning that it will be valid for test functions in $\mathcal{K}(0,\d)$ that are supported on the left of $B^{-1}$, where $B>0$ is the constant appearing in hypothesis {\bf L2'}.

\begin{theorem}[Exact distributional tube formula via $\widetilde{\zeta}_{A,\O}$]\label{dist_no_error}
Let $(A,\O)$ be a relative fractal drum in $\eR^N$ which is strongly languid for some $\kappa\in\eR$ and $\d>0$.
Furthermore, let $\d_0:=\min\{1,\d,B^{-1}\}$.
Then, for every $k\in\Ze$, the distribution $\mathcal V^{[k]}_{A,\O}$ in $\mathcal{D}'(0,\d_0)$ is given by the following exact distributional tube formula in $\mathcal{D}'(0,\d_0)$, expressed in terms of the tube zeta function $\widetilde{\zeta}_{A,\O}:=\widetilde{\zeta}_{A,\O}(\,\cdot\,;\d)$$:$
\begin{equation}\label{dist_form_no_error}
\mathcal V^{[k]}_{A,\O}(t)=\sum_{\omega\in\po(\widetilde{\zeta}_{A,\O},\Ce)}\res\left(\frac{t^{N-s+k}}{(N\!-\! s\!+\! 1)_k}\widetilde{\zeta}_{A,\O}(s),\omega\right).
\end{equation}
That is, the action of $\mathcal V^{[k]}_{A,\O}$ on an arbitrary test function $\varphi\in\mathcal{D}(0,\d_0)$ is given by
\begin{equation}\label{no_error_action}
\begin{aligned}
\big\langle\mathcal V^{[k]}_{A,\O},\varphi\big\rangle&=\sum_{\omega\in\po(\widetilde{\zeta}_{A,\O},\Ce)}\res\left(\frac{\{\mathfrak{M}\varphi\}(N\!-\! s\!+\! 1\!+\! k)\,\widetilde{\zeta}_{A,\O}(s)}{(N\!-\! s\!+\! 1)_k},\omega\right).\\
\end{aligned}
\end{equation}
\end{theorem}

\begin{proof}
The theorem is proved by applying Theorem~\ref{dist_error} to the sequence of screens $\bm{S}_m$ (occurring in hypothesis {\bf L2'} of strong languidity, see Equation \eqref{L2'})  and then showing that the corresponding error term tends to zero as $m\to\ty$.
More specifically, by choosing $q\in\eN$ such that $k+q>\kappa +1$ and $m\in\eN$ such that $\sup S_m<0$, we deduce from~\eqref{dist_form_error} the following distributional identity, viewed as an equality in $\mathcal{D}'(0,\d_0)$:
\begin{equation}\label{dist_form_error_m}
\begin{aligned}
\mathcal V^{[k+q]}_{A,\O}(t)&=\sum_{\omega\in\po(\widetilde{\zeta}_{A,\O},\bm W_{m})}\res\left(\frac{t^{N-s+k+q}}{(N\!-\! s\!+\! 1)_{k+q}}\widetilde{\zeta}_{A,\O}(s),\omega\right)+\widetilde{\mathcal R}_m^{[k+q]}(t).\\
\end{aligned}
\end{equation}
Next, we fix a test function $\varphi\in\mathcal{D}(0,\d_0)$.
Since by definition, $\varphi$ has compact support, there exists $\nu\in(0,1)$ such that the support of $\varphi$ is contained in $(0,\nu B^{-1}]$.
Using this fact, we estimate the Mellin transform of $\varphi$ in the following way, for all $s\in\Ce$ such that $\re s<0$:
\begin{equation}\label{estm2}
\big|\{\mathfrak{M}\varphi\}(N\!-\! s\!+\! 1\!+\! k\!+\! q)\big|\leq(\nu B^{-1})^{-\re s}\int_0^{+\ty}t^{N+k+q}|\varphi(t)|\di t.
\end{equation}
By using the above estimate \eqref{estm2}, hypothesis {\bf L2'}, along with the obvious fact that
$
|(N\!-\! S_m(\tau)\!-\! \I\tau\!+\! 1)_{k+q}|\geq\big(\sqrt{1+\tau^2}\,\big)^{k+q},
$
we now estimate the distributional error term $\widetilde{\mathcal R}_m^{[k+q]}$ as follows (we let $|\di s|:=|s'(\tau)|\di\tau$):
\begin{equation}\label{5.2.13.1/2}
\begin{aligned}
\big|\big\langle\widetilde{\mathcal R}_m^{[k+q]},\varphi\big\rangle\big|&\leq\int_{\bm{S}_m}\big|\{\mathfrak{M}\varphi\}(N\!-\! s\!+\! 1\!+\! k\!+\! q)\big|\frac{|\widetilde{\zeta}_A(s)|}{|(N\!-\! s\!+\! 1)_{k+q}|}\,|\di s|\\
&\leq \widetilde{K}(1+\|S_m\|_{\mathrm{Lip}})\int_{-\ty}^{+\ty}(B\nu B^{-1})^{|S_m(\tau)|}\frac{(1+|\tau|)^{\kappa}}{\big(\sqrt{1+\tau^2}\,\big)^{k+q}}\di\tau\\
&\leq K\nu^{|\sup S_m|}\int_{-\ty}^{+\ty}\frac{(1+|\tau|)^{\kappa}}{\big(\sqrt{1+\tau^2}\,\big)^{k+q}}\di\tau,
\end{aligned}
\end{equation}
with $K$ being a suitable positive constant.
The last inequality follows since, according to hypothesis {\bf L2'}, the sequence of screens $(\bm{S}_m)_{m\geq 1}$ has a uniform Lipschitz bound; see the definition of strong languidity given in Definition \ref{str_languid}.
Furthermore, the last integral in the above calculation is convergent since $k+q>\kappa+1$.

Next, by letting $m\to\ty$, we deduce that $\langle\widetilde{\mathcal R}_m^{[k+q]},\varphi\rangle\to 0$ since $|\sup S_m|\to\ty$, and thus we conclude that $\widetilde{\mathcal R}_m^{[k+q]}\to 0$ as $m\to\ty$, in $\mathcal{D}'(0,\d_0)$.
Finally, in light of \eqref{dist_form_error_m}, we obtain the statement of the theorem for the distribution $\mathcal{V}^{[k+q]}_{A,\O}$ in $\mathcal{D}'(0,\d_0)$.
Finally, in order to complete the proof of the theorem, i.e., to obtain the statement for $\mathcal{V}^{[k]}_{A,\O}$ itself, we use the exact same argument as in the proof of Theorem~\ref{dist_error} in Subsection \ref{subs_distr} above.
\end{proof}


Of course, the most interesting special case of the distributional fractal tube formula (with and without an error term) is the case when $k=0$ and hence, $\mathcal{V}_{A,\O}^{[0]}(t)=|A_t\cap\O|$ for all $t>0$ (and as a regular distribution in $\mathcal{D}'(0,\d_0)$).

\subsection{Estimate for the Distributional Error Term}\label{subsec_error}

In this subsection, our goal is to give an asymptotic estimate for the distributional error term appearing in Theorem~\ref{dist_error}, interpreted in the sense of~\cite[Subsection~5.2.4]{lapidusfrank12}. 
In order to do so, we now introduce the notion of the distributional order of growth (see~\cite{EsKa,jm,PiStVi} and also, independently, [Lap-vFr1--2] and ~\cite[Definition~5.29]{lapidusfrank12}).

For a test function $\varphi\in\mathcal{D}(0,+\ty)$ and $a>0$, we let
\begin{equation}\label{varphi_a}
\varphi_a(t):=\frac{1}{a}\varphi\left(\frac{t}{a}\right).
\end{equation}
Observe that $\int_0^{+\ty}\varphi_a(t)\di t=\int_0^{+\ty}\varphi(t)\, \di t$, for every $a>0$.

\begin{definition}\label{dist_order_dis}
Let $\mathcal{R}$ be a distribution in $\mathcal D'(0,\d)$ and let $\alpha\in\eR$.
We say that $\mathcal R$ is of {\em asymptotic order} at most $t^{\alpha}$ (resp., less than $t^{\alpha}$) as $t\to 0^+$ if applied to an arbitrary test function $\varphi_a$ in $\mathcal{D}(0,\d)$, we have that\footnote{In this formula, the implicit constant may depend on the test function $\varphi$.}
\begin{equation}
\langle\mathcal R,\varphi_a\rangle=O(a^{\alpha})\quad(\textrm{resp., }\ \langle\mathcal R,\varphi_a\rangle=o(a^{\alpha})),\quad\textrm{as}\quad a\to 0^+.
\end{equation}
We then write that $\mathcal R(t)=O(t^{\alpha})$ (resp., $\mathcal R(t)=o(t^{\alpha})$),\quad as\quad $a\to 0^+$.
\end{definition}

\begin{remark}\label{asymp_dist}
Note that it is easy to see that if $f$ is a continuous function such that pointwise, $f(t)=O(t^{\alpha})$ or $f(t)=o(t^\a)$ as $t\to 0^+$, for some $\a\in\eR$, then $f$ also satisfies the same asymptotics, in the distributional sense of Definition~\ref{dist_order_dis}.
On the other hand, we note that clearly, a distributional asymptotic estimate (in the case of regular distributions), does not in general imply the usual pointwise one; see, e.g., \cite{PiStVi} where an explicit counterexample is given.
\end{remark}

Finally, also observe that for a test function $\varphi\in\mathcal{D}(0,\d)$ and $a>0$, the Mellin transform of $\varphi_a$ satisfies the following (see Equation \eqref{ddeff_mell} defining $\{\mathfrak{M}f\}(s)$):
\begin{equation}\label{mellin_a}
\{\mathfrak{M}\varphi_a\}(s)=a^{s-1}\{\mathfrak{M}\varphi\}(s),
\end{equation}
for all $s\in\Ce$.

\medskip

We now state the main result of this subsection, dealing with the order of growth of the distributional error term appearing in Theorem~\ref{dist_error}.
It is the analog in our present context of \cite[Theorem~5.30]{lapidusfrank12}.

\begin{theorem}[Estimate for the distributional error term]\label{dist_estimate}
Assume that the hypotheses of Theorem~\ref{dist_error} are satisfied, for a fixed $k\in\Ze$.
Then, the distribution $\widetilde{\mathcal R}^{[k]}_{A,\O}(t)$ given by~\eqref{R_distr} is of asymptotic order at most $t^{N-\sup S+k}$ as $t\to 0^+$; i.e,
\begin{equation}
\widetilde{\mathcal R}^{[k]}_{A,\O}(t)=O(t^{N-\sup S+k})\quad\mathrm{as}\quad t\to0^+,
\end{equation}
in the sense of Definition~\ref{dist_order_dis}.

Moreover, if $S(\tau) < \sup S$ for all $\tau\in\eR$ $($that is, if the screen $\bm S$ lies strictly
to the left of the vertical line $\{\re s =\sup S\})$, then $\widetilde{\mathcal R}^{[k]}_{A,\O}(t)$ is of asymptotic order less than $t^{N-\sup S+k}$; i.e.,
\begin{equation}\label{dist_estimate_o}
\widetilde{\mathcal R}^{[k]}_{A,\O}(t)=o(t^{N-\sup S+k})\quad\mathrm{as}\quad t\to0^+,
\end{equation}
also in the sense of Definition~\ref{dist_order_dis}.
\end{theorem}

\begin{proof}
For a test function $\varphi$, the integral defining $\langle\widetilde{\mathcal R}^{[k]}_{A,\O},\varphi\rangle$ in Equation \eqref{R_distr} converges absolutely.
Furthermore, for any $a\in(0,1)$, and by using~\eqref{mellin_a}, we obtain the following estimate:
$$
\begin{aligned}
\big|\big\langle\widetilde{\mathcal R}^{[k]}_{A,\O},\varphi_a\big\rangle\big|&\leq\frac{1}{2\pi}\int_S\frac{\big|\{\mathfrak{M}\varphi_a\}(N\!-\! s\!+\! 1\!+\!k)\big|}{|(N\!-\! s\!+\! 1)_k|}|\widetilde{\zeta}_{A,\O}(s)|\,|\di s|\\
&=\frac{1}{2\pi}\int_Sa^{N-\re s+k}\frac{\big|\{\mathfrak{M}\varphi\}(N\!-\! s\!+\! 1\!+\! k)\big|\,|\widetilde{\zeta}_{A,\O}(s)|}{|(N\!-\! s\!+\! 1)_k|}|\di s|.
\end{aligned}
$$
The last integral above is bounded by $Ca^{N-\sup S+k}$, where $C$ is a positive constant and this proves the first part of the theorem.

In order to establish the second part of the theorem, we use an argument similar to the one used in the proof of estimate~\eqref{S(t)<sup} of Theorem~\ref{pointwise_formula}.
\end{proof}

\section{Tube Formulas via the Relative Distance Zeta Function}\label{distance_tube}

In this section, our main goal is to reformulate the results from the previous sections in terms of the relative distance zeta function $\zeta_{A,\O}:=\zeta_{A,\O}(\,\cdot\,;\d)$.
This is extremely useful in the applications since the relative distance zeta function $\zeta_{A,\O}$ of an RFD $(A,\O)$, can be calculated without knowing its relative tube function $t\mapsto|A_t\cap\O|$ (which, of course, is not the case for the tube zeta function $\widetilde{\zeta}_{A,\O}$).
Naturally, the results will follow, in particular, from the functional equation~\eqref{equ_tilde} which connects these two fractal zeta functions, $\zeta_{A,\O}$ and $\widetilde{\zeta}_{A,\O}$.
More precisely, in order to derive the analogous results in terms of the distance zeta function, we will introduce a new fractal zeta function, called the {\em relative shell zeta function}, which satisfies a more direct functional equation, compared to \eqref{equ_tilde}.

For $A\subseteq\eR^N$ and $t,\d>0$ with $t\leq\d$, we let
\begin{equation}\label{shell_ozn}
A_{t,\d}:=A_\d\setminus \ov{A_t}.
\end{equation}
Note that $A_{t,\d}$ can be thought of as the $(t,\d)$-{\em shell} associated with $A$.
It was proved in \cite{stacho} that for any bounded set $A\subset\eR^N$ and every $t>0$, we have that $|\partial A_t|=0$, where $\pa A_t$ denotes the boundary of $A_t$ in $\eR^N$ and (as usual) $|\pa A_t|$ denotes its $N$-dimensional volume.
Since any unbounded set in $\eR^N$ may be partitioned into a countable union of bounded subsets, this also holds for unbounded subsets of $\eR^N$.
Consequently, for any relative fractal drum $(A,\O)$ in $\eR^N$, we have (for $0<t\leq\d$)
\begin{equation}\label{shell_tb}
|A_{t,\d}\cap\O|=|A_\d\cap\O|-|\ov{A_t}\cap\O|=|A_\d\cap\O|-|{A_t}\cap\O|.
\end{equation}

\subsection{The Relative Shell Zeta Function}\label{subsec_shell}

Let $\widetilde{\zeta}_{A,\O}(\,\cdot\,;\d)$ be the tube zeta function of the relative fractal drum $(A,\O)$ in $\eR^N$ and assume that $\re s>N$, then we have
\begin{equation}\label{racun}
\begin{aligned}
\widetilde{\zeta}_{A,\O}(s;\d)&=\int_0^{\d}t^{s-N-1}|A_t\cap\O|\di t=\int_0^{\d}t^{s-N-1}(|A_\d\cap\O|-|A_{t,\d}\cap\O|)\di t\\
&=\frac{\d^{s-N}|A_\d\cap\O|}{s-N}-\int_0^{\d}t^{s-N-1}|A_{t,\d}\cap\O|\di t.
\end{aligned}
\end{equation}

\begin{definition}\label{shell_defn}
Let $(A,\O)$ be an RFD in $\eR^N$ and fix $\d>0$. We define the {\em shell zeta function} $\breve{\zeta}_{A,\O}:=\breve{\zeta}_{A,\O}(\,\cdot\,;\d)$ {\em of $A$ relative to $\O$} (or the {\em relative shell zeta function} of $(A,\O)$) by
\begin{equation}\label{shell_zeta}
\breve{\zeta}_{A,\O}(s;\d):=-\int_0^{\delta}t^{s-N-1}|A_{t,\d}\cap\O|\di t,
\end{equation}
for all $s\in\Ce$ with $\re s$ sufficiently large.
Here, the integral is taken in the Lebesgue sense and $A_{t,\d}$ is defined by Equation \eqref{shell_ozn}. 
\end{definition}

In light of \eqref{racun}, we can now easily obtain the following theorem.

\begin{theorem}\label{shell_z_thm}
Let $(A,\O)$ be an RFD in $\eR^N$ and fix $\d>0$. Then the shell zeta function $\breve{\zeta}_{A,\O}(\,\cdot\,;\d)$ of $(A,\O)$ is holomorphic on the open right half-plane $\{\re s>N\}$ and
\begin{equation}\label{shell_der}
\frac{\mathrm{d}}{\mathrm{d}s}\breve{\zeta}_{A,\O}(s;\d)=-\int_0^{\delta}t^{s-N-1}|A_{t,\d}\cap\O|\log t\di t,
\end{equation}
for all $s\in\Ce$ such that $\re s>N$.

Furthermore, for all $s\in\Ce$ such that $\re s>N$,  $\breve{\zeta}_{A,\O}(\,\cdot\,;\d)$ satisfies the following functional equations, connecting it to the tube and distance zeta functions of $(A,\O)$, respectively$:$
\begin{equation}\label{shell_tube}
\widetilde{\zeta}_{A,\O}(s;\d)=\frac{\d^{s-N}|A_\d\cap\O|}{s-N}+\breve{\zeta}_{A,\O}(s;\d)
\end{equation}
and
\begin{equation}\label{shell_dist}
{\zeta}_{A,\O}(s;\d)=(N-s)\breve{\zeta}_{A,\O}(s;\d).
\end{equation}
\end{theorem}

\begin{proof}
To prove the holomorphicity of $\breve{\zeta}_{A,\O}(\,\cdot\,;\d)$, one observes that for every real number $\sigma>N$, we have 
$
|\breve{\zeta}_{A,\O}(\sigma;\d)|\leq |A_\d\cap\O|\int_0^{\d}t^{\sigma-N-1}\di t<\ty,
$
and uses a well-known theorem about differentiation of an integral depending analytically on a parameter (see, e.g., \cite{Matt} or \cite[Theorem~31]{titch}) which also gives the formula \eqref{shell_der} for the derivative.
Formula \eqref{shell_tube} is a rewriting of \eqref{racun} and by combining it with the functional equation \eqref{equ_tilde}, which connects the relative distance and tube zeta functions, we obtain \eqref{shell_dist}.
\end{proof}

In light of Theorem \ref{an_rel}, the principle of analytic continuation combined with Equation \eqref{shell_tube} (or \eqref{shell_dist}) now immediately yields the following properties of the relative shell zeta function.

\begin{theorem}\label{an_shell}
Let $(A,\O)$ be a relative fractal drum in $\eR^N$ such that $\ov{\dim}_B(A,\O)<N$ and fix $\d>0$. Then the following properties hold$:$

\bigskip 

$(a)$ The relative shell zeta function $\breve{\zeta}_{A,\O}(s;\d)$ is meromorphic in the half-plane $\{\re s>\overline{\dim}_B(A,\Omega)\}$, with a single simple pole at $s=N$.
Furthermore,
\begin{equation}
\res(\breve{\zeta}_{A,\O}(\,\cdot\,;\d),N)=-|A_\d\cap\O|.
\end{equation}

\bigskip

$(b)$ If the relative box $($or Minkowski$)$ dimension $D:=\dim_B(A,\O)$ exists, $D<N$ and $\M_*^D(A,\O)>0$, then $\breve{\zeta}_{A,\O}(s)\to+\ty$ as $s\in\eR$
converges to $D$ from the right.   
\end{theorem}
\begin{proof}
We deduce from the principle of analytic continuation that the functional equations \eqref{shell_tube} and \eqref{shell_dist} continue to hold on any open connected set $U\supseteq\{\re s>N\}$ to which any of the three relative zeta functions, $\breve{\zeta}_{A,\O}$, $\widetilde{\zeta}_{A,\O}$ or $\zeta_{A,\O}$, has a holomorphic continuation.
In light of this, part $(a)$ follows from the counterpart of Theorem~\ref{an_rel} for the relative tube zeta function (see also Equation \eqref{401/2} and the last paragraph of Remark \ref{tube_holo}) and \eqref{shell_tube}, while part $(b)$ follows from Theorem~\ref{an_rel} and \eqref{shell_dist}. 
\end{proof}

The following corollary is an immediate consequence of the above theorem, or more precisely, of the functional equation \eqref{shell_tube} and the fact that for a given RFD $(A,\O)$ in $\eR^N$ and any fixed $\d_1,\d_2>0$, the difference $\widetilde{\zeta}_{A,\O}(s;\d_1)-\widetilde{\zeta}_{A,\O}(s;\d_2)$ is an entire function.
(See \cite{refds} or \cite{fzf}.)

\begin{corollary}
Let $(A,\O)$ be an RFD in $\eR^N$ such that $\ov{\dim}_B(A,\O)<N$ and fix $\d_1,\d_2>0$ such that $\d_1<\d_2$.
Then, the difference $\breve{\zeta}_{A,\O}(s;\d_1)-\breve{\zeta}_{A,\O}(s;\d_2)$ is meromorphic on all of $\Ce$ with a single simple pole at $s=N$ of residue $|A_{\d_1,\d_2}\cap\O|$.
\end{corollary}

The next corollary follows at once from the first part of the proof of Theorem~\ref{an_shell}.

\begin{corollary}\label{5.3.4.1/2}
Let $(A,\O)$ be an RFD in $\eR^N$.
The functional equations \eqref{shell_tube} and \eqref{shell_dist} continue to hold on any connected open neighborhood $U\subseteq\Ce$ of the critical line $\{\re s=\ov{\dim}_B(A,\O)\}$ to which any of the three relative zeta functions $\breve{\zeta}_{A,\O}$, $\widetilde{\zeta}_{A,\O}$ or $\zeta_{A,\O}$ can be meromorphically continued.
More specifically, if either $\breve{\zeta}_{A,\O}$, $\widetilde{\zeta}_{A,\O}$ or $\zeta_{A,\O}$ has a $($necessarily unique$)$ meromorphic continuation on the domain $U\subseteq\Ce$, then so do the other two fractal zeta functions and the functional equations \eqref{shell_tube} and \eqref{shell_dist} continue to hold for all $s\in U$ between the resulting meromorphic extensions of $\breve{\zeta}_{A,\O}$, $\widetilde{\zeta}_{A,\O}$ and $\zeta_{A,\O}$.
\end{corollary}

Moreover, in light of Theorem~\ref{pole1mink_tilde} and the functional equation \eqref{shell_tube}, we have the following result.

\begin{theorem}\label{pole_shell}
Assume that $(A,\O)$ is a Minkowski nondegenerate RFD in $\eR^N$, 
that is, $0<\M_*^D(A,\O)\le{\M}^{*D}(A,\O)<\ty$ $($in particular, $\dim_B(A,\O)=D)$, 
and $D<N$.
If $\breve{\zeta}_{A,\O}(s)$ can be meromorphically  extended to a connected open neighborhood of $s=D$,
then $D$ is necessarily a simple pole of $\breve{\zeta}_{A,\O}(s)$ and 
\begin{equation}\label{res_shell}
\M_*^D(A,\O)\le\res(\breve{\zeta}_{A,\O},D)\le{\M}^{*D}(A,\O).
\end{equation}
Furthermore, if $(A,\O)$ is Minkowski measurable, then 
\begin{equation}\label{pole1minkg1=2}
\res(\breve{\zeta}_{A,\O}, D)=\M^D(A,\O).
\end{equation}
\end{theorem} 

The most useful fact about the relative shell zeta function is that the residues of its meromorphic extension at any of its (simple) poles belonging to the open left half-plane $\{\re s<N\}$ have a simple connection to the residues of the relative tube or distance zeta functions.
(See also Corollary \ref{5.3.4.1/2} just above.)

\begin{lemma}\label{reziduuumi}
Assume that $(A,\O)$ is an RFD in $\eR^N$ such that its tube or distance or shell zeta function is meromorphic on some connected open  neighborhood $U\subseteq\Ce$ of the critical line $\{\re s=\ov{\dim}_B(A,\O)\}$.
Then, the multisets of poles located in $U\setminus\{N\}$ of each of the three zeta functions, $\breve{\zeta}_{A,\O}$, $\widetilde{\zeta}_{A,\O}$ and $\zeta_{A,\O}$, coincide$:$
\begin{equation}\label{5.3.10.1/2}
\po\big(\breve{\zeta}_{A,\O},U\setminus\{N\}\big)=\po\big(\widetilde{\zeta}_{A,\O},U\setminus\{N\}\big)=\po\big({\zeta}_{A,\O},U\setminus\{N\}\big).
\end{equation}

Moreover, if $\omega\in U\setminus\{N\}$ is a simple pole of one of the three fractal zeta functions $\breve{\zeta}_{A,\O}$, $\widetilde{\zeta}_{A,\O}$ or $\zeta_{A,\O}$, then it is also a simple pole of the other two fractal zeta functions and we have
\begin{equation}
\res(\breve{\zeta}_{A,\O},\omega)=\res(\widetilde{\zeta}_{A,\O},\omega)=\frac{\res({\zeta}_{A,\O},\omega)}{N-\omega}.
\end{equation}
\end{lemma}


The shell zeta function, $\breve{\zeta}_{A,\O}$, may seem rather artificial and unnecessary in the present context of relative fractal drums, but it will prove to be quite useful as a ``translation tool'' for deriving the tube formulas (originally obtained via the tube zeta function, $\widetilde{\zeta}_{A,\O}$, in Sections \ref{sec_point} and \ref{sec_distr}) in terms of the much more practical distance zeta function, $\zeta_{A,\O}$.
We note that the shell zeta function originally arose naturally in \cite{ra}, where it was used, in particular, to generalize the theory of complex dimensions developed in [LapRa\v Zu1--8] 
to the special case of {\em unbounded sets at infinity} having infinite Lebesgue measure.
(See also \cite{ra2} for a study of the fractal properties of unbounded sets of finite Lebesgue measure at infinity.)

\subsection{Pointwise Tube Formula via the Distance Zeta Function}\label{subsec_point_dist}

Analogously as in the case of the relative tube zeta function of $(A,\O)$, we observe that
$\breve{\zeta}_{A,\O}(s)=\{\mathfrak{M}f\}(s)$, where $f(s):=-t^{-N}\chi_{(0,\d)}(t)|A_{t,\d}\cap\O|$.
We also note that $f$ is continuous and of bounded variation on $(0,+\ty)$; so that we can apply the Mellin inversion theorem (Theorem \ref{mellin_inv}), much as in the proof of Theorem \ref{tube_inversion}, and conclude that
\begin{equation}\label{shell_invv}
|A_{t,\d}\cap\O|=-\frac{1}{2\pi\I}\int_{c-\I\ty}^{c+\I\ty}t^{N-s}\breve{\zeta}_{A,\O}(s;\d)\di s,
\end{equation}
where $c>N$ is arbitrary and $t\in(0,\d)$.
In light of \eqref{shell_tb}, the following theorem is an immediate consequence of the identity \eqref{shell_invv}.

\begin{theorem}\label{shell_inversion}
Let $(A,\O)$ be a relative fractal drum in $\eR^N$ and fix $\d>0$.
Then, for every $t\in(0,\d)$ and any real number $c>N$, we have
\begin{equation}\label{shell_inversion_formula}
|A_t\cap\O|=|A_\d\cap\O|+\frac{1}{2\pi\I}\int_{c-\I\ty}^{c+\I\ty}t^{N-s}\breve{\zeta}_{A,\O}(s;\d)\di s.
\end{equation}
\end{theorem}

It is now clear that if the shell zeta function of $(A,\O)$ satisfies the languidity conditions of Definition \ref{languid}, with the constant $c>N$ in the condition {\bf L1}, or the strong languidity conditions of Definition \ref{str_languid}, we can rewrite the results of Sections \ref{sec_point} and \ref{sec_distr} verbatim in terms of the shell zeta function.
Note that for this to work, it was crucial that in the truncated pointwise formula of Lemma \ref{trunc_point}, we had the freedom to choose any $c\in(\ov{\dim}_B(A,\O),N+1)$.
Furthermore, observe that the additional pole of the shell zeta function at $s=N$ will cancel out the term $|A_\d\cap\O|$ in \eqref{shell_inversion_formula} above.
More specifically, in the analog for the relative shell zeta function of the pointwise formula stated in Theorem \ref{pointwise_formula}, we obtain the following pointwise fractal tube formula with error term, expressed in terms of the shell zeta function $\breve{\zeta}_{A,\O}:=\breve{\zeta}_{A,\O}(\,\cdot\,;\d)$$:$
\begin{equation}\label{5.15}
\begin{aligned}
V^{[k]}_{A,\O}(t)&=\sum_{\omega\in\po(\breve{\zeta}_{A,\O},\bm{W})}\res\left(\frac{t^{N-s+k}}{(N\!-\! s\!+\! 1)_k}\breve{\zeta}_{A,\O}(s),\omega\right)+|A_\d\cap\O|\frac{t^k}{(1)_k}+\breve{R}^{[k]}_{A,\O}(t),
\end{aligned}
\vspace{6pt}
\end{equation}
valid pointwise for all $t\in(0,\d)$.
Here, just as in the statement of Theorem \ref{pointwise_formula}, the shell zeta function $\breve{\zeta}_{A,\O)}$ of the RFD $(A,\O)$ is assumed to be languid for some fixed $\d>0$ and some fixed constant $\kappa\in\eR$, as well as with the constant $c$ satisfying $c>N$.
Furthermore, the nonnegative integer $k$ is assumed to be such that $k>\kappa+1$ and for every $t\in(0,\d)$, the error term $\breve{R}_{A,\O}^{[k]}$ is given (much as in \eqref{error_term}) by the absolutely convergent (and hence, convergent) integral
\begin{equation}\label{5.3.14.1/2}
\breve{R}_{A,\O}^{[k]}(t)=\frac{1}{2\pi\I}\int_S\frac{t^{N-s+k}}{(N\!-\! s\!+\! 1)_k}\breve{\zeta}_{A,\O}(s;\d)\di s.
\end{equation}
Moreover, it satisfies the exact analog of the pointwise error estimate \eqref{R_upper_bound}, valid pointwise for all $t\in(0,\d)$.
Hence, it satisfies (for $\breve{\zeta}_{A,\O}$ instead of for $\widetilde{\zeta}_{A,\O}$) the error estimate \eqref{estm} and, in the special case when the screen $\bm S$ lies strictly to the left of the vertical line $\{\re s=\sup S\}$, it satisfies the exact analog (for $\breve{\zeta}_{A,\O}$) of the stronger error estimate \eqref{S(t)<sup}.

In addition, by singling out the residue at $s=N$ from the above sum and using Lemma \ref{reziduuumi} and Theorem \ref{an_shell}$(a)$, along with the functional equation \eqref{shell_dist}, we can rewrite the above equation (in \eqref{5.15}) as follows:
\begin{equation}
\begin{aligned}
V^{[k]}_{A,\O}(t)&=\sum_{\omega\in\po({\zeta}_{A,\O},\bm{W})}\res\left(\frac{t^{N-s+k}}{(N\!-\! s)_{k+1}}{\zeta}_{A,\O}(s;\d),\omega\right)+{R}^{[k]}_{A,\O}(t),
\end{aligned}
\end{equation}
where the pointwise error term ${R}_{A,\O}^{[k]}$ is now given by the absolutely convergent (and hence, convergent) integral
\begin{equation}\label{5.3.14.1.1/2}
{R}_{A,\O}^{[k]}(t)=\frac{1}{2\pi\I}\int_S\frac{t^{N-s+k}}{(N\!-\! s)_{k+1}}{\zeta}_{A,\O}(s;\d)\di s.
\end{equation}

We next introduce the analogs of the languidity conditions for a relative fractal drum, now formulated in terms of its relative distance zeta function.
We call them {\em $d$-languidity conditions} in order to stress that they are related to the distance zeta function.
\begin{definition}({\em $d$-languidity and strong $d$-languidity}).\label{d_lang}
We say that a relative fractal drum $(A,\O)$ in $\eR^N$ is {\em $d$-languid} (resp., {\em strongly $d$-languid}) if it is languid in the sense of Definition~\ref{languid} (resp., Definition~\ref{str_languid}), but with the relative tube zeta function $\widetilde{\zeta}_{A,\O}=\widetilde{\zeta}_{A,\O}(\,\cdot\,;\d)$ replaced by the relative distance zeta function $\zeta_{A,\O}={\zeta}_{A,\O}(\,\cdot\,;\d)$ and with the constant $c$ appearing in {\bf L1} satisfying $c>N$.
\end{definition}

The following lemma is an immediate consequence of the functional equation \eqref{shell_dist}.
It is crucial in the sense that it allows us to deduce the languidity exponent $\kappa$ of the shell zeta function directly from the $d$-languidity exponent $\kappa_d$ of the distance zeta function.
This cannot be done for the tube zeta function, due to the presence of the term $\delta^{s-N}|A_\delta\cap\O|$ in the functional equation \eqref{equ_tilde}; this is in fact the technical reason for introducing the shell zeta function in the first place.

\begin{lemma}\label{dist_lang}
Let $(A,\O)$ be a relative fractal drum in $\eR^N$ such that $\ov{\dim}_B(A,\O)<N$ and which is $d$-languid for some value $\d>0$ and with some exponent $\kappa_d\in\eR$.
Then the shell zeta function $\breve{\zeta}_{A,\O}$ of $(A,\O)$ satisfies the languidity conditions of Definition~\ref{languid} for the same value of $\d$ and with the exponent $\kappa:=\kappa_d-1$.

Furthermore, if $(A,\O)$ is strongly $d$-languid with the corresponding constant $B>0$ and for some exponent $\kappa_d\in\eR$ and some $\d>0$, then the shell zeta function $\breve{\zeta}_{A,\O}$ of $(A,\O)$ satisfies the strong languidity conditions of Definition \ref{str_languid} with the exponent $\kappa:=\kappa_d-1$ and with the same constant $B$ as well as the same value of $\d$.
\end{lemma}

We are now able to state the main theorem of this section, which is the analog for $\zeta_{A,\O}$ of Theorem \ref{pointwise_formula} stated in terms of $\widetilde{\zeta}_{A,\O}$.

\begin{theorem}[Pointwise fractal tube formula with error term, via $\zeta_{A,\O}$]\label{pointwise_formula_rd}
Let $(A,\O)$ be a relative fractal drum in $\eR^N$ which is $d$-languid for some $\d>0$ and with exponent $\kappa_d\in\eR$.
Furthermore, assume that $\ov{\dim}_B(A,\O)<N$ and let $k>\kappa_d$ be a nonnegative integer.
Then, the following pointwise fractal tube formula, expressed in terms of the distance zeta function ${\zeta}_{A,\O}:={\zeta}_{A,\O}(\,\cdot\,;\d)$, is valid for every $t\in(0,\d)$$:$
\begin{equation}\label{point_form_d}
V^{[k]}_{A,\O}(t)=\sum_{\omega\in\po({\zeta}_(A,\O),\bm{W})}\res\left(\frac{t^{N-s+k}}{(N\!-\! s)_{k+1}}{\zeta}_{A,\O}(s),\omega\right)+R^{[k]}_{A,\O}(t).
\end{equation}
Here, for every $t\in(0,\d)$, the error term $R^{[k]}_{A,\O}$ is given by the absolutely convergent $($and hence, convergent$)$ integral
\begin{equation}\label{error_term_d}
R^{[k]}_{A,\O}(t)=\frac{1}{2\pi\I}\int_S\frac{t^{N-s+k}}{(N\!-\! s)_{k+1}}{\zeta}_{A,\O}(s)\di s.
\end{equation}
Furthermore, for every $t\in(0,\d)$, we have
\begin{equation}
|R^{[k]}_{A,\O}(t)|\leq t^{N+k}\max\{t^{-\sup S},t^{-\inf S}\}\left(\frac{C\big(1+\|S\|_{\mathrm{Lip}}\big)}{2\pi(k-\kappa_d)}+C'\right),
\end{equation}
where $C$ is the constant appearing in {\bf L1} and {\bf L2} and $C'$ is some suitable positive constant.
These constants depend only on the relative fractal drum $(A,\O)$ and the screen, but not on $k$.

In particular, we have the following pointwise error estimate$:$
\begin{equation}
R^{[k]}_{A,\O}(t)=O(t^{N-\sup S+k})\quad\mathrm{ as }\quad t\to 0^+.
\end{equation}

Moreover, if $S(\tau)<\sup S$ $($i.e., if the screen $\bm S$ lies strictly left of the vertical line $\{\re s=\sup S\}$$)$, then we have the following stronger pointwise error estimate$:$
\begin{equation}\label{S(t)<sup_d}  
R^{[k]}_{A,\O}(t)=o(t^{N-\sup S+k})\quad\mathrm{ as }\quad t\to 0^+.
\end{equation}
\end{theorem}

\begin{proof}
In light of Lemma \ref{dist_lang}, we have that $\breve{\zeta}_{A,\O}$, the shell zeta function of $(A,\O)$, also satisfies the appropriate languidity conditions with  $\kappa:=\kappa_d-1$ and for the same value of $\d$.
The theorem now follows much as in the case of the relative tube zeta function $\widetilde{\zeta}_{A,\O}$; see the proof of Theorem \ref{pointwise_formula} and the discussion following Theorem \ref{shell_inversion}.
\end{proof} 





\begin{remark}
In Theorem \ref{pointwise_formula_rd}, as well as in all of the following theorems below involving the relative distance zeta function, the additional assumption according to which $\ov{\dim}_B(A,\O)<N$ is made in order to avoid the situation where $s=N$ is a pole of $\widetilde{\zeta}_{A,\O}$.
\end{remark} 

The next result is the counterpart for $\zeta_{A,\O}$ of Theorem \ref{str_pointwise_formula}, which is stated in terms of $\widetilde{\zeta}_{A,\O}$.

\begin{theorem}[Exact pointwise fractal tube formula via $\zeta_{A,\O}$]\label{str_pointwise_formula_d_bez}
Let $(A,\O)$ be a relative fractal drum in $\eR^N$ which is strongly $d$-languid for some $\d>0$ and with exponent $\kappa_d\in\eR$.
Furthermore, let $k>\kappa_d-1$ be a nonnegative integer and assume that $\ov{\dim}_B(A,\O)<N$.
Then, the following exact pointwise fractal tube formula, expressed in terms of the distance zeta function ${\zeta}_{A,\O}:={\zeta}_{A,\O}(\,\cdot\,;\d)$, holds for every $t\in(0,\min\{1,\d,B^{-1}\})$$:$
\begin{equation}\label{point_form_w_d_bez}
V^{[k]}_{A,\O}(t)=\sum_{\omega\in\po({\zeta}_{A,\O},\Ce)}\res\left(\frac{t^{N-s+k}}{(N\!-\! s)_{k+1}}{\zeta}_{A,\O}(s),\omega\right).
\end{equation}
Here, $B$ is the constant appearing in {\bf L2'} and $\kappa_d$ is the exponent occurring in the statement of hypotheses {\bf L1} and {\bf L2'}.
\end{theorem}

\begin{proof}
In light of Lemma \ref{dist_lang}  and the functional equation \eqref{shell_dist}, the proof of the theorem parallels that of Theorem~\ref{pointwise_formula_rd} and of Theorem \ref{str_pointwise_formula}, except (in the latter case) for the tube zeta function $\widetilde{\zeta}_{A,\O}(\,\cdot\,;\d)$ now being replaced by the shell zeta function $\breve{\zeta}_{A,\O}(\,\cdot\,;\d)$.
\end{proof}

A situation that occurs frequently in the applications is when a relative fractal drum $(A,\O)$ is `almost' strongly $d$-languid.
More precisely, $(A,\O)$ will satisfy all of the conditions of strong $d$-languidity, except the condition that ${\bf L1}$ is satisfied for {\em all} $\sigma<c$.
For example, let $A$ be the middle-third Cantor set constructed in $[0,1]$ and let $\O=(0,1)$.
Then, the relative distance zeta function $\zeta_{A,\O}$ is meromorphic on all of $\Ce$ and given for all $s\in\Ce$ by (see \cite[Example 3.4]{dtzf} or \cite[Example 2.1.81]{fzf}):
\begin{equation}
\zeta_{A,\O}(s)=\frac{2^{1-s}}{s(3^s-2)}.
\end{equation}
As one can easily check, it almost satisfies the strong languidity conditions with $\kappa_d:=-1$, where the sequence of screens $\bm{S}_m$ can be taken as the sequence of vertical lines $\{\re s=-m\}$ for $m\in\eN$.
The problem here arises from the factor $2^{-s}$ which tends exponentially fast to $+\ty$ as $\re s\to-\ty$, so that condition ${\bf L1}$ cannot be fulfilled for all $\sigma<c$.
In order to remedy this problem and obtain a pointwise formula in this and similar situations, we can multiply $\zeta_{A,\O}(s)$ by $2^s$ and then, the resulting function will be strongly $d$-languid.
On the other hand, by the scaling property of the relative distance zeta function (see \cite[Section 2.2]{mefzf} or \cite[Theorem 4.1.38]{fzf}), we have that $2^s\,\zeta_{A,\O}(s)=\zeta_{2A,2\O}(s)$, where $(2A,2\O)$ is the scaled version of the RFD $(A,\O)$.
(For a subset $A$ of $\eR^N$ and any $\lambda>0$, we define $\lambda A:=\{\lambda x:x\in A\}$.)
In light of the above discussion, we can now state the following corollary dealing with this situation and which will be used repeatedly (most often implicitly) in Section \ref{exp_app}.

\begin{corollary}[Exact pointwise fractal tube formula via $\zeta_{A,\O}$; scaled version]\label{str_pointwise_formula_d}
Let $(A,\O)$ be a relative fractal drum in $\eR^N$ such that $\ov{\dim}_B(A,\O)<N$.
Furthermore, assume that there exists a scaling factor $\lambda>0$ such that $(\lambda A,\lambda\O)$ is a strongly $d$-languid RFD in $\eR^N$, for some $\d>0$ and with exponent $\kappa_d\in\eR$.
Moreover, let $k>\kappa_d-1$ be a nonnegative integer.
Then, the following exact pointwise fractal tube formula, expressed in terms of the distance zeta function $\zeta_{A,\O}$, holds for every $t\in(0,\lambda^{-1}\min\{1,\d,B^{-1}\})$$:$
\begin{equation}\label{point_form_w_d}
V^{[k]}_{A,\O}(t)=\sum_{\omega\in\po({\zeta}_{A,\O},\Ce)}\res\left(\frac{t^{N-s+k}}{(N\!-\! s)_{k+1}}{\zeta}_{A,\O}(s),\omega\right).
\end{equation}
Here, $B$ is the constant appearing in {\bf L2'} $($for the function $s\mapsto{\zeta}_{\lambda A,\lambda\O}(s;\d)=\lambda^{s}{\zeta}_{A,\O}(s;\d\lambda^{-1}))$ and $\kappa_d$ is the exponent occurring in the statement of hypotheses {\bf L1} and {\bf L2'}.
\end{corollary}

\begin{proof}
Let us denote by $V_\lambda^{[k]}(\tau)$ the $k$-th primitive of the function
$$
\tau\mapsto|(\lambda A)_\tau\cap\lambda\O|.
$$
It is easy to show that $V_\lambda^{[0]}(\tau)=\lambda^NV^{[0]}(\tau/\lambda)$ (see also \cite[Lemma 4.6.10]{fzf}).
Using this scaling property, we then see that
\begin{equation}
V_\lambda^{[1]}(\tau)=\int_0^{\tau}V_\lambda^{[0]}(t)\di t=\lambda^N\int_0^{\tau}V^{[0]}_{A,\O}(t/\lambda)\di t=\lambda^{N+1}\int_0^{\tau/\lambda}V^{[0]}_{A,\O}(\xi)\di\xi,
\end{equation}
or, in other words, $V_\lambda^{[1]}(\tau)=\lambda^{N+1}V^{[1]}_{A,\O}(\tau/\lambda)$. Hence, by induction, we deduce that
\begin{equation}\label{lambda_k}
V_\lambda^{[k]}(\tau)=\lambda^{N+k}V^{[k]}_{A,\O}(\tau/\lambda),
\end{equation}
for all nonnegative integers $k$.

We now apply Theorem \ref{str_pointwise_formula_d_bez} to the relative fractal drum $(\lambda A,\lambda\O)$ and obtain the following exact fractal tube formula, valid pointwise for all $\tau\in(0,\min\{1,\d,B^{-1}\})$:
\begin{equation}\label{llll}
V_{\lambda}^{[k]}(\tau)=\sum_{\omega\in\po({\zeta}_{\lambda A,\lambda\O},\Ce)}\res\left(\frac{\tau^{N-s+k}}{(N\!-\! s)_{k+1}}{\zeta}_{\lambda A,\lambda\O}(s;\d),\omega\right).
\end{equation}

Next, combining \eqref{lambda_k} with \eqref{llll} and the scaling property of the relative distance zeta function (namely, $\zeta_{\lambda A,\lambda\O}(s)=\lambda^s\zeta_{A,\O}(s)$; see \cite[Section 2.3]{refds} or \cite[Theorem 4.1.38]{fzf}), we deduce that
\begin{equation}
\lambda^{N+k}V^{[k]}_{A,\O}(\tau/\lambda)=\sum_{\omega\in\po({\zeta}_{A,\O},\Ce)}\res\left(\frac{\tau^{N-s+k}\lambda^s}{(N\!-\! s)_{k+1}}{\zeta}_{A,\O}(s;\d\lambda^{-1}),\omega\right).
\end{equation}
Finally, we complete the proof of the corollary by multiplying the above identity by $\lambda^{-N-k}$ and introducing a new variable $t:=\tau/\lambda$.
\end{proof}

\begin{remark}
We point out that an analogous corollary can be stated in terms of the relative tube zeta function and the exact pointwise tube formula of Theorem~\ref{str_pointwise_formula}.
\end{remark}

The most interesting situation is, of course, the case when we can apply Theorems \ref{pointwise_formula_rd} and \ref{str_pointwise_formula_d_bez} or Corollary \ref{str_pointwise_formula_d} at the level $k=0$.
We now state the corresponding corollaries of these two theorems as a separate (and single) theorem.

\begin{theorem}[Pointwise fractal tube formula via $\zeta_{A,\O}$; level $k=0$]\label{pointwise_thm_d}\mbox{}

\medskip

$(i)$\ \ Under the same hypotheses as in Theorem~\ref{pointwise_formula_rd}, with $k:=0$, and using the same notation as in that theorem, with $\kappa_d<0$, the following pointwise fractal tube formula with error term, expressed in terms of the distance zeta function $\zeta_{A,\O}:=\zeta_{A,\O}(\,\cdot\,;\d)$, 
holds for all $t\in(0,\d)$$:$
\begin{equation}\label{point_tube_form_d}
|A_t\cap\O|=\sum_{\omega\in\po({\zeta}_{A,\O},\bm{W})}\res\left(\frac{t^{N-s}}{N\!-\! s}{\zeta}_{A,\O}(s),\omega\right)+R^{[0]}_{A,\O}(t),
\end{equation}
where $R^{[0]}_{A,\O}(t)$ is the error term given by formula~\eqref{error_term_d} with $k:=0$.
Furthermore, we have the following pointwise error estimate$:$
\begin{equation}\label{est_0_1}
R^{[0]}_{A,\O}(t)=O(t^{N-\sup S})\quad\mathrm{ as }\quad t\to 0^+.
\end{equation}
Moreover, if $S(\tau)<\sup S$ for every $\tau\in\eR$ $($i.e., if the screen $\bm S$ lies strictly to the left of the vertical line $\{\re s=\sup S\}$$)$, we then have the following stronger pointwise error estimate:
\begin{equation}\label{S(t)<sup_0_d}  
R^{[0]}_{A,\O}(t)=o(t^{N-\sup S})\quad\mathrm{ as }\quad t\to 0^+.
\end{equation}

\medskip

$(ii)$\ \ Finally, under the same hypotheses as in Theorem \ref{str_pointwise_formula_d_bez} or Corollary~\ref{str_pointwise_formula_d}, with $k:=0$ and $\kappa_d<1$, and if, in addition $(\lambda A,\lambda \O)$ is strongly $d$-languid for some $\lambda>0$, then the fractal tube formula \eqref{point_tube_form_d} holds pointwise for all $t\in(0,\lambda^{-1}\min\{1,\d,B^{-1}\})$, with $R^{[0]}_{A,\O}(t)\equiv 0$ and $\bm{W}:=\Ce$; so that \eqref{point_tube_form_d} becomes an {\rm exact} fractal tube formula in this case.
\end{theorem}

The exact analog of Remark \ref{5.1.13.1/4}, 
and Remark \ref{5.1.13.3/4} holds in the present situation, except for the relative tube zeta function $\widetilde{\zeta}_{A,\O}$ replaced by the relative distance zeta function ${\zeta}_{A,\O}$ of the RFD $(A,\O)$.
We state the most interesting case in a separate theorem which is of course, the corollary of Theorem \ref{pointwise_thm_d} corresponding to the level $k=0$.

\begin{theorem}[Pointwise fractal tube formula via ${\zeta}_{A,\O}$; level $k=0$ and the case of simple poles]\label{5.3.15.1/2}
Assume that the hypotheses of Theorem \ref{pointwise_thm_d} hold.
Suppose, in addition, that all of the visible complex dimensions of the RFD $(A,\O)$ are simple $($i.e., all of the poles of ${\zeta}_{A,\O}$ or, equivalently, since $\ov{D}:=\ov{\dim}_B(A,\O)<N$ here, of $\widetilde{\zeta}_{A,\O}$, belonging to the window $\bm W$ are simple$)$.
Then, the pointwise fractal tube formula \eqref{point_tube_form_d}, expressed in terms of ${\zeta}_{A,\O}$, takes the following simpler form, valid for all $t\in(0,\d)$$:$
\begin{equation}\label{5.3.29.1/2}
|A_t\cap\O|=\sum_{\omega\in\po({\zeta}_{A,\O},\bm{W})}\frac{t^{N-\omega}}{N-\omega}\res\left({\zeta}_{A,\O}(s),\omega\right)+R^{[0]}_{A,\O}(t),
\end{equation}
where the $($pointwise$)$ error term $R^{[0]}_{A,\O}$ is the same as in Theorem \ref{pointwise_formula_rd} at level $k=0$ and hence, satisfies the same $($pointwise$)$ error estimates $[$\eqref{est_0_1} or \eqref{S(t)<sup_0_d}, depending on the hypotheses$]$ as in Theorem \ref{pointwise_thm_d}.
In particular, in the strongly languid case $($i.e., if $(\lambda A,\lambda\O)$ is strongly languid for some $\lambda>0$$)$, we have $R^{[0]}_{A,\O}\equiv 0$ and $\bm{W}:=\Ce$, so that \eqref{5.3.29.1/2} then becomes an {\rm exact} pointwise fractal tube formula, valid for all $t\in(0,\lambda^{-1}\min\{1,\d,B^{-1}\})$. 
\end{theorem}


\subsection{Distributional Tube Formula via the Distance Zeta Function}\label{subsec_distr_dist}

In this subsection, we state the distributional analogs of the results from Section \ref{subsec_point_dist} above, still expressed in terms of the relative distance zeta function.
The proofs are completely analogous to the ones from Section \ref{sec_distr} for the case of the relative tube zeta function.
Again, we use the relative shell zeta function and the same scaling technique as in the proof of Corollary \ref{str_pointwise_formula_d} (and Theorem \ref{str_pointwise_formula_d_bez}) above in order to obtain the desired results under the hypotheses of $d$-languidity or of strong $d$-languidity. 

\begin{theorem}[Distributional fractal tube formula with error term, via $\zeta_{A,\O}$]\label{dist_error_d}
Let $(A,\O)$ be a $d$-languid relative fractal drum in $\eR^N$ for some $\d>0$ and $\kappa_d\in\eR$.
Furthermore, assume that $\ov{\dim}_B(A,\O)<N$.
Then, for every $k\in\Ze$, the distribution $\mathcal V^{[k]}_{A,\O}$ in $\mathcal{K}'(0,\d)$ $($and hence, also in $\mathcal{D}'(0,\d)$$)$ is given by the following distributional fractal tube formula, with error term and expressed in terms of the distance zeta function ${\zeta}_{A,\O}:={\zeta}_{A,\O}(\,\cdot\,;\d)$$:$
\begin{equation}\label{dist_form_error_d}
\mathcal V^{[k]}_{A,\O}(t)=\sum_{\omega\in\po({\zeta}_{A,\O},\bm{W})}\res\left(\frac{t^{N-s+k}}{(N\!-\! s)_{k+1}}{\zeta}_{A,\O}(s),\omega\right)+\mathcal R^{[k]}_{A,\O}(t).
\end{equation}
That is, the action of $\mathcal V^{[k]}_{A,\O}(t)$ on an arbitrary test function $\varphi\in\mathcal{K}(0,\d)$ is given by
\begin{equation}\label{error_action_}
\begin{aligned}
\big\langle\mathcal V^{[k]}_{A,\O},\varphi\big\rangle&=\sum_{\omega\in\po({\zeta}_{A,\O},\bm{W})}\res\left(\frac{\{\mathfrak{M}\varphi\}(N\!-\! s\!+\! 1\!+k)\,{\zeta}_{A,\O}(s)}{(N\!-\! s)_{k+1}},\omega\right)+\big\langle\mathcal R^{[k]}_{A,\O},\varphi\big\rangle.
\end{aligned}
\end{equation}
Here, the distribution $\mathcal R^{[k]}_{A,\O}$ in $\mathcal{K}'(0,\d)$ is the distributional error term given for all $\varphi\in\mathcal{K}(0,\d)$ by
\begin{equation}\label{R_distr_d}
\big\langle\mathcal R^{[k]}_{A,\O},\varphi\big\rangle=\frac{1}{2\pi\I}\int_S\frac{\{\mathfrak{M}\varphi\}(N\!-\! s\!+\! 1\!+\! k)\,{\zeta}_{A,\O}(s)}{(N\!-\! s)_{k+1}}\di s.
\end{equation}
Furthermore, the distribution $\mathcal R^{[k]}_{A,\O}(t)$ is of asymptotic order at most $t^{N-\sup S+k}$ as $t\to 0^+$; i.e,
\begin{equation}
\mathcal R^{[k]}_{A,\O}(t)=O(t^{N-\sup S+k})\quad\mathrm{as}\quad t\to0^+
\end{equation}
in the sense of Definition~\ref{dist_order_dis}.

Moreover, if $S(\tau) < \sup S$ for all $\tau\in\eR$ $($that is, if the screen lies strictly
to the left of the line $\re s =\sup S)$, then $\mathcal R^{[k]}_{A,\O}(t)$ is of asymptotic order less than $t^{N-\sup S+k}$; i.e.,
\begin{equation}\label{dist_estimate_o_d}
\mathcal R^{[k]}_{A,\O}(t)=o(t^{N-\sup S+k})\quad\mathrm{as}\quad t\to0^+.
\end{equation}
\end{theorem}

In the case of a (possibly scaled) strongly $d$-languid relative fractal drum, as before, we obtain a distributional formula without an error term, as stated in the next theorem.

\begin{theorem}[Exact distributional fractal tube formula via $\zeta_{A,\O}$]\label{dist_no_error_d}
Let $(A,\O)$ be a relative fractal drum in $\eR^N$ and assume also that $\ov{\dim}_B(A,\O)<N$.
Furthermore, assume that there exists $\lambda>0$ such that $(\lambda A,\lambda\O)$ is strongly $d$-languid for some $\d>0$, $\kappa_d\in\eR$, and let $\d_0:=\lambda^{-1}\min\{1,\d,B^{-1}\}$.\footnote{Here, $B$ is the constant appearing in condition {\bf L2'} for the function $\zeta_{\lambda A}(s,\lambda\O;\d)$.}
Then, for every $k\in\Ze$, the distribution $\mathcal V^{[k]}_{A,\O}$ in $\mathcal{D}'(0,\d_0)$ is given in terms of ${\zeta}_{A,\O}:={\zeta}_{A,\O}(\,\cdot\,;\d)$ by
\begin{equation}\label{dist_form_no_error_d}
\mathcal V^{[k]}_{A,\O}(t)=\sum_{\omega\in\po({\zeta}_{A,\O},\Ce)}\res\left(\frac{t^{N-s+k}}{(N\!-\! s)_{k+1}}{\zeta}_{A,\O}(s),\omega\right).
\end{equation}
That is, the action of $\mathcal V^{[k]}_{A,\O}$ on an arbitrary test function $\varphi\in\mathcal{D}(0,\d_0)$ is given by
\begin{equation}\label{no_error_action_d}
\begin{aligned}
\big\langle\mathcal V^{[k]}_{A,\O}(t),\varphi\big\rangle&=\!\!\!\sum_{\omega\in\po({\zeta}_{A,\O},\Ce)}\!\!\!\res\left(\frac{\{\mathfrak{M}\varphi\}(N\!-\! s\!+\! 1\!+\! k)\,{\zeta}_{A,\O}(s)}{(N\!-\! s)_{k+1}},\omega\right).\\
\end{aligned}
\end{equation}
\end{theorem}

We conclude this section by stating as a separate (and single) theorem the most interesting special case of Theorems \ref{dist_error_d} and \ref{dist_no_error_d}, when $k=0$.

\begin{theorem}[Distributional fractal tube formula via ${\zeta}_{A,\O}$; level $k=0$]\label{dist_tube_formula_d}
Under the same hypotheses as in Theorem~\ref{dist_error_d}, with $k:=0$, we have the following distributional equality in $\mathcal{K}'(0,\d)$ for the relative tube function $t\mapsto|A_t\cap\O|$ of the relative fractal drum $(A,\O)$ in $\eR^N:$
\begin{equation}\label{distr_tube_d}
|A_t\cap\O|=\sum_{\omega\in\po({\zeta}_{A,\O},\bm{W})}\res\left(\frac{{t^{N-s}}}{N\!-\! s}{\zeta}_{A,\O}(s),\omega\right)+\mathcal R^{[0]}_{A,\O}(t),
\end{equation}
where $\mathcal R^{[0]}_{A,\O}(t)$ is given by \eqref{R_distr_d}  for $k=0$ and $\mathcal R^{[0]}_{A,\O}(t)=O(t^{N-\sup S})$ as $t\to 0^+$ or, if $S(\tau) < \sup S$ for all $\tau\in\eR$, then  $\mathcal R^{[0]}_{A,\O}(t)=o(t^{N-\sup S})$ as $t\to 0^+$. 

Moreover, under the same hypotheses as in Theorem \ref{dist_no_error_d}, with $k:=0$, and if $(\lambda A,\lambda \O)$ is strongly $d$-languid for some $\lambda>0$, then the analog of \eqref{distr_tube_d} holds in $\mathcal{D}'(0,\d_0)$, where $\d_0:=\lambda^{-1}\min\{1,\d,B^{-1}\}$ and with $\mathcal R^{[0]}_{A,\O}(t)\equiv 0$ and $\bm{W}:=\Ce$; so that we obtain an {\rm exact} distributional fractal tube formula in this case.

Finally, if, in addition, each visible complex dimension of $(A,\O)$ is simple $($i.e., if each pole of $\zeta_{A,\O}$ or, equivalently, of $\widetilde{\zeta}_{A,\O}$, located in $\bm W$ is simple$)$, then the sum over the complex dimensions in \eqref{distr_tube_d} $($or in its analog with $\bm{W}:=\Ce$, for the exact tube formula$)$ becomes
\begin{equation}\label{5.3.41.1/2}
\sum_{\omega\in\po({\zeta}_{A,\O},\bm{W})}\frac{{t^{N-\omega}}}{N\!-\! \omega}\res\left({\zeta}_{A,\O}(s),\omega\right).
\end{equation}
\end{theorem}

\subsection{The Relative Mellin Zeta Function}\label{subsec_mellin}

In this subsection, we introduce a new fractal zeta function, called the {\em relative Mellin zeta function}.
Our motivation for doing so is to use this new zeta function in order to obtain a distributional tube formula which is valid on a larger space of test functions.
This extension will be required in order to obtain a Minkowski measurability criterion in \cite{ftf_b} (see also \cite{fzf}) but will not be needed elsewhere in the present  paper.
Nevertheless, we include it here since it is a natural extension of the present theory of fractal tube formulas and essentially follows from the same ideas already used in previous sections. 
 
We want to extend our distributional tube formulas to the space $\mathcal{K}(0,+\ty)$.\label{Kspace+}\footnote{Here, $\mathcal{K}(0,+\ty)$ is defined exactly as $\mathcal{K}(0,\d)$ just before Definition \ref{Vdist}, except for $\d$ replaced by $+\ty$, and in this case, we require that for every $\varphi\in\mathcal{K}(0,+\ty)$, $t^m\varphi^{(q)}(t)\to 0$ as $t\to +\ty$, where $\varphi^{(q)}$ denotes the $q$-th derivative of $\varphi$.}
Recall that in Definition \ref{drum}, we have assumed that an RFD $(A,\O)$ has the property that there exists $\d>0$ such that $\O\subseteq A_{\d}$.
Therefore, for an RFD $(A,\O)$ we observe that $A_\d\cap\O=\O$ for all $\d$ sufficiently large; for such values of $\d$, we have that $|A_{\d}\cap\O|=|\O|$, which enables us to redefine the tube zeta function in the following way.
Namely, assume that $\overline{D}:=\overline{\dim}_B(A,\O)<N$ and recall the functional equation \eqref{equ_tilde}, written in the following integral form: 
\begin{equation}
\int_{A_\d\cap\O}d(x,A)^{s-N}\di x=\d^{s-N}|A_\d\cap\O|+(N-s)\int_0^{\d}t^{s-N-1}|A_t\cap\O|\di t,
\end{equation}
initially valid for all $s\in\Ce$ such that $\re s>\overline{D}$.
Now, we may take the complex number $s$ to be in the vertical strip $\{\re s>\ov{D}\}\cap\{\re s<N\}$ and let $\d\to+\ty$ to obtain the following equality:
\begin{equation}\label{mellin-tube}
\int_{\O}d(x,A)^{s-N}\di x=(N-s)\int_0^{+\ty}t^{s-N-1}|A_t\cap\O|\di t.
\end{equation}
Observe that on the right-hand side of \eqref{mellin-tube}, we have the Mellin transform of the function $t^{-N}|A_t\cap\O|$ and this integral is absolutely convergent inside the vertical strip $\{\re s>\ov{D}\}\cap\{\re s<N\}$.
Indeed, to see this, we note that
\begin{equation}\label{right_hand_5}
\int_0^{+\ty}t^{s-N-1}|A_t\cap\O|\di t=\int_0^{1}t^{s-N-1}|A_t\cap\O|\di t+\int_1^{+\ty}t^{s-N-1}|A_t\cap\O|\di t,
\end{equation}
where the integral over $(0,1)$ is equal to $\widetilde{\zeta}_A(s,\O;1)$ and hence, is absolutely convergent on $\{\re s>\overline{D}\}$, while for the integral over $(1,+\ty)$, we have
\begin{equation}
\begin{aligned}
\left|\int_1^{+\ty}t^{s-N-1}|A_t\cap\O|\di t\right|&\leq\int_1^{+\ty}t^{\re s-N-1}|A_t\cap\O|\di t\\
&\leq|\O|\int_1^{+\ty}t^{\re s-N-1}\di t=\frac{|\O|}{N-\re s}.
\end{aligned}
\end{equation}

The classic theorem about the holomorphicity of an integral depending analytically on a complex parameter (see \cite[Theorem 2.1.46]{fzf} or \cite{Matt}) implies that the integral on the right-hand side of \eqref{mellin-tube} defines a holomorphic function on the vertical strip $\{\overline{D}<\re s<N\}$ and upon analytic continuation, that the entire right-hand side of \eqref{mellin-tube} coincides (within that strip) with the relative distance zeta function $\zeta_{A,\O}(s)$; i.e., the identity \eqref{mellin-tube} holds as an equality between holomorphic functions defined on the open vertical strip $\{\ov{D}<\re s<N\}$.

Moreover, upon further meromorphic continuation (and since, by Theorem \ref{an_rel}, $\zeta_{A,\O}$ is holomorphic in the open right half-plane $\{\re s>\ov{D}\}$), we also deduce that if $\zeta_{A,\O}$ can be meromorphically continued to a given connected open neighborhood $U$ of the critical line $\{\re s=\ov{D}\}$, then (with the terminology and notation of Definition \ref{mellin_zeta_def} just below), so can the Mellin zeta function $\zeta_{A,\O}^{\mathfrak M}$.
Hence, we deduce that the following functional equation holds (between meromorphic functions):
\begin{equation}\label{5.4.26.1/2}
\zeta_{A,\O}(s)=(N-s)\zeta_{A,\O}^{\mathfrak M}(s),
\end{equation}
for all $s\in U$.

\begin{definition}\label{mellin_zeta_def}
Let $(A,\O)$ be an RFD in $\eR^N$ such that $\overline{\dim}_B(A,\O)<N$. We define the {\em Mellin zeta function} $\zeta_{A,\O}^{\mathfrak M}$ of $(A,\O)$ by
\begin{equation}\label{mellin_zeta_1}
{\zeta}^{\mathfrak{M}}_{A,\O}(s):=\int_0^{+\ty}t^{s-N-1}|A_{t}\cap\O|\di t,
\end{equation}
for all $s\in\Ce$ with $\re s\in(\overline{\dim}_B(A,\O),N)$, where the integral is taken in the Lebesgue sense.
\end{definition}

In the discussion preceding Definition \ref{mellin_zeta_def}, we have already proven a part of the following theorem.

\begin{theorem}\label{mellin_an}
Let $(A,\O)$ be an RFD in $\eR^N$ such that $\overline{\dim}_B(A,\O)<N$.
Then, the Mellin zeta function ${\zeta}^{\mathfrak{M}}_{A,\O}$ is holomorphic on the open vertical strip $\{\overline{\dim}_B(A,\O)\allowbreak<\re s<N\}$ and
\begin{equation}\label{mellin_der}
\frac{\di}{\di s}{\zeta}^{\mathfrak{M}}_{A,\O}(s)=\int_0^{+\ty}t^{s-N-1}|A_{t}\cap\O|\log t\di t,
\end{equation}
for all $s$ in $\{\overline{\dim}_B(A,\O)<\re s<N\}$.
Furthermore,  $\{\overline{\dim}_B(A,\O)<\re s<N\}$ is the largest vertical strip $($of the form $\{\alpha<\re s<\beta\}$, with $-\ty\leq\alpha<\beta\leq +\ty$$)$ on which the integral on the right-hand side of \eqref{mellin_zeta_1} is absolutely convergent $($i.e., is a convergent Lebesgue integral$)$.

Moreover, for all $s\in\Ce$ such that $\overline{\dim}_B(A,\O)<\re s<N$ and for any fixed $\d>0$ such that $\O\subseteq A_\d$, ${\zeta}^{\mathfrak{M}}_{A,\O}$ satisfies the following functional equations$:$
\begin{equation}\label{mellin_tube_t}
{\zeta}^{\mathfrak{M}}_{A,\O}(s)=\widetilde{\zeta}_{A,\O}(s;\d)+\frac{\d^{s-N}|\O|}{N-s}
\end{equation}
and
\begin{equation}\label{mellin_dist}
{\zeta}^{\mathfrak{M}}_{A,\O}(s)=\frac{{\zeta}_{A,\O}(s;\d)}{N-s}.
\end{equation}
\end{theorem}


\begin{proof}
We have already proven the first part of the theorem.
The optimality of the vertical strip follows directly from \eqref{right_hand_5} (or, more precisely, from \eqref{5.4.26.1/2}).
Namely, the lower bound $\overline{\dim}_B(A,\O)$ is a consequence of the presence of the first integral on the right-hand side of \eqref{right_hand_5} since the latter integral is equal to $\widetilde{\zeta}_A(s,\O;1)$.
Furthermore, the upper bound $N$ is a consequence of the presence of the second integral on the right-hand side of \eqref{right_hand_5}, since that integral is divergent for any real number $s$ such that $s>N$.
To see this, let $\d\geq 1$ be such that $\O\subseteq A_\d$ and make the following observation:
\begin{equation}
\begin{aligned}
\int_1^{+\ty}t^{s-N-1}|A_t\cap\O|\di t&\geq\int_{\d}^{+\ty}t^{s-N-1}|A_t\cap\O|\di t\\
&=|\O|\int_{\d}^{+\ty}t^{s-N-1}\di t=+\ty.
\end{aligned}
\end{equation}

The functional equation \eqref{mellin_dist} is already proven, while \eqref{mellin_tube_t} can be proven directly by splitting the integral defining ${\zeta}^{\mathfrak{M}}_{A,\O}$ over the intervals $(0,\d)$ and $(\d,+\ty)$ or by using the functional equation \eqref{equ_tilde} connecting the tube and distance zeta functions. 
\end{proof}

As a consequence of the functional equations \eqref{mellin_dist}, \eqref{mellin_tube_t} and the principle of analytic continuation, we immediately obtain the following two theorems, which follow from the corresponding ones for the relative distance and tube zeta functions (see \cite[Theorem 2.1 and Section 2.4]{refds} or \cite[Theorems 2.1.11 and 2.2.11]{fzf}, as well as Theorems \ref{pole1} and \ref{pole1mink_tilde}).

\begin{theorem}\label{an_mellin}
Let $(A,\O)$ be a relative fractal drum in $\eR^N$ such that $\overline{\dim}_{B}(A,\O)<N$. Then the following properties hold$:$

\bigskip 

$(a)$ The Mellin zeta function ${\zeta}^{\mathfrak{M}}_{A,\O}$ is meromorphic in the half-plane $\{\re s>\overline{\dim}_B(A,\Omega)\}$ with a single, simple pole at $s=N$.
Furthermore,
\begin{equation}
\res({\zeta}^{\mathfrak{M}}_{A,\O},N)=-|\O|.
\end{equation}

\bigskip

$(b)$ If the relative box $($or Minkowski$)$ dimension $D:=\dim_B(A,\O)$ exists, and $\M_*^D(A,\O)>0$, then ${\zeta}^{\mathfrak{M}}_{A,\O}(s)\to+\ty$ as $s\in\eR$
converges to $D$ from the right.   
\end{theorem}
\begin{proof}
By the principle of analytic continuation, we conclude that the functional equalities \eqref{mellin_tube_t} and \eqref{mellin_dist} continue to hold on any connected open neighborhood $U\subseteq\Ce$ of the vertical strip $\{\overline{\dim}_{B}(A,\O)<\re s<N\}$ to which any of the three relative zeta functions has a holomorphic continuation.
(See also the text surrounding Equation \eqref{5.4.26.1/2}.)
As a result, part $(a)$ follows from the counterpart of Theorem~\ref{an_rel} for the relative tube zeta function (see also \cite[Theorem 2.2.11]{fzf}) and \eqref{mellin_tube_t}, while part $(b)$ follows from Theorem~\ref{an_rel} and \eqref{mellin_dist}. 
\end{proof}

Furthermore, in light of Theorem~\ref{pole1mink_tilde} and \eqref{mellin_tube_t}, one obtains the following analogous result.

\begin{theorem}\label{pole_mellin}
Assume that $(A,\O)$ is a Minkowski  nondegenerate RFD in $\eR^N$, 
that is, $0<\M_*^D(A,\O)\le{\M}^{*D}(A,\O)<\ty$ $($in particular, $D:=\dim_B(A,\O)$ exists$)$, 
and $D<N$.
If ${\zeta}^{\mathfrak{M}}_{A,\O}$ can be extended meromorphically to a connected open neighborhood of $s=D$,
then $D$ is necessarily a simple pole of ${\zeta}^{\mathfrak{M}}_{A,\O}$ and 
\begin{equation}\label{res_mellin}
\M_*^D(A,\O)\le\res({\zeta}^{\mathfrak{M}}_{A,\O},D)\le{\M}^{*D}(A,\O).
\end{equation}
Furthermore, if $(A,\O)$ is Minkowski measurable, then 
\begin{equation}\label{pole1minkg1_mellin}
\res({\zeta}^{\mathfrak{M}}_{A,\O}, D)=\M^D(A,\O).
\end{equation}
\end{theorem} 

\begin{lemma}\label{reziduuumi_mellin}
Assume that $(A,\O)$ is an RFD in $\eR^N$ with $\overline{\dim}_B(A,\O)<N$ and such that its tube or distance or Mellin zeta function is meromorphic on some  connected open neighborhood $U$ of the vertical strip $\{\overline{\dim}_B(A,\O)<\re s<N\}$.\footnote{Recall from Theorem \ref{an_rel} and its counterpart for the relative tube zeta function that $\zeta_{A,\O}=\zeta_{A,\O}$ and $\widetilde{\zeta}_{A,\O}=\widetilde{\zeta}_{A,\O}$ are holomorphic on the open right-half plane $\{\re s>\ov{\dim}_B(A,\O)\}$.}
Then, all of the above fractal zeta functions are meromorphic on $U$ and the multisets of poles located in $U\setminus\{N\}$ of each of these three zeta functions, $\widetilde{\zeta}_{A,\O}$, $\zeta_{A,\O}$ and  $\zeta_{A,\O}^{\mathfrak M}$, coincide$:$
\begin{equation}\label{5.4.35.1/2.1/2}
\po\big(\widetilde{\zeta}_{A,\O},U\setminus\{N\}\big)=\po\big({\zeta}_{A,\O},U\setminus\{N\}\big)=\po\big({\zeta}^{\mathfrak M}_{A,\O},U\setminus\{N\}\big).
\end{equation}
Moreover, if $\omega\in U\setminus\{N\}$ is a simple pole of any of these three zeta functions, then\footnote{Clearly, in the case when $\omega\in U\setminus\{N\}$ is a multiple pole, an analogous relation holds between the principal parts at $\omega$ of $\widetilde{\zeta}_{A,\O}(s)$, $\zeta_{A,\O}^{\mathfrak M}(s)$ and the meromorphic function $\zeta_{A,\O}(s)/(N-s)$.
Also, $\omega$ has the same multiplicity for either of $\widetilde{\zeta}_{A,\O}$, $\zeta_{A,\O}$ or ${\zeta}^{\mathfrak M}_{A,\O}$.}
\begin{equation}
\res\big({\zeta}^{\mathfrak{M}}_{A,\O},\omega\big)=\res\big(\widetilde{\zeta}_{A,\O},\omega\big)=\frac{\res\big({\zeta}_{A,\O},\omega\big)}{N-\omega}.
\end{equation}
\end{lemma}

We may now use the Mellin inversion theorem (Theorem \ref{mellin_inv}) to derive the following inversion formula for the Mellin zeta function.

\begin{theorem}\label{int_form_m}
Let $(A,\O)$ be an RFD in $\eR^N$ such that $\overline{\dim}_B(A,\O)<N$.
Then, for any $c\in(\overline{\dim}_B(A,\O),N)$ and $t>0$, the following formula is valid pointwise$:$
\begin{equation}\label{5.4.35}
|A_t\cap\O|=\frac{1}{2\pi\I}\int_{c-\I\ty}^{c+\I\ty}t^{N-s}{\zeta}^{\mathfrak{M}}_{A,\O}(s)\di s.
\end{equation}
\end{theorem}

\begin{proof}
The conclusion follows directly from Theorem \ref{mellin_inv}, along with the fact that the function $t\mapsto t^{-N}|A_t\cap\O|$ is continuous and of locally bounded variation on $(0,+\ty)$ and $t\mapsto t^{c-N-1}|A_t\cap\O|$ is in $L^1(0,+\ty)$ for every $c\in(\overline{\dim}_B(A,\O),N)$.
(See also the proof of Theorem \ref{tube_inversion} since the reasoning here is analogous.)
\end{proof}

One can now impose languidity conditions on the Mellin zeta function $\zeta_{A,\O}^{\mathfrak M}$ and rewrite the results of Sections \ref{sec_point} and \ref{sec_distr} in terms of $\zeta_{A,\O}^{\mathfrak M}$ since the fact that we have to choose $c\in(\overline{\dim}_B(A,\O),N)$ in the above theorem is not a hindrance.
Indeed, recall that originally, we had the freedom to choose any $c\in(\overline{\dim}_B(A,\O),N+1)$ in Proposition \ref{kth_prim}.
Furthermore, choosing $c\in(\overline{\dim}_B(A,\O),N)$ also ensures that although $s=N$ is always a pole of the Mellin zeta function, it will never be a part of the sum over the residues of $\zeta_{A,\O}^{\mathfrak{M}}$ in the fractal tube formulas since it is always located strictly to the right of the vertical line $\{\re s=c\}$ over which we integrate in \eqref{5.4.35}.

One could now also potentially derive the corresponding results about the fractal tube formulas in terms of the distance zeta function $\zeta_{A,\O}$ directly from the Mellin zeta function $\zeta_{A,\O}^{\mathfrak M}$ and without the use of the shell zeta function $\breve{\zeta}_{A,\O}$.
However, one then has to be careful and always choose $\d$ sufficiently large so that $\O\subseteq A_\d$ in order for \eqref{mellin_dist} to be satisfied.
Another issue that is not fully resolved in this potential alternative approach is whether the restriction of having to choose $\d$ large enough for the inclusion $\O\subseteq A_\d$ to hold could increase the `languidity exponent' $\kappa_d$ of $\zeta_{A,\O}$.
This is not the case in all of the examples we will consider, but a general result along these lines has yet to be obtained.

\begin{proposition}\label{propB}
Let $(A,\O)$ be a relative fractal drum in $\eR^N$.
If the relative distance zeta function ${\zeta}_{A,\O}(\,\cdot\,;\d)$ satisfies the languidity conditions {\bf L1} and {\bf L2} for some $\d>0$ and $\kappa_d\in\eR$, then so does
${\zeta}_{A,\O}(\,\cdot\,;\d_1)$ for any $\d_1>0$ and for ${(\kappa_d)}_{\d_1}:=\max\{\kappa_d,0\}$.

Furthermore, the analogous statement is also true in the case when ${\zeta}_{A,\O}(\,\cdot\,;\d)$ is strongly $d$-languid, under the additional assumption that $\d\geq 1$ and $\d_1\geq 1$.
\end{proposition}

\begin{proof}
Without loss of generality, we may assume that $\d<\d_1$.
Then, the conclusion follows from the fact that for a given a window $\bm W$, we have ${\zeta}_A(s,\O;\d_1)={\zeta}_{A,\O}(s;\d)+g(s)$ for all $s\in W$, where $g$ is defined for all $s\in\Ce$ and is an entire function.
(This fact follows directly from the classical theorem about the holomorphicity of an integral depending on a parameter; see \cite[Theorem 2.1.46]{fzf} or \cite{Matt}.)
Furthermore, for all $s\in\Ce$, we have the following upper bound on $|g(s)|$:
\begin{equation}\label{g_s_upper}
|g(s)|\leq\int_{(A_{\d_1}\setminus \ov{A_{\d}})\cap\O}d(x,A)^{\re s-N}\di x\leq|\O|\max\{\d^{\re s-N},\d_1^{\re s-N}\}.
\end{equation}
As we can see, the upper bound on $|g(s)|$ does not depend on $\im s$ and therefore, we conclude that $g$ satisfies the languidity conditions {\bf L1} and {\bf L2} with the languidity exponent $\kappa_g:=0$ and for any given window $\bm W$.
This observation implies that then, ${\zeta}_{A,\O}(\,\cdot\,;\d_1)$ is languid for $(\kappa_d)_{\d_1}:=\max\{\kappa_d,0\}$ and for the same window as  for ${\zeta}_{A,\O}(\,\cdot\,;\d)$.

The additional assumption about the strong $d$-languidity is needed since {\bf L1} must then be satisfied for all $\sigma\in(-\ty,c)$, in the notation of Definition \ref{languid}.
Furthermore, for this condition to be achieved, we need that $\d_1>\d\geq 1$ in \eqref{g_s_upper} since otherwise, we cannot obtain an upper bound on $|g(s)|$ when $\re s\to -\ty$.
\end{proof}

We refrain from restating all of the theorems of Sections \ref{sec_point} and \ref{sec_distr} in terms of the Mellin zeta function and leave this task for the interested reader.
We will restate only the distributional fractal tube formula with error term, which will be explicitly needed in \cite{ftf_b} (see also \cite{fzf}) for establishing a Minkowski measurability criterion for a large class of RFDs in terms of the location of the principal complex dimensions.
Recall that our original motivation for introducing the Mellin zeta function was to obtain a distributional fractal tube formula valid on a larger space of test functions,
more precisely, on the space $\mathcal{K}(0,+\ty)$; that is, the space of test functions $\varphi$ in the class $C^{\ty}(0,+\ty)$, such that for all $m\in\Ze$ and $q\in\eN$, we have $t^m\varphi^{(q)}(t)\to 0$, as $t\to 0^+$ and $t\to+\ty$.
We point out that the key difference from working with the tube or distance zeta function which enables us to obtain the distributional tube formula in this greater generality is in Theorem \ref{int_form_m}.
More precisely, the integral representation of the tube function $t\mapsto|A_t\cap\O|$ in Theorem \ref{int_form_m} is now valid for all $t>0$, while in Theorem \ref{tube_inversion} we obtained an integral representation valid only for $t\in(0,\d)$.
Hence, the integral representation given in Theorem \ref{int_form_m}, valid for all $t>0$, enables us to work with test functions $\mathcal{K}(0,+\ty)$ when deriving the distributional tube formulas in terms of the Mellin zeta function.
 
Also note that $\mathcal{D}(0,+\ty)\subseteq\mathcal{K}(0,+\ty)$,\label{Dspace+} and hence, we have the following (reverse) relation between the corresponding spaces of distributions (or dual spaces):
\begin{equation}\label{5.4.35.1/2}
\mathcal{K}'(0,+\ty)\subseteq\mathcal{D}'(0,+\ty).
\end{equation}

\begin{theorem}[Distributional fractal tube formula with error term, via $\zeta_{A,\O}^{\mathfrak M}$; level $k=0$]\label{dist_error_mellin}
Let $(A,\O)$ be a relative fractal drum in $\eR^N$ such that $\overline{\dim}_B(A,\O)<N$.
Furthermore, assume that $\zeta_{A,\O}^{\mathfrak M}$ satisfies the languidity conditions for some $\kappa\in\eR$ and $\d>0$.
Then, the regular distribution $\mathcal V^{[0]}_{A,\O}(t):=|A_t\cap\O|$ in $\mathcal{K}'(0,+\ty)$ is given by the following distributional identity in $\mathcal{K}'(0,+\ty)$$:$
\begin{equation}\label{dist_form_error_mellin}
\mathcal V^{[0]}_{A,\O}(t)=\sum_{\omega\in\po({\zeta}_{A,\O}^{\mathfrak{M}},\bm{W})}\res\left({t^{N-s}}{\zeta}_{A,\O}^{\mathfrak{M}}(s),\omega\right)+\mathcal{R}^{{\mathfrak M}[0]}_{A,\O}(t).
\end{equation}
That is, the action of $\mathcal V^{[0]}_{A,\O}$ on an arbitrary test function $\varphi\in\mathcal{K}(0,+\ty)$ is given by
\begin{equation}\label{error_action_mellin}
\begin{aligned}
\big\langle\mathcal V^{[0]}_{A,\O},\varphi\big\rangle&=\sum_{\omega\in\po({\zeta}_{A,\O}^{\mathfrak{M}},\bm{W})}\res\Big({\{\mathfrak{M}\varphi\}(N\!-\! s\!+\! 1)\,{\zeta}_{A,\O}^{\mathfrak{M}}(s)},\omega\Big)+\big\langle\mathcal{R}^{{\mathfrak M}[0]}_{A,\O},\varphi\big\rangle.
\end{aligned}
\end{equation}
Here, the distributional error term $\mathcal{R}^{{\mathfrak M}[0]}_{A,\O}$ is the distribution in $\mathcal{K}'(0,+\ty)$ given for all $\varphi\in\mathcal{K}(0,+\ty)$ by
\begin{equation}\label{R_distr_mellin}
\big\langle\mathcal{R}^{{\mathfrak M}[0]}_{A,\O},\varphi\big\rangle=\frac{1}{2\pi\I}\int_S{\{\mathfrak{M}\varphi\}(N\!-\! s\!+\! 1)\,{\zeta}_{A,\O}^{\mathfrak{M}}(s)}\di s.
\end{equation}
Furthermore, the distribution $\mathcal{R}^{{\mathfrak M}[0]}_{A,\O}(t)$ is of asymptotic order at most $t^{N-\sup S}$ as $t\to 0^+$; i.e.,
\begin{equation}
\mathcal{R}^{{\mathfrak M}[0]}_{A,\O}(t)=O(t^{N-\sup S})\quad\mathrm{ as }\quad t\to0^+,
\end{equation}
in the sense of Definition~\ref{dist_order_dis}.

Moreover, if $S(\tau) < \sup S$ for all $\tau\in\eR$ $($that is, if the screen $\bm S$ lies strictly
to the left of the vertical line $\{\re s =\sup S\}$$)$, then $\mathcal{R}^{{\mathfrak M}[0]}_{A,\O}(t)$ is of asymptotic order less than $t^{N-\sup S}$; i.e.,
\begin{equation}\label{dist_estimate_o_mellin}
\mathcal{R}^{{\mathfrak M}[0]}_{A,\O}(t)=o(t^{N-\sup S})\quad\mathrm{ as }\quad t\to0^+,
\end{equation} 
again in the sense of Definition~\ref{dist_order_dis}.
\end{theorem}

\section{Examples and Applications}\label{exp_app}

In this final section, we illustrate the theory of fractal tube formulas developed in Sections \ref{sec_point}--\ref{distance_tube} by means of several examples of bounded (fractal) sets and relative fractal drums.
These examples include the line segment, the recovery of the known tube formulas (from \cite{lapidusfrank12}) for fractal strings (Subsection \ref{subsec_frstr}), the Sierpi\'nski gasket and the $3$-dimensional Sierpi\'nski carpet, along with the inhomogeneous higher-dimensional $N$-gasket RFDs, with $N\geq 3$ (Subsection \ref{subsec_sier}), a suitable version of the Cantor graph (the `devil's staircase') and an associated discussion of `fractality' expressed in terms of the presence of nonreal complex dimensions (Subsection \ref{subsec_devil}), fractal nests and (unbounded) geometric chirps (Subsection \ref{subsec_nestch}), as well as, finally, the recovery and significant extensions of the known fractal tube formulas (from [LapPe2--3, LapPeWi1--2]) 
for self-similar sprays (Subsection \ref{subsec_self_similar_sp}).

\subsection{The Line Segment, Convex Sets and Smooth Submanifolds}\label{subsec_lisp}

Let us first consider the trivial example of the unit interval in $\eR$, which illustrates the case when we cannot use the distance zeta function in order to recover the tube formula, since $D=N=1$.

\begin{example}\label{segment_example}
Let $I=[0,1]$ be the unit interval in $\eR$.
Then the meromorphic continuations to $\Ce$ of its distance and tube zeta functions are respectively given by
\begin{equation}
\zeta_I(s)=\frac{2\d^s}{s}\quad\textrm{and}\quad\widetilde{\zeta}_I(s)=\frac{2\d^s}{s}+\frac{\d^{s-1}}{s-1},\quad\textrm{for all}\ s\in\Ce.
\end{equation}
As we can see, the distance zeta function fails to provide information about the Minkowski content in this case, because the pole at $s=1$ is canceled by means of the functional equation~\eqref{equ_tilde}.
This also demonstrates why when working with meromorphic extensions of the relative distance zeta function, one must always assume additionally that $\ov{\dim}_{B}(A,\O)<N$, as opposed to the situation with the relative tube zeta function. 
Furthermore, it is clear that $\widetilde{\zeta}_I$ is strongly languid if we choose $\d>1$ for $\kappa:=-1$ and a sequence of screens consisting of the vertical lines $\{\re s=-m\}$, where $m\in\eN$.
We then recover from Theorem~\ref{str_pointwise_formula} (with $k=0$ in the notation of that theorem) the following exact pointwise tube formula:
\begin{equation}\label{interval_tube}
|I_t|={t^{N-0}}\res(\widetilde{\zeta}_I,0)+{t^{N-1}}\res(\widetilde{\zeta}_I,1)=2t+1,
\end{equation}
initially valid for all $t\in(0,\d)$.
Actually, since $\d>1$ may be taken arbitrary large, the exact tube formula \eqref{interval_tube} is valid for all $t>0$.
Note that, of course, it is immediate to check directly that the tube formula \eqref{interval_tube} holds for all $t>0$.

\end{example}

We next explain how to calculate the tube and distance zeta functions, as well as the complex dimensions, of compact convex sets (or, more generally, of sets of positive reach) and smooth compact submanifolds of Euclidean space $\eR^N$, based on key results of Federer \cite{Fed} unifying and extending Steiner's tube formula for convex sets \cite{Stein} and Weyl's tube formula  for compact submanifolds of $\eR^N$ \cite{Wey3} (see also \cite{BergGos} and \cite[Chapter 4]{Schn} for an exposition).

Recall that a closed subset $A$ of $\eR^N$ is said to be of {\em positive reach} if there exists $\d>0$ such that every point $x\in A_{\d}$ has a unique metric projection onto $A$; see \cite{Fed}.
The {\em reach} of $A$, denoted by $\mathrm{reach}(A)$, is then defined as the supremum of all such positive $\d$.
Clearly, every closed convex subset of $\eR^N$ is of infinite (and hence, positive) reach.
Furthermore, if, for instance, $A\subset\eR^2$ is an arc of a circle of radius $r$, then the reach of $A$ is equal to $r$.
Moreover, compact smooth submanifolds of Euclidean space $\eR^N$ are also examples of sets of positive reach.

In the present context, for a compact set $A\subset\eR^N$ of positive reach, it is easy to deduce from the tube formula obtained in \cite{Fed} an explicit expression for $\widetilde{\zeta}_{A}$.\footnote{Relative versions of Theorem \ref{6.1.1/2} are also possible, but we will not consider them here.}

\begin{theorem}\label{6.1.1/2}
Let $A$ be a $($nonempty$)$ compact set of positive reach in $\eR^N$.
Then, for any $\d>0$ such that $0<\delta<\mathrm{reach}(A)$, we have that
\begin{equation}
\widetilde{\zeta}_A(s):=\widetilde{\zeta}_A(s;\d)=\sum_{k=0}^{N}c_k\frac{\d^{s-k}}{s-k},
\end{equation}
where $|A_t|=\sum_{k=0}^{N}c_kt^{N-k}$ for all $t\in(0,\d)$ and the coefficients $c_k$ are the $($normalized$)$ Federer curvatures.
$($From the functional equation \eqref{equ_tilde}, one then deduces at once a corresponding explicit expression for $\zeta_A(s):=\zeta_A(s;\d)$.$)$

Hence, $\dim_BA$ exists and 
\begin{equation}
D:=D(\widetilde{\zeta}_A)=D(\zeta_A)=\dim_BA=\max\{k\in\{0,1,\ldots,N\}:c_k\neq 0\}
\end{equation}
and,\footnote{More precisely, the second equality in Equation \eqref{6.55} holds only if $D<N$.}
\begin{equation}\label{6.55}
\po:=\po(\widetilde\zeta_A)=\po(\zeta_A)\stq\{0,1,\dots,N\}.
\end{equation}
In fact,
\begin{equation}
\po=\big\{k\in\{0,1,\dots,N\}:c_k\ne0\big\}\stq\{k_0,\dots,D\},
\end{equation}
where $k_0:=\min\big\{k\in\{0,1,\dots,D\}:c_k\ne0\big\}$.
Furthermore, each of the complex dimensions of $A$ is simple. 

Finally, if $A$ is such that its affine hull is all of $\eR^N$ $($which is the case when the interior of $A$ is nonempty and, in particular, if $A$ is a convex body$)$, then $D=N$, while if $A$ is a $($smooth$)$ compact $d$-dimensional submanifold of $\eR^N$ $($with $0\le d\le N$$)$, then $D=d$.
\end{theorem}

For example, for the $2$-torus $A\subset\eR^3$, we have $N=3$, $D=2$ (since the Euler characteristic of $A$ is equal to zero), $c_2\neq 0$,\footnote{Note that $c_2$ is just proportional to the the area of the $2$-torus, with the proportionality constant being a standard positive constant.} $c_1= 0$, and hence, $c_0=0$, $k_0=2$ and $\po=\{2\}$, as can also be easily checked via a direct computation.

\subsection{Tube Formulas for Fractal Strings}\label{subsec_frstr}

In the present subsection, we apply our general theory of fractal tube formulas for relative fractal drums (and, in particular, for bounded sets) in $\eR^N$ to the one-dimensional case (i.e., $N=1$) in order to recover the known (pointwise and distributional) fractal tube formulas for fractal strings obtained in \cite{lapidusfrank12}.
We begin by discussing the prototypical example of the Cantor string (viewed as an RFD), in Example \ref{ecant}, and further illustrate our results by means of the well-known example of the $a$-string (in Example \ref{ex_a}).
Along the way, we discuss the case of general fractal strings as well as the associated fractal tube formulas.

\begin{example}\label{ecant}({\em The standard ternary Cantor set and string}).
Let $C$ be the standard ternary Cantor set in $[0,1]$ and fix $\d\geq 1/6$.
Then, it is not difficult to show that the `absolute' distance zeta function of $C$ is meromorphic in all of $\Ce$ and given by
\begin{equation}\label{cant}
\zeta_{C,C_{\d}}(s)=\frac{2^{1-s}}{s(3^s-2)}+\frac{2\d^s}{s},\quad\textrm{for all}\ s\in\Ce,
\end{equation}
where the term $2\d^s/s$ corresponds to the integral over the `outer' neighborhood of the two endpoints $0$ and $1$ (see \cite[Example 3.4]{dtzf} or \cite[Example 2.1.18]{fzf} ).
Consequently, the relative distance zeta function of $(C,(0,1))$ is also meromorphic on all of $\Ce$ and given by
\begin{equation}\label{Cant_zeta}
\zeta_{C,(0,1)}(s)=\frac{2^{1-s}}{s(3^s-2)},\quad\textrm{for all}\ s\in\Ce.
\end{equation}
Furthermore, the sets of complex dimensions of the Cantor set $C$ and of the Cantor string $(C,(0,1))$, viewed as an RFD, coincide:
\begin{equation}\label{comp_dim_CC}
\po({\zeta}_C)=\po({\zeta}_{C,(0,1)})=\{0\}\cup\left(\log_32+\frac{2\pi}{\log3}\I\Ze\right).
\end{equation}
In \eqref{comp_dim_CC}, each of the complex dimensions is simple.
Furthermore, $D:=\dim_B(C,(0,1))$, the Minkowski dimension of the Cantor string, exists and $D=\log_32$, the Minkowski dimension of the Cantor set, which also exists.
Furthermore, $\mathbf{p}:=\frac{2\pi}{\log 3}$ is the oscillatory period of the Cantor set (or string), viewed as a  {\em lattice} self-similar set (or string); see \cite[Chapter 2, esp., Subsection 2.3.1 and Section 2.4]{lapidusfrank12}.

It is clear that $(\lambda C,\lambda(0,1))$ is strongly $d$-languid for $\kappa_d:=-1$, any $\lambda\geq 2$ and a sequence of screens consisting of the vertical lines $\{\re s=-m\}$ for $m\in\eN$, along with the constant $B_{\lambda}:=2/\lambda$ in the strong languidity condition {\bf L2'}.\footnote{Without loss of generality, we can fix $\d\geq 1$ here.}
Theorem~\ref{pointwise_thm_d} (or, really, Theorem \ref{5.3.15.1/2} since all of the complex dimensions of the RFD are simple) then enables us to recover the following exact pointwise fractal formula for the inner $t$-neighborhood of $C$, valid for all $t\in(0,\min\{1/\lambda,1/2\})=(0,1/2)$:
\begin{equation}\label{CC_compute}
\begin{aligned}
|C_t\cap(0,1)|&=\!\!\!\sum_{\omega\in\po({\zeta}_{C,(0,1)})}\!\!\!\res\left(\frac{t^{1-s}}{1-s}{\zeta}_{C,(0,1)}(s),\omega\right)=\!\!\!\sum_{\omega\in\po({\zeta}_{C,(0,1)})}\!\!\!\frac{t^{1-\omega}\res\left({\zeta}_{C,(0,1)},\omega\right)}{1-\omega}\\
&=\frac{1}{2\log 3}\sum_{k=-\ty}^{+\ty}\frac{(2t)^{1-\omega_k}}{(1-\omega_k)\omega_k}-2t=\frac{(2t)^{1-D}}{2\log 3}\sum_{k=-\ty}^{+\ty}\frac{(2t)^{-\I k\mathbf{p}}}{(1-\omega_k)\omega_k}-2t\\
&=t^{1-D}G\left(\log_{3}(2t)^{-1}\right)-2t,
\end{aligned}
\end{equation}
where $\omega_k:=D+\I k\mathbf{p}$ for each $k\in\Ze$, $D:=\dim_B(C,(0,1))=\log_32$ (as  above), and $\mathbf{p}:=\frac{2\pi}{\log 3}$ denote, respectively, the relative Minkowski dimension and the `oscillatory period' of the Cantor string RFD $(C,(0,1))$ in $\eR$ (or, equivalently, of the Cantor string ${\mathcal L}_{CS}$).
Furthermore, $G$ is the positive, nonconstant $1$-periodic function, which is bounded away from zero and infinity and given by the following Fourier series expansion:
\begin{equation}\label{5.5.10.1/4}
G(x):=\sum_{k\in\Ze}\frac{2^{-D}\E^{2\pi\I kx}}{\omega_k(1-\omega_k)\log 3}.
\end{equation}
In \eqref{CC_compute}, the second equality follows from the fact that all of the complex dimensions of $\zeta_{C,(0,1)}$ are simple (see also Theorem \ref{5.3.15.1/2} above), while the third equality is obtained by computing the residues of $\zeta_{C,(0,1)}$ at each $s:=\omega_k$ (for $k\in\Ze$) and at $s=0$;
in particular, we have that
\begin{equation}\label{5.5.10.1/2}
\res\left({\zeta}_{C,(0,1)},\omega_k\right)=\frac{2^{-\omega_k}}{\omega_k\log 3},\quad\textrm{for all}\ k\in\Ze.
\end{equation}

Of course, the above exact pointwise fractal tube formula \eqref{CC_compute} coincides with the one obtained by a direct computation for the Cantor string (see~\cite[Subsection~1.1.2]{lapidusfrank12}) or from the general theory of fractal tube formulas for fractal strings (see \cite[Chapter~8, esp., Sections 8.1 and 8.2]{lapidusfrank12}) and, in particular, for self-similar strings (see, especially, \cite[Subsection~8.4.1, Example 8.2.2]{lapidusfrank12}).\footnote{{\em Caution}: in \cite[Section~8.4]{lapidusfrank}, the Cantor string is defined slightly differently, and hence, $C$ is replaced by $3^{-1}C$.}
Note that the `absolute' tube function $|C_t|$ has the same expression as in \eqref{CC_compute} above but now without the term $-2t$, which is in accordance with~\eqref{cant}.

Finally, observe that, in agreement with the lattice case of the general theory of self-similar strings developed in \cite[Chapters~2--3, and Section 8.4]{lapidusfrank12}, we can rewrite the pointwise fractal tube formula \eqref{CC_compute} as follows (with $D:=\dim_BC=\log_23$):
\begin{equation}\label{5.5.10.3/4}
t^{-(1-D)}V_{C,(0,1)}(t)=t^{-(1-D)}|C_t\cap(0,1)|=G\left(\log_{3}(2t)^{-1}\right)+o(1),
\end{equation}
where $G$ is given by \eqref{5.5.10.1/4}.
Therefore, since $G$ is periodic and nonconstant, it is clear that $t^{-(1-D)}V_{C,(0,1)}(t)$ cannot have a limit as $t\to 0^+$.
It follows that the Cantor string RFD $(C,(0,1))$ (or, equivalently, the Cantor string ${\mathcal L}_{CS}$) is {\em not} Minkowski measurable but (since $G$ is also bounded away from zero and infinity) is Minkowski nondegenerate.
(This was first proved in [LapPo1--2] 
via a direct computation, leading to the precise values of $\mathcal{M}_*$ and $\mathcal{M^*}$, and reproved in \cite[Subsection~8.4.2]{lapidusfrank12} by using either the pointwise fractal tube formulas or a self-similar fractal string Minkowski measurability criterion; which is expanded to the case of $\eR^N$ with $N$ arbitrary in \cite{ftf_b}. 
\end{example}

The above example demonstrates how the theory developed in this paper generalizes (to arbitrary dimensions $N\geq 1$) the corresponding one for fractal strings developed in~\cite[Chapter~8]{lapidusfrank12}.\footnote{One should slightly qualify this statement, however, because the higher-dimensional counterpart of the theory of fractal tube formulas for self-similar strings developed in \cite[Section~8.4]{lapidusfrank12} is not developed in this paper in the general case of self-similar RFDs (and, for example, of self-similar sets satisfying the open set condition), except in the special case of self-similar sprays discussed in Subsection~\ref{subsec_self_similar_sp} below.}
More generally, the following result gives a general connection between the geometric zeta function of a nontrivial bounded fractal string $\mathcal{L}=(\ell_j)_{j\geq 1}$ and the (relative) distance zeta function of the bounded subset of $\eR$ given by
\begin{equation}
A_{\mathcal{L}}:=\bigg\{a_k:=\sum_{j\geq k}\ell_j:k\geq 1\bigg\}
\end{equation}
or, more specifically, of the RFD $(A_{\mathcal L},(0,\ell))$.

\begin{proposition}\label{geo_dist}
Let $\mathcal{L}=(\ell_j)_{j\geq 1}$ be a nontrivial bounded fractal string and let $\ell:=\zeta_{\mathcal{L}}(1)=\sum_{j=1}^{\ty}\ell_j$ denote its total length.
Then, for every $\delta\geq \ell_1/2$, we have the following functional equation for the distance zeta function of the relative fractal drum $(A_{\mathcal{L}},(0,\ell))$$:$
\begin{equation}\label{geo_equ}
\zeta_{A_{\mathcal{L}},(0,\ell)}(s;\d)=\frac{2^{1-s}\zeta_{\mathcal{L}}(s)}{s},
\end{equation}
valid on any connected open neighborhood $U\subseteq\Ce$ of the critical line $\{\re s=\ov{\dim}_B(A_{\mathcal{L}},(0,\ell))\}$ to which any of the two fractal zeta functions $\zeta_{{A_{\mathcal{L}}},(0,\ell)}$ and $\zeta_{\mathcal{L}}$ possesses a meromorphic continuation.\footnote{If we do not require that $\delta\geq \ell_1/2$, then we have that $\zeta_{A_{\mathcal{L}}}(s;\d)={2^{1-s}s^{-1}\zeta_{\mathcal{L}}(s)}+v(s)$, where $v$ is holomorphic on $\{\re s>0\}$.
On the other hand, for applying the theory, we may restrict ourselves to the case when $\delta\geq \ell_1/2$.}

Furthermore, if $\zeta_{\mathcal{L}}$ is languid for some $\kappa_{\mathcal{L}}\in\eR$, then $\zeta_{A_{\mathcal{L}},(0,\ell)}(\,\cdot\,;\d)$ is $d$-languid for $\kappa_d:=\kappa_{\mathcal{L}}-1$, with any $\d\geq \ell_1/2$.

Moreover, if $\zeta_{\mathcal{L}}$ is strongly languid, then so is $\zeta_{\lambda A_{\mathcal{L}},(0,\lambda \ell)}(s;\d\lambda)$ for any $\lambda\geq 2$ and any $\d\geq \ell_1/2$.
\end{proposition}

\begin{proof}
The functional equation~\eqref{geo_equ} is derived in \cite[Example 2.9 and Theorem 2.10]{dtzf} (see also \cite[Example 2.1.57]{fzf}) and can be easily obtained directly from the appropriate definitions.
Furthermore, the statements about the languidity follow directly from the definition.
\end{proof}

\begin{remark}\label{5.5.4.1/4}
There is nothing special about the bounded set $A_{\mathcal L}\subset\eR$ associated with $\mathcal L$. 
In fact, in the statement of Proposition \ref{geo_dist}, we could replace $A_{\mathcal L}$ with $\partial\O$, where the bounded open set $\O\subset\eR$ is an arbitrary geometric realization of the fractal string $\mathcal L$.
Similarly, in recovering the fractal tube formulas for fractal strings obtained in \cite[Chapter 8]{lapidusfrank12}, one can use $\zeta_{\partial\O,\O}:=\zeta_{\partial\O,\O}(\,\cdot\,;\d)$ instead of $\zeta_{A_{\mathcal L},(0,\ell)}:=\zeta_{A_{\mathcal L},(0,\ell)}(\,\cdot\,;\d)$.
This is precisely what we will do in the subsequent discussion.
\end{remark}

Let $\pa\O$ be the boundary of $\O$, where the bounded open set $\O\subset\eR$ is any geometric realization of the bounded (nontrivial) fractal string $\mathcal L$ such that $\ov{\dim}_B(\pa\O,\O)<1$. 
That is, we can write $\O$ as a disjoint union of bounded open intervals $I_j$ (i.e., the connected components of $\O$) such that $I_j$ has length $\ell_j$, for each $j\geq 1$.
It is unimportant in which order the intervals $I_j$ are arranged.
Then, under suitable hypotheses (namely, we assume that either $\zeta_{\pa\O,\O}$ or $\zeta_{\mathcal{L}}$ has a meromorphic continuation to a connected open neighborhood of the critical line $\{\re s=\ov{\dim}_B(\pa\O,\O)\}$), we have (much as in \eqref{geo_equ} above) 
the following key {\em functional equation connecting the distance zeta function $\zeta_{\pa\O,\O}$ of the RFD $(\pa\O,\O)$ and the geometric zeta function $\zeta_{\mathcal{L}}$ of the fractal string} $\mathcal{L}:=(\ell_j)_{j=1}^{\ty}$:\footnote{We note that the functional equation \eqref{5.5.15.1/2D} is valid, without any hypothesis on the bounded fractal string $\mathcal{L}$ (or on its distance and geometric zeta functions), for all $s\in\Ce$ with $\re s$ sufficiently large (namely, for $\re s>\ov{D}$, where $\ov{D}:=\dim_B(\pa\O,\O)=D(\zeta_{\pa\O,\O})=D(\zeta_{\mathcal{L}})$).}
\begin{equation}\label{5.5.15.1/2D}
\zeta_{\pa \O,\O}(s)=\frac{2^{1-s}\zeta_{\mathcal L}(s)}{s},
\end{equation}
valid for all $s\in U$.
(Of course, it then follows that each of the two fractal zeta functions $\zeta_{\pa\O,\O}$ and $\zeta_{\mathcal{L}}$ has a unique meromorphic continuation to all of $U$.)

Consequently, by choosing $U:=\mathring{W}$ to be the interior of a suitable window $\bm W$ (with an associated screen $\bm S$), we deduce from the results of Section \ref{distance_tube} (especially, Subsections \ref{subsec_point_dist} and \ref{subsec_distr_dist}) that the tube function
\begin{equation}
V_{\mathcal L}(t):=|\{x\in\O\,:\,d(x,\pa\O)<t\}|_1=V_{\pa\O,\O}(t)
\end{equation}
can be expressed via the following fractal tube formula (with or without error term and pointwise or distributionally, depending on the assumptions), for every $\d\geq\ell_1/2$:\footnote{Namely, the hypotheses of Theorem \ref{pointwise_thm_d} (i.e., of Theorem \ref{pointwise_formula_rd} at level $k=0$), for the pointwise tube formula, or else, the hypotheses of Theorem \ref{dist_tube_formula_d} (i.e., of Theorem \ref{dist_error_d} at level $k=0$), for the distributional tube formula.}
\begin{equation}\label{5.5.12.1/4}
\begin{aligned}
V_{\mathcal L}(t)&=V_{\pa\O,\O}(t)=\sum_{\omega\in\po(\zeta_{\pa\O,\O},\bm{W})}\res\left(\frac{t^{1-s}}{1-s}\zeta_{\pa\O,\O}(s),\omega\right)+R^{[0]}_{\pa\O,\O}(t)\\
&=\sum_{\omega\in\po(\zeta_{\pa\O,\O},\bm{W})}\res\left(\frac{(2t)^{1-s}}{s(1-s)}\zeta_{\mathcal L}(s),\omega\right)+R^{[0]}_{\pa\O,\O}(t),
\end{aligned}
\end{equation}
where, in the languid case, we have the error estimate $R^{[0]}_{\pa\O,\O}(t)=O(t^{1-\sup S})$ as $t\to 0^+$ or $R^{[0]}_{\pa\O,\O}(t)=o(t^{1-\sup S})$ as $t\to 0^+$ (also depending on the hypotheses),\footnote{More specifically, in order to obtain the better  error estimate, we also have to assume that the screen is strictly to the left of the vertical line $\{\re s=\sup S\}$.} or else, $R^{[0]}_{\pa\O,\O}(t)\equiv 0$ and $\bm{W}:=\Ce$ in the strongly languid case.
Here, $\po(\zeta_{\pa\O,\O},\bm{W})$ denotes the set of visible complex dimensions of $(\pa\O,\O)$, visible through a given window $\bm W$ (with an associated screen $\bm S$), and in light of the counterpart for the RFD $(\pa \O,\O)$ of Equation \eqref{geo_equ} along with Remark~\ref{5.5.4.1/4}, we have that
\begin{equation}\label{5.5.12.1/2}
\po(\zeta_{\pa\O,\O},W\setminus\{0\})=\po(\zeta_{\mathcal L},W\setminus\{0\}),
\end{equation}
where the equality holds between multisets.
Furthermore, if $0\in W$ and if $\zeta_{\mathcal L}(0)$ is defined and not equal to zero (i.e., if $\zeta_{\mathcal L}(0)\neq 0$), then, $0\in\po(\zeta_{\pa\O,\O},\bm{W})$ and it has multiplicity one.
On the other hand, if $0\in\po(\zeta_{\mathcal L},\bm{W})$ and is a pole of multiplicity $m$ for some $m\in\eN$, then, $0\in\po(\zeta_{\pa\O,\O},\bm{W})$ and it has multiplicity $m+1$.
In other words, we have the following equality between multisets:
\begin{equation}\label{5.5.18.1/2D}
\po(\zeta_{\pa\O,\O},\bm{W})=\po(\zeta_{\mathcal L},\bm{W})\cup\{0\}_{0\in W,\ \zeta_{\mathcal L}(0)\neq 0},
\end{equation}
where $\{0\}_{0\in W,\ \zeta_{\mathcal L}(0)\neq 0}$ is equal to $\{0\}$ if $0\in W$ and $\zeta_{\mathcal L}(0)\neq 0$, and to the empty set otherwise.

If, in addition, each of the visible complex dimensions of $(\pa\O,\O)$ (i.e., each pole of $\zeta_{\pa\O,\O}$ in $\bm W$) is simple, then (in light of \eqref{5.5.15.1/2D}) the fractal tube formula \eqref{5.5.12.1/4} takes the following simpler form:
\begin{equation}\label{5.5.12.3/4}
\begin{aligned}
V_{\mathcal L}(t)&=V_{\pa\O,\O}(t)\\
&=\sum_{\omega\in\po(\zeta_{\mathcal L},\bm{W})}\frac{(2t)^{1-\omega}}{\omega(1-\omega)}\res\big(\zeta_{\mathcal L}(s),\omega\big)+\left\{2t\zeta_{\mathcal L}(0)\right\}_{0\in W}+R^{[0]}_{\pa\O,\O}(t),
\end{aligned}
\end{equation}
where the (pointwise or distributional) error term $R^{[0]}_{\pa\O,\O}(t)$ is estimated as above (in the languid case) or else, $R^{[0]}_{\pa\O,\O}(t)\equiv 0$ and $\bm{W}:=\Ce$ (in the strongly languid case).
Here, provided $\zeta_{\mathcal{L}}(0)$ is well defined, the term $\{2t\zeta_{\mathcal{L}}(0)\}_{0\in W}$ is equal to zero if $0\notin W$ and to $2t\zeta_{\mathcal{L}}(0)$ if $0\in W$.
If, however, $0$ is a simple, visible pole of $\zeta_{\mathcal L}$, then we should replace $\{2t\zeta_{\mathcal{L}}(0)\}_{0\in W}$ on the right-hand side of \eqref{5.5.12.3/4} with the term
\begin{equation}
2t\big(1-\log(2t)\big)\res(\zeta_{\mathcal L},0)+2t\zeta_{\mathcal L}[0]_0,
\end{equation}
where $\zeta_{\mathcal L}[0]_0$ stands for the constant term in the Laurent series expansion of $\zeta_{\mathcal L}$ around $s=0$.
This is in agreement with \cite[Corollary~8.3]{lapidusfrank12} (resp., \cite[Corollary~8.10]{lapidusfrank12}) in the case of a distributional (resp., pointwise) fractal tube formula.

Note that in light of \eqref{5.5.12.1/2},
formula \eqref{5.5.12.1/4} can be rewritten as follows, in terms of the set $\po(\zeta_{\mathcal L},\bm{W})$ of all visible poles of $\zeta_{\mathcal L}$ (see also Remark \ref{5.5.6.1/2R} below):
\begin{equation}\label{5.5.12.4/5}
\begin{aligned}
V_{\mathcal L}(t)&=V_{\pa\O,\O}(t)=\sum_{\omega\in\po(\zeta_{\mathcal L},\bm{W})}\res\left(\frac{(2t)^{1-s}}{s(1-s)}\zeta_{\mathcal L}(s),\omega\right)\\
&\phantom{=}+\left\{2t\zeta_{\mathcal L}(0)\right\}_{0\in W\setminus\po(\zeta_{\mathcal L},\bm{W})}+R^{[0]}_{\pa\O,\O}(t),
\end{aligned}
\end{equation}
which is in agreement with \cite[Theorem~8.1]{lapidusfrank12} (resp., \cite[Theorem~8.7]{lapidusfrank12}) in the case of a distributional (resp., pointwise) fractal tube formula.

Naturally, $\po(\zeta_{\mathcal{L}},\bm{W})$ is viewed as a multiset; that is, on the right-hand side of \eqref{5.5.12.1/2} or \eqref{5.5.18.1/2D}, each visible `scaling complex dimension' $\omega\in\po(\zeta_{\mathcal{L}},\bm{W})$ (i.e., each visible pole of the geometric zeta function $\zeta_{\mathcal{L}}$) occurs according to its multiplicity.
An entirely analogous comment can be made about the multiset  $\po(\zeta_{\pa\O,\O},\bm{W})$ and the associated visible complex dimensions $\omega\in\po(\zeta_{\pa\O,\O},\bm{W})$.


\begin{remark}\label{5.5.6.1/2R}
In \cite{lapidusfrank12}, the elements of $\po(\zeta_{\mathcal L},\bm{W})$ are called the (visible) complex dimensions of $\mathcal L$.
In the present paper, the relationship with the actual (visible) complex dimension of the RFD $(\pa\O,\O)$ (i.e., the visible poles of $\zeta_{\pa\O,\O}$) is given by \eqref{5.5.12.1/2} and the text surrounding it.
Much as in [LaPe2--3, LapPeWi1--2] 
and \cite[Section 13.1]{lapidusfrank12}, we propose to refer to the elements of $\po(\zeta_{\mathcal{L}},\bm{W})$ (i.e., to the visible poles of the geometric zeta function $\zeta_{\mathcal{L}}$) as the visible {\em scaling complex dimensions} of the fractal string $\mathcal{L}$.
Similarly, $\zeta_{\mathcal{L}}$ will also be occasionally referred to as the {\em scaling zeta function} of $\mathcal{L}$ (or rather, of the associated RFD $(\pa\O,\O)$) and denoted by $\zeta_{\mathfrak{S}}$. 
\end{remark}

\begin{remark}\label{5.5.4.3/4}
We leave it as an easy exercise for the interested reader to use the counterpart for the RFD $(\pa\O,\O)$ of the functional equation \eqref{geo_equ} in Proposition \ref{geo_dist} in order to express the languidity, as well as the strong languidity conditions, in terms of the geometric zeta function $\zeta_{\mathcal L}$ instead of the distance zeta function $\zeta_{\pa\O,\O}$.
Furthermore, the reader can easily check that the results of Example \ref{ecant} concerning the Cantor string
$
\mathcal{L}:=\left\{1,\frac{1}{3},\frac{1}{3},\frac{1}{9},\frac{1}{9},\frac{1}{9},\frac{1}{9},\ldots\right\}
$
(see, especially, Equation \eqref{CC_compute}) are compatible with both \eqref{5.5.12.3/4} and \eqref{5.5.12.4/5}.
Indeed, in light of \eqref{Cant_zeta} and \eqref{geo_equ}, we have (for all $s\in\Ce$)
\begin{equation}\label{5.5.12.6/7}
\zeta_{CS}(s)=\frac{1}{3^s-2},
\end{equation}
from which it follows that $\zeta_{CS}(0)=-1$ and (with $\bm{W}:=\Ce$) the term $\{2t\zeta_{\mathcal L}(0)\}_{0\in W}$ in both \eqref{5.5.12.3/4} and \eqref{5.5.12.4/5} becomes $-2t$, in agreement with \eqref{CC_compute}.
\end{remark}

\begin{example}\label{ex_a}({\em The $a$-string}).
For a given $a>0$, the $a$-string $\mathcal{L}_{a}$ can be realized as the bounded open set $\O_a\subset\eR$ obtained by removing the points $j^{-a}$ for $j\in\eN$ from the interval $(0,1)$; that is, 
\begin{equation}\label{O_a}
\O_a=\bigcup_{j=1}^{\ty}\big((j+1)^{-a},j^{-a}\big),
\end{equation}
so that the sequence of lengths of $\mathcal{L}_a$ is defined by
\begin{equation}\label{lj}
\ell_j:=j^{-a}-(j+1)^{-a},\ \textrm{ for }\ j=1,2,\ldots,
\end{equation}
and $\pa\O_a=\{j^{-a}\,:\,j\geq 1\}\cup\{0\}=A_{\mathcal{L}_a}\cup\{0\}$.
Hence, its geometric zeta function is given (for all $s\in\Ce$ such that $\re s>\dim_B\mathcal{L}_a$) by
$$
\zeta_{{\mathcal{L}}_a}(s)=\sum_{j=1}^{\ty}\ell_j^s=\sum_{j=1}^{\ty}\big(j^{-a}-(j+1)^{-a}\big)^s
$$
and it then follows from Proposition \ref{geo_dist} that for $\d>(1-2^{-a})/2$, its distance zeta function is given by 
\begin{equation}\label{a-dist}
\zeta_{A_{\mathcal{L}_a},(0,1)}(s;\d)=\frac{\zeta_{\mathcal{L}_a}(s)}{2^{s-1}s}=\frac{1}{2^{s-1}s}\sum_{j=1}^{\ty}\big(j^{-a}-(j+1)^{-a}\big)^s,
\end{equation}
where the second equality holds for all $s\in\Ce$ such that $\re s>\dim_B\mathcal{L}_a$ while the first equality holds for all $s\in\Ce$ (since, as will be recalled just below, $\zeta_{\mathcal{L}_a}$ and hence also $\zeta_{A_{\mathcal{L}_a},(0,1)}$, admits a meromorphic extension to all of $\Ce$).

Furthermore, the properties of the geometric zeta function $\zeta_{\mathcal{L}_a}$ of the $a$-string are well-known (see \cite[Theorem~6.21]{lapidusfrank12}).
Namely, $\zeta_{\mathcal{L}_a}$ has a meromorphic continuation to the whole of $\Ce$ and its poles in $\Ce$ are located at 
\begin{equation}
D:=\dim_B\mathcal{L}_a=\dim_BA_{\mathcal{L}_a}=\frac{1}{a+1}
\end{equation}
and at (a subset of) $\{-\frac{m}{a+1}:m\in\eN\}$.
Furthermore, all of its poles are simple and $\res(\zeta_{\mathcal{L}_a},D)=Da^D$.\footnote{In \cite[Theorem~6.21]{lapidusfrank12}, it is stated that $\res(\zeta_{\mathcal{L}_a},D)=a^D$, which is a misprint.
More specifically, in the proof of that theorem, the source of the misprint is the fact that the residue of $\zeta((a+1)s)$ at $s=1/(a+1)$ is equal to $1/(a+1)$ and not to $1$. Here, $\zeta$ is the Riemann zeta function.\label{25}}
Moreover, for any screen $\bm S$ not passing through a pole, the function $\zeta_{\mathcal{L}_a}$ satisfies {\bf L1} and {\bf L2} with $\kappa:=\frac{1}{2}-(a+1)\inf S$, if $\inf S\leq 0$ and $\kappa:=\frac{1}{2}$ if $\inf S\geq 0$.
From these facts and Equation \eqref{a-dist}, we conclude that the set $A_{\mathcal{L}_a}$ is $d$-languid with $\kappa_d:=-\frac{1}{2}-(a+1)\inf S$ if $\inf S\leq 0$ and with $\kappa_d:=-\frac{1}{2}$ if $\inf S\geq 0$.
For $M\in\eN\cup\{0\}$, we can now choose the screen $\bm{S}_M$ to be some vertical line between $-\frac{M+1}{1+a}$ and $-\frac{M+2}{1+a}$ and let $\bm{W}_M$ be the corresponding window.
Applying Theorem~\ref{dist_tube_formula_d}, we now obtain the following asymptotic distributional formula for the tube function $t\mapsto|(A_{\mathcal{L}_a})_t\cap(0,1)|$ when $t\to0^+$:
\begin{equation}
|(A_{\mathcal{L}_a})_t\cap(0,1)|=\!\!\sum_{\omega\in\po({\zeta}_{A_{{\mathcal{L}_a}}},W_M)}\!\!\res\left(\frac{t^{1-s}}{1-s}\zeta_{A_{\mathcal{L}_a}}(s;\d),\omega\right)+O(t^{1-\sup{S_M}}).
\end{equation}
More specifically, since we know that all the poles are simple and $\zeta_{\mathcal{L}_a}(0)=-1/2$ (see \cite[p.~205]{lapidusfrank12}),
we have that
\begin{equation}
\begin{aligned}
\res(\zeta_{A_{\mathcal{L}_a}},D)&={2^{1-D}}D^{-1}\res(\zeta_{{\mathcal{L}_a}},D)=2^{1-D}a^D,\ 
\res(\zeta_{A_{\mathcal{L}_a}},0)=2\zeta_{\mathcal{L}_a}(0)=-1.
\end{aligned}
\end{equation}
Consequently, and in agreement with the discussion following Proposition \ref{geo_dist} in the special case of simple complex dimensions (see, especially, Equation \eqref{5.5.12.3/4} above), we have that
\begin{equation}\label{a-str-t}
\begin{aligned}
|(A_{\mathcal{L}_a})_t\cap(0,1)|&=\frac{2^{1-D}a^D}{1-D}t^{1-D}-t-\sum_{m=1}^M\frac{\res\left(\zeta_{{\mathcal{L}_a}},-mD\right)(2t)^{1+mD}}{(1+mD)mD}\\
&+O\big(t^{1+(M+1)D}\big),\quad\textrm{as}\quad t\to0^+,
\end{aligned}
\end{equation}
where the sum is interpreted as being equal to $0$ if $M=0$.
In particular, since by choosing a screen to the right of $-D/2$, we conclude that \eqref{a-str-t} is actually valid pointwise because in that case $\kappa_d<0$ (see Theorem~\ref{pointwise_thm_d}), the above (pointwise) tube formula for such a screen then implies that $\dim_BA_{\mathcal{L}_a}=D$ (as was stated above), 
and that the $a$-string is Minkowski measurable with Minkowski content given by
\begin{equation}
\mathcal{M}^D(A_{\mathcal{L}_a})=\frac{2^{1-D}a^D}{1-D},
\end{equation}
as was first established in \cite[Example~5.1]{Lap1} and later reproved in [LapPo1--2] 
via a general Minkowski measurability  criterion for fractal strings (expressed in terms of the asymptotic behavior of $(\ell_j)_{j=1}^{\ty}$, here, $\ell_j\sim aj^{-1/D}$ as $j\to\ty$) and then, in [Lap-vFr1--3] 
(via the the theory of complex dimensions of fractal strings). 
We point out that \eqref{a-str-t} coincides with the `inner' tube formula of the $a$-string (see \cite[Subsection~8.1.2]{lapidusfrank12}).\footnote{More precisely, the two expressions coincide after we have taken into account the misprint mentioned in footnote \ref{25} and add the term $2\zeta_{\mathcal{L}}(0)$ which seems to be forgotten in \cite{lapidusfrank12}.}
\end{example}


\subsection{The Sierpi\'nski Gasket and $3$-Carpet}\label{subsec_sier}

In this subsection, we provide an exact, pointwise fractal tube formula for the Sierpi\'nski gasket (Example \ref{gsk_fract}) and for a three-dimensional analog of the Sierpi\'nski carpet (Example \ref{ex2}).
Naturally, although the required computation involved is somewhat more complicated, one could similarly derive from our general results in Section \ref{distance_tube} (and by using, in particular, the results of \cite[Example 4.1.4]{refds} or \cite{fzf} concerning the inhomogeneous Sierpi\'nski $N$-gasket RFD) exact, pointwise fractal tube formulas for the $N$-dimensional analogs of the Sierpi\'nski gasket and carpet, with $N\geq 2$ arbitrary.
We leave it to the interested reader to carry out the corresponding detailed computations and to imagine other (two- or higher-dimensional) examples of self-similar fractal sets or self-similar RFDs which can be dealt with explicitly within the present general theory of (higher-dimensional) fractal tube formulas.\footnote{The authors have recently obtained an explicit fractal tube formula for the Koch drum (or the Koch RFD).
This important example should be discussed in a later work and its conclusions compared with those of \cite{lappe1} (as discussed in \cite[Subsection 12.2.1]{lapidusfrank12}).}
The example of the $3$-carpet discussed in detail in Example \ref{ex2} below should give a good idea as to how  to proceed in other, related situations, including especially for the higher-dimensional inhomogeneous $N$-gasket RFDs (with $N\geq 4$).

\begin{example}({\em The Sierpi\'nski gasket}).\label{gsk_fract}
Let $A$ be the Sierpi\'nski gasket in $\eR^2$, constructed in the usual way inside the unit triangle.
Furthermore, we assume without loss of generality that $\d>1/4\sqrt{3}$, so that $A_\d$ be simply connected.
Then, the distance zeta function $\zeta_A$ of the Sierpi\'nski gasket is meromorphic on the whole complex plane and is given by
\begin{equation}
\zeta_A(s;\d)=\frac{6(\sqrt3)^{1-s}2^{-s}}{s(s-1)(2^s-3)}+2\pi\frac{\d^s}s+3\frac{\d^{s-1}}{s-1},
\end{equation}
for all $s\in\Ce$ (see \cite[Proposition 3.2.3]{fzf} or \cite[Example 3.27]{refds}). In particular, the set of complex dimensions of the Sierpi\'nski gasket is given by
\begin{equation}
\po({\zeta}_{A}):=\po(\zeta_A,\Ce)=\{0,1\}\cup\left(\log_23+\frac{2\pi}{\log 2}\I\Ze\right),
\end{equation}
with each complex dimension being simple.

By letting $\omega_k:=\log_23+\I k\mathbf{p}$ (for each $k\in\Ze$) and $\mathbf{p}:=2\pi/\log 2$, we have that
\begin{equation}
\res(\zeta_A,\omega_k)=\frac{6(\sqrt3)^{1-\omega_k}}{4^{\omega_k}(\log2)\omega_k(\omega_k-1)} \quad\textrm{for all}\ k\in\Ze,
\end{equation}
\begin{equation}
\res(\zeta_A,0)=3\sqrt{3}+2\pi,\quad\textrm{and}\quad\res(\zeta_A,1)=0.
\end{equation}
Similarly as in Example \ref{ecant}, one can check that $\zeta_{\lambda A}(\,\cdot\,;\d\lambda)$ is strongly languid with $\kappa_d:=-1$ for every $\d\geq 1/2\sqrt{3}$ and any $\lambda\geq 2\sqrt{3}$; so that we can apply Theorem~\ref{pointwise_thm_d} (or, more specifically, its corollary given in Theorem~\ref{5.3.15.1/2} at level $k=0$ and in the case of simple poles) in order to obtain the following exact pointwise fractal tube formula:
\begin{equation}\nonumber
\begin{aligned}
|A_t|&=\sum_{\omega\in\po({\zeta}_{A})}\res\left(\frac{t^{2-s}}{2-s}{\zeta}_A(s;\d),\omega\right)\\
&=t^{2-\log_23}\,\frac{6\sqrt{3}}{\log 2}\sum_{k=-\ty}^{+\ty}\frac{(4\sqrt{3})^{-\omega_k}t^{-\I k\mathbf{p}}}{(2-\omega_k)(\omega_k-1)\omega_k}+\left(\frac{3\sqrt{3}}{2}+\pi\right)t^2,\\
\end{aligned}
\end{equation}
valid for all $t\in(0,1/2\sqrt{3})$.
Note that this formula coincides with the one obtained in~\cite{lappe2} and \cite{lappewi1} and, more recently, via a different (but related) technique in~\cite{DeKoOzUr}.
\end{example}

\begin{figure}[ht]
\begin{center}
\includegraphics[trim=1.5cm 2.5cm 2.5cm 1.5cm,clip=true,width=6cm]{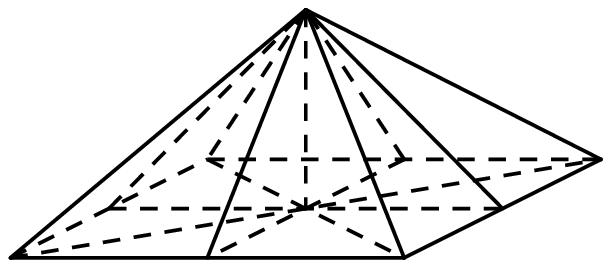}
\includegraphics[trim=6cm 3cm 9cm 2cm,clip=true,width=6cm]{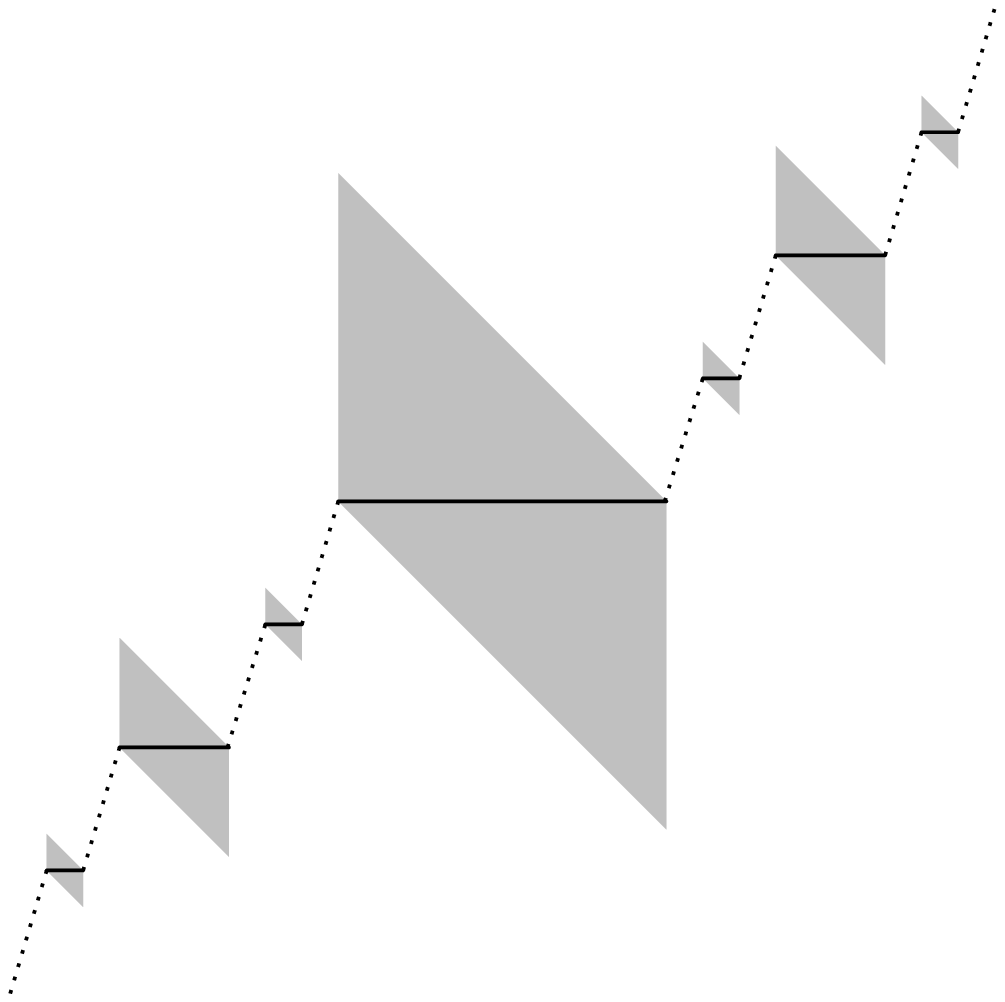}
\end{center}
\caption{{\bf Left:} The mutually congruent pyramids into which we subdivide the cube $A_1$ from Example \ref{ex2}. Eight of them, corresponding to one face of $A_1$, are shown here. {\bf Right:} The third step in the construction of the Cantor graph relative fractal drum $(A,\O)$ from Example~\ref{stair}.
One can see, in particular, the sets $B_k$, $\triangle_k$ and $\widetilde{\triangle}_k$ for $k=1,2,3$.}
\label{piramide}
\end{figure}

\begin{example}({\em The $3$-carpet}).\label{ex2}
Let $A$ be the three-dimensional analog of the Sierpi\'nski carpet.
More specifically, we construct $A$ by dividing the closed unit cube of $\eR^3$ into $27$ congruent cubes and remove the open middle cube.
Then, we iterate this step with each of the $26$ remaining smaller closed cubes; and so on, ad infinitum.
By choosing $\d>1/6$, we have that $A_\d$ is simply connected.
Let us now calculate the distance zeta function $\zeta_A$ of the three-dimensional carpet $A$.
Note that 
$$
\zeta_A(s;\d)=\zeta_{A,I}(s)+\zeta_{A,A_\d\setminus I}(s),
$$
where $I$ denotes the closed unit cube in $\eR^3$.
Let us denote by $B_1$ the open unit cube of side $1/3$ removed in the first step of the construction; so that we have the following equalities:
\begin{equation}\label{5542}
\zeta_{A,I}(s)=\zeta_{A,B_1}(s)+\zeta_{A,I\setminus B_1}(s)=\zeta_{\partial B_1,B_1}(s)+26\,\zeta_{3^{-1}A,3^{-1}I}(s),
\end{equation}
for all $s\in\Ce$ with $\re s$ sufficiently large.
The first equality is obvious, while the second equality in \eqref{5542} follows from the self-similarity of $A$.
More precisely, it follows since the relative fractal drum $(A,I\setminus B_1)$ consists of 26 copies of $(A,I)$ scaled down by $3^{-1}$.
Hence, by the scaling property of the relative distance zeta function (see \cite[Theorem 4.1.38]{fzf} or \cite[Section 2.2]{mefzf}), we have that
\begin{equation}\nonumber
\zeta_{A,I}(s)=\zeta_{\partial B_1,B_1}(s)+26\cdot 3^{-s}\zeta_{A,I}(s),
\end{equation}
which yields
\begin{equation}\label{par1}
\zeta_{A,I}(s)=\frac{\zeta_{\partial B_1,B_1}(s)}{1-26\cdot 3^{-s}},
\end{equation}
for all $s\in\Ce$ with $\re s$ sufficiently large.
The distance zeta function $\zeta_{\partial B_1,B_1}$ can be easily calculated by dividing the cube $B_1$ into $48$ mutually congruent pyramids (see Figure \ref{piramide}, left) and then integrating in local Cartesian coordinates $(x,y)\in\eR^2$ over each resulting pyramid:
\begin{equation}\label{par2}
\zeta_{\partial B_1,B_1}(s)=48\int_0^{1/6}\di x\int_0^x\di y\int_0^yz^{s-3}\di z=\frac{48\cdot 6^{-s}}{s(s-1)(s-2)},
\end{equation}
valid for all $s\in\Ce$ such that $\re s>2$.
On the other hand, the distance zeta function $\zeta_{A,A_\d\setminus I}(s)$ corresponding to the `outside' of the unit cube $I$ is easy to calculate once we have subdivided the parts that correspond to the faces, edges and vertices of the unit cube and used local Cartesian, cylindrical and spherical coordinates in $\eR^3$, respectively:
\begin{equation}
\begin{aligned}
\zeta_{A,A_\d\setminus I}(s)&=6\int_0^1\di x\int_0^1\di y\int_0^{\d}z^{s-3}\di z+12\int_0^{\pi/2}\di\varphi\int_0^{\d}r^{s-2}\di r\int_0^1\di z\\
&\phantom{=}+8\int_0^{\pi/2}\sin\theta\di\theta\int_0^{\pi/2}\di\varphi\int_0^{\d}r^{s-1}\di r=\frac{6\d^{s-2}}{s-2}+\frac{6\pi\d^{s-1}}{s-1}+\frac{4\pi\d^s}{s},
\end{aligned}
\end{equation} 
again valid for all $s\in\Ce$ such that $\re s>2$.
From the above calculation and from \eqref{par1} together with \eqref{par2}, we deduce that $\zeta_A$ can be meromorphically continued to all of $\Ce$ and is then given by
\begin{equation}
\zeta_A(s):=\zeta_A(s;\d)=\frac{48\cdot 2^{-s}}{s(s-1)(s-2)(3^s-26)}+\frac{4\pi\d^s}{s}+\frac{6\pi\d^{s-1}}{s-1}+\frac{6\d^{s-2}}{s-2},
\end{equation}
for every $s\in\Ce$.

It follows that the set of complex dimensions of the $3$-carpet $A$ is given by
\begin{equation}\label{3-carp_po}
\po(\zeta_A):=\po({\zeta}_A,\Ce)=\{0,1,2\}\cup\big(\log_326+\mathbf{p}\I\Ze\big),
\end{equation}
where $D:=\log_326\ (=D(\zeta_A))$ is the Minkowski (or box) dimension of the $3$-carpet $A$ and $\mathbf{p}:=2\pi/\log 3$ is the oscillatory period of $A$ (viewed as a lattice self-similar set).
In \eqref{3-carp_po}, each of the complex dimensions is simple.
Furthermore, a routine computation shows that
\begin{equation}
\res(\zeta_A,0)=4\pi-\frac{24}{25},\q \res(\zeta_A,1)=6\pi+\frac{24}{23},\q \res(\zeta_A,2)=\frac{96}{17}
\end{equation}
and, by letting $\omega_k:=\log_326+\I k\mathbf{p}$ (for all $k\in\Ze$),
\begin{equation}
\res(\zeta_A,\omega_k)=\frac{24}{13\cdot 2^{\omega_k}\omega_k(\omega_k-1)(\omega_k-2)\log 3}.
\end{equation}

One also easily checks that the hypotheses of Theorem \ref{pointwise_thm_d} (or, really, Theorem \ref{5.3.15.1/2} since all of the complex dimensions in \eqref{3-carp_po} are simple) are satisfied for every $\d\geq 1/2$ and any scaling factor $\lambda\geq 2$, and thus we obtain the following exact pointwise tube formula, valid for all $t\in(0,1/{2})$:
\begin{equation}\label{tube_3_carp}
\begin{aligned}
|A_t|&=\frac{24\,t^{3-\log_326}}{13\log 3}\sum_{k=-\ty}^{+\ty}\frac{2^{-\omega_k}t^{-\I k\mathbf{p}}}{(3-\omega_k)(\omega_k-1)(\omega_k-2)\omega_k}\\
&\phantom{=}+\left(6-\frac{6}{17}\right)t+\left(3\pi+\frac{12}{23}\right)t^2+\left(\frac{4\pi}{3}-\frac{8}{25}\right)t^3.
\end{aligned}
\end{equation}
In particular, from the above formula we conclude that $D:=\dim_BA=\log_326$ (as was noted before) and 
that the three-dimensional Sierpi\'nski carpet is not Minkowski measurable, which is expected (see \cite{Lap3}).
We also point out that the part $6t+3\pi t^2+4\pi t^3/3$ from the above Equation~\eqref{tube_3_carp} is exactly equal to $|I_t|-|I|$, where $I$ is the closed unit cube of $\eR^3$.

Finally, we note that clearly, the first term in the right-hand side of \eqref{tube_3_carp} can be rewritten in the following form (still with $D:=\dim_BA$):
\begin{equation}\label{5.5.27.1/2}
(2t)^{3-D}G\left(\log_3(2t)^{-1}\right),
\end{equation}
where $G$ is a positive, nonconstant $1$-periodic function which is bounded away from zero and infinity.
Therefore, also as expected (see \cite{Lap3}), the $3$-carpet is Minkowski nondegenerate: $0<\mathcal{M}_*(A)<\mathcal{M}^*(A)<\ty$.

\medskip

Of course, exactly the same comment as above about the Minkowski nonmeasurability and the Minkowski nondegeneracy could have been made about the Sierpi\'nski gasket discussed in Example \ref{gsk_fract}.
\end{example}

We caution the reader, however, that the situation concerning the $N$-dimensional Sierpi\'nski $N$-gasket studied in \cite[Example 4.1.4]{refds} or \cite[Example 4.2.24]{fzf} is more complicated in higher dimensions.
For instance, for $N=3$, this RFD is Minkowski degenerate (specifically, $\mathcal{M}=+\infty$) but (because its distance zeta function has a double pole at $s=D=2$) it is $h$-Minkowski measurable with respect to the gauge function $h(t):=\log t^{-1}$.
(See also \cite{brezis}, in addition to [LapRa\v Zu1,4].) 
Furthermore, and somewhat surprisingly, when $N\geq 4$, the Sierpi\'nski $N$-gasket RFD is Minkowski measurable and subcritically  Minkowski nonmeasurable (while nondegenerate); see \cite[Remark 5.5.26$(c)$]{fzf}.
This follows from the fact that when $N\geq 4$, the dimension of the boundary  of the generator of the $N$-gasket RFD (viewed as a self-similar fractal spray) is strictly larger that the similarity dimension of the RFD.



\subsection{A Relative Fractal Drum Generated by the Cantor Function}\label{subsec_devil}

The example dicussed in this subsection, namely, a version of the Cantor graph (or ``devil's staircase'', in the terminology of \cite{Man}) plays an important role in showing why the notion of complex dimensions gives a lot more information than the mere (Minkowski or Hausdorff) fractal dimension, as will be explained below in relation to the elusive notion of ``fractality''.

\begin{example}\label{stair}({\em The Cantor function RFD}).
In this example, we compute the distance zeta function of the RFD $(A,\O)$ in $\eR^{2}$, where $A$ is the graph of the Cantor function and $\O$ is the union of triangles $\triangle_k$ that lie above and the triangles $\widetilde{\triangle}_k$ that lie below each of the horizontal parts of the graph denoted by $B_k$.
(At each step of the construction there are $2^{k-1}$ mutually congruent triangles $\triangle_k$ and $\widetilde{\triangle}_k$.)
Each of these triangles is isosceles, has for one of its sides a horizontal part of the Cantor function graph, and has a right angle at the left end of $B_k$, in the case of $\triangle_k$, or at the right end of $B_k$, in the case of $\widetilde{\triangle}_k$.
(See Figure~\ref{piramide}, right.)

For obvious geometric reasons and by using the scaling property of the relative distance zeta function of the resulting RFD $(A,\O)$ (see \cite[Theorem 4.1.38]{fzf} or \cite[Section 2.2]{mefzf}), we then have the following identity:
\begin{equation}\label{eq5.5.52}
\begin{aligned}
\zeta_{A,\O}(s)&=\sum_{k=1}^{\ty}2^{k}\zeta_{B_k,\triangle_k}(s)=\sum_{k=1}^{\ty}2^{k}\zeta_{3^{-k}B_1,3^{-k}\triangle_1}(s)\\
&=\zeta_{B_1,\triangle_1}(s)\sum_{k=1}^{\ty}\frac{2^{k}}{3^{ks}}=\frac{2\zeta_{B_1,\triangle_1}(s)}{3^s-2},
\end{aligned}
\end{equation}
valid for all $s\in\Ce$ with $\re s$ sufficiently large. 
Here, $(B_1,\triangle_1)$ is the relative fractal drum described above with two perpendicular sides of length equal to 1.
It is straightforward to compute its relative distance zeta function:
\begin{equation}
\zeta_{B_1,\triangle_1}(s)=\int_0^1\di x\int_0^xy^{s-2}\di y=\frac{1}{s(s-1)},
\end{equation}
valid, initially, for all $s\in\Ce$ such that $\re s>1$ and, upon meromorphic continuation, for all $s\in\Ce$.
This fact, combined with (the last equality of) Equation \eqref{eq5.5.52}, gives us the distance zeta function of $(A,\O)$, which is clearly meromorphic on all of $\Ce$:
\begin{equation}\label{zeta_devil_stair}
\zeta_{A,\O}(s)=\frac{2}{s(3^s-2)(s-1)},\quad\textrm{for all}\ s\in\Ce.
\end{equation}
We therefore deduce that the set of complex dimensions of the RFD $(A,\O)$ is given by
\begin{equation}\label{devil_dim}
\po(\zeta_{A,\O}):=\po({\zeta}_{A,\O},\Ce)=\{0,1\}\cup\left(\log_32+\frac{2\pi}{\log3}\I\Ze\right),
\end{equation}
with each complex dimension being simple.

We will see in a moment that $\dim_B(A,\O)=1$ and that the RFD $(A,\O)$ is Minkowski measurable.
Moreover, we will also see that the (one-dimensional) Minkowski content of $(A,\O)$ is given by 
\begin{equation}\label{devil_mink}
\mathcal{M}^{1}(A,\O)=\frac{\res(\zeta_{A,\O},1)}{2-1}=2,
\end{equation}
which coincides with the length of the Cantor graph (i.e., the graph of the Cantor function, also called the devil's staircase in \cite{Man}).

In the sequel, we associate the RFD $(A,A_{1/3})$ in $\eR^2$ to the classic Cantor graph.
We do not know if~\eqref{devil_dim} coincides with the set of complex dimensions of the `full' graph of the Cantor function (i.e., the original devil's staircase), or equivalently, the RFD $(A,A_{1/3})$, but we expect that this is indeed the case since $(A,\O)$ is a `relative fractal subdrum' of $(A,A_{1/3})$.
Moreover, it is obvious that for the distance zeta function of the RFD $(A,A_{1/3})$ associated with the graph of the Cantor function, one has
\begin{equation}\label{pole-pole}
{\zeta}_{A,A_{1/3}}(s)=\zeta_{A,\O}(s)+\zeta_{A,A_{1/3}\setminus\O}(s).
\end{equation}
In order to prove that $\po(\zeta_{A,\O})$, given by~\eqref{devil_dim}, is a subset of the complex dimensions of the `full' Cantor graph, it would therefore remain to show that $\zeta_{A,A_{1/3}\setminus\O}(s)$ has a meromorphic continuation to some connected open neigborhood $U$ of the critical line $\{\re s=1\}$ such that $U$ contains the set of complex dimensions of $(A,\O)$, as given by \eqref{devil_dim}, and that there are no pole-pole cancellation in the right-hand side of~\eqref{pole-pole}.
%

One easily checks that $\lambda^s\zeta_{A,\O}(s;1/3)$ is strongly $d$-languid for any $\lambda\geq 1$, with $\kappa_d:=-2$, and thus we can apply Theorem~\ref{pointwise_thm_d} in order to obtain the following exact pointwise fractal tube formula for the RFD $(A,\O)$, valid for all $t\in(0,1)$:
\begin{equation}\label{racun_devil}
\begin{aligned}
V_{A,\O}(t):=|A_t\cap\O|&=\!\!\!\sum_{\omega\in\po({\zeta}_{A,\O})}\!\!\!\res\left(\frac{t^{2-s}}{2-s}{\zeta}_{A,\O}(s),\omega\right)=\!\sum_{\omega\in\po({\zeta}_{A,\O})}\!\frac{t^{2-\omega}}{2-\omega}\res\left({\zeta}_{A,\O},\omega\right)\\
&=2t+\frac{t^{2-\log_32}}{\log 3}\sum_{k=-\ty}^{+\ty}\frac{t^{-\I k\mathbf{p}}}{(2-\omega_k)(\omega_k-1)\omega_k}+t^2\\
&=2t^{2-D_{CF}}+{t^{2-D_{CS}}}G_{CF}\left(\log_3 t^{-1}\right)+t^2,\\
\end{aligned}
\vspace{4pt}
\end{equation}
where $\omega_k:=\log_32+\I k\mathbf{p}$ (for each $k\in\Ze$), $D_{CF}=\dim_{B}(A,\O)=1$, $D_{CS}=\log_32$ and $\mathbf{p}:=2\pi/\log 3$.

In the last line of \eqref{racun_devil}, $G_{CF}$ is a nonconstant $1$-periodic function on $\eR$, which is bounded away from zero and infinity.
It is given by the following absolutely convergent (and hence, convergent) Fourier series:
\begin{equation}\label{5.5.58.1//2}
G_{CF}(x):=\frac{1}{\log 3}\sum_{k=-\ty}^{+\ty}\frac{\E^{2\pi\I kx}}{(2-\omega_k)(\omega_k-1)\omega_k},\quad\textrm{for all}\ x\in\eR. 
\end{equation}
Note that in order to obtain the third equality in \eqref{racun_devil}, and hence also the above expression for $G_{CF}$ given in \eqref{5.5.58.1//2}, we have used the fact that (in light of \eqref{zeta_devil_stair} and \eqref{devil_dim})
\begin{equation}\label{5.5.8.1/2}
\res\big(\zeta_{A,\O}(s),\omega_k\big)=\frac{1}{\log 3(\omega_k-1)\omega_k},
\end{equation}
for all $k\in\Ze$.

It is interesting that it follows from \eqref{racun_devil} and \eqref{5.5.58.1//2} that even though this version of the Cantor graph, described by the RFD $(A,\O)$, is Minkowski measurable and hence does not have any oscillations of leading order, it has {\em oscillations of lower order}, corresponding to the complex dimensions of the Cantor set (or string) of the form $D_{CS}+\I k\mathbf{p}$, with $k\in\Ze$ (see Example \ref{ecant}, especially, Equation \eqref{comp_dim_CC}); that is, it has {\em subcritical oscillations}, of order $2-D_{CS}\approx 1.3691$, where $D_{CS}:=\log_32$ is the Minkowski dimension of the Cantor set (or string).
In fact, in light of the pointwise fractal tube formula \eqref{racun_devil} and since the RFD $(A,\O)$ has Minkowski content $\mathcal{M}_{CF}:=\mathcal{M}(A,\O)=2$ (see Equation \eqref{devil_mink} above), as well as Minkowski dimension $D_{CF}:=\dim_B(A,\O)=1$, we have that
\begin{equation}\label{5.5.34.1/4}
\begin{aligned}
0&<\liminf_{t\to 0^+}\,\,t^{-(2-D_{CS})}\left|\mathcal{M}_{CF}t^{2-D_{CF}}-V_{A,\O}(t)\right|\\
&<\limsup_{t\to 0^+}\,\,t^{-(2-D_{CS})}\left|\mathcal{M}_{CF}t^{2-D_{CF}}-V_{A,\O}(t)\right|<\ty.
\end{aligned}
\end{equation}
Hence, we see that even though the {\em leading term} (as $t\to 0^+$) in the fractal tube formula \eqref{racun_devil} is of order $2-D_{CF}=1$, determined by the Minkowski dimension $D_{CF}=1$ of $(A,\O)$, as should be case, and is {\em monotonic} (and therefore, {\em nonoscillatory}), the {\em asymptotic second term}, $h(t):=t^{2-D_{CS}}G(\log_3t^{-1})$, is of order $2-D_{CS}$, determined by the Minkowski dimension $D_{CS}=\log_32$ of the {\em Cantor set} (or string), and is {\em oscillatory} (in fact, multiplicatively periodic, or ``$\log$-periodic'', to use the physicists' terminology).
\end{example}

\begin{remark}\label{5.5.9.1/2}({\em Critical vs subcritical fractals}).
The above example motivates us to propose to call a geometric object ``fractal'' if it has at least one nonreal complex dimension (or if its fractal zeta function has a natural boundary along a suitable screen, in which case it is said to be ``hyperfractal'').
(See \cite[Sections 12.1 and 12.2]{lapidusfrank12}, along with \cite[Subsection~13.4.3]{lapidusfrank12}, as adapted and extended to our general higher-dimensional theory of complex dimensions in \cite[Definition 2.38]{refds} \cite[Definition 4.6.23 and Remark 4.6.24 of Subsection 4.6.3]{fzf}.)
Accordingly, the present version of the Cantor graph (i.e., the RFD $(A,\O)$ from Example \ref{stair} just above) is ``fractal'' in this sense.

In addition, following [Lap-vFr1--3] 
(see, especially, \cite[Section~3.7]{lapidusfrank12}), given $d\in\eR$ (with $d\leq N$), we say that a geometric object is {\em fractal in dimension $d$} if it has at least one {\em nonreal} (visible) complex dimension of real part $d$.\footnote{We allow here the number $d$ to be nonpositive, since it enables us to deal with a broader class of potential fractals.}
(Automatically, it will have at least one pair of nonreal complex conjugate complex dimensions of real part $d$.)
If the object in question is an RFD $(A,\O)$ (and, in particular, a bounded set $A$) in $\eR^N$, with upper (relative) Minkowski dimension $\ov{\dim}_B(A,\O)$ (or, in particular $\ov{\dim}_BA$) denoted by $\ov{D}$, then we can distinguish between the following two different and interesting cases:\footnote{We assume here implicitly that the fractal zeta  function of $(A,\O)$ under consideration has a meromorphic extension to a connected open neighborhood of the critical line $\{\re s=\ov{\dim}_{B}(A,\O)\}$, say, to the interior of a window $\bm W$ with associated screen $\bm S$ such that $\sup S<\ov{D}:=\ov{\dim}_B(A,\O)$.
We also assume that $\ov{D}\in\eR$; i.e. (since $\ov{D}\leq N$), $\ov{D}\neq -\ty$.}

\medskip

$(i)$ ({\em Critical case}). The RFD $(A,\O)$ is fractal in dimension $d:=\ov{D}$, in which case $(A,\O)$ is said to be {\em critically fractal}.
Indeed, under suitable hypotheses, it then follows from the fractal tube formulas of Sections \ref{sec_point}--\ref{distance_tube} that it has at least one nonreal complex dimension on the critical line $\{\re s=\ov{D}\}$, thereby giving rise to {\em geometric oscillations of leading order}.

\medskip

$(ii)$ ({\em Subcritical case}). The RFD $(A,\O)$ is not fractal in dimension $\ov{D}$ (i.e., it does not have any nonreal principal complex dimension), but it is fractal in some dimension $d<\ov{D}$.
The RFD $(A,\O)$ is then said to be {\em subcritically fractal}.
(Sometimes, we will also say that $(A,\O)$ is {\em ``strictly subcritically fractal''} in order to emphasize the fact that $d<D$, and we will say that $(A,\O)$ is {\em ``possibly subcritically fractal''} in order to indicate that $d\leq D$ instead of $d<D$.)

\medskip

[Other cases are possible, such as $(A,\O)$ being hyperfractal (in the sense of [LapRa\v Zu1--4]), 
even in case $(i)$ or $(ii)$, or else $(A,\O)$ being {\em nonfractal}; that is, neither having a nonreal (visible) complex dimension nor being hyperfractal.
However, we are not concerned with these situations in the present context.]

\medskip

Given an RFD $(A,\O)$, we define $\alpha\in\eR\cup\{-\ty\}$, the {\em subcriticality index} of $(A,\O)$, via the following formula:
\begin{equation}\label{5.5.34.1/2}
\alpha=\alpha_{A,\O}:=\sup\left\{d\in\eR\,:\,(A,\O)\textrm{ is fractal in dimension }d\right\}.
\end{equation}
By convention, we let $\a_{A,\O}=-\ty$ if $(A,\O)$ is not fractal for any $d\in\eR$.
(Clearly, we always have $\a_{A,\O}\leq\ov{D}\leq N$.)

\medskip

We note that even if $(A,\O)$ is subcritically fractal, it could happen that $\alpha_{A,\O}=\ov{D}:=\ov{\dim}_B(A,\O)$.
This is the case, for instance, if $(A,\O):=(\pa\O,\O)$ is a generic, nonlattice self-similar string, in the sense of \cite[Subsection~3.2.1]{lapidusfrank12}.\footnote{Recall from \cite[Chapters~2--3]{lapidusfrank12} that a self-similar string with distinct scaling ratios $\rho_1,\ldots,\rho_n$ in $(0,1)$ is said to be {\em lattice} (resp., {\em nonlattice}) if the rank of the group generated by $\rho_1,\ldots,\rho_n$ (viewed as a multiplicative subgroup of $(0,+\ty)$) is equal to $1$ (resp., $>1$), and {\em generic nonlattice} if the rank is equal to $n$, the maximal possible rank.}
Then, as was conjectured in \cite[Subsection~3.7.1]{lapidusfrank12} (as well as, more specifically, in reference [Lap-vF6] of \cite{lapidusfrank12}) and later proved in \cite{MorSepVi}, the {\em set of dimensions of fractality} of $(A,\O)$ (i.e., the set of real numbers $d$ such that $(A,\O)$ is fractal in dimension $d$) is dense in some compact interval of the form $[D_*,\ov{D}]$, with $D_*\in\eR$ and $D_*<\ov{D}$.
As a result, in light of \eqref{5.5.34.1/2}, it follows that $\alpha_{A,\O}=\ov{D}$.
However, $(A,\O)$ is not critically fractal (because according to \cite[Theorem~2.16]{lapidusfrank12}, a (generic) nonlattice string does not have any nonreal complex dimensions of real part $\ov{D}$), even though it is subcritically fractal in dimension $d<\ov{D}$ for a dense (and countable) set of real numbers $d$ in $[D_*,\ov{D}]$.
\end{remark}

We now return to the RFD considered in Example \ref{stair} (that is, the version of the Cantor graph denoted by $(A,\O)$), and we refer to Remark \ref{5.5.9.1/2} just above for the appropriate terminology and definitions.
As we have seen, $(A,\O)$ is fractal.
More specifically, {\em it is not critically fractal} (because its only complex dimension of real part $D_{CF}\ (=\ov{D}=\dim_{B}(A,\O))=1$ is $1$ itself, the Minkowski dimension of the Cantor graph, and it is simple) {\em but it is $($strictly$)$ subcritically fractal}.
In fact, it is subcritically fractal in a single dimension, namely, in dimension $d=D_{CS}=\log_32$, the Minkowski dimension of the Cantor set.
Consequently, in light of \eqref{5.5.34.1/2}, the subcriticality index of $(A,\O)$ is given by
$
\alpha_{A,\O}=D_{CS}=\log_32,
$
and it is attained.

We expect the same result to hold for the devil's staircase itself (i.e., the `full' graph of the Cantor function), represented by the RFD $(A,A_{1/3})$ and of which $(A,\O)$ is a `relative fractal subdrum', as above.
Clearly, in light of \eqref{pole-pole} and \eqref{devil_dim}, we have the following inclusions (between multisets):
\begin{equation}\label{5.5.34.4/5}
\begin{aligned}
\po(\zeta_{A,A_{1/3}})&\subseteq\po(\zeta_{A,\Omega})\cup\po(\zeta_{A,A_{1/3}\setminus\O})\subseteq\{0,1\}\cup\left\{D_{CS}+\frac{2\pi}{\log 3}\I\Ze\right\}.
\end{aligned}
\end{equation}

Also, we know for a fact that $\dim_B(A,A_{1/3})$ exists and 
\begin{equation}\label{DA1/3}
D(\zeta_{A,A_{1/3}})=\dim_B(A,A_{1/3})=1,
\end{equation}
so that
\begin{equation}\label{DA1/3PC}
\dim_{PC}(A,A_{1/3}):=\po_c(\zeta_{A,A_{1/3}})=\{1\}.
\end{equation}
(Thus, we also have that $\{1\}\stq \po(\zeta_{A,A_{1/3}})$ in \eqref{5.5.34.4/5}.) Note that \eqref{DA1/3} (and hence, \eqref{DA1/3PC}) follows from the rectifiability of the devil's staircase, combined with a well-known result in \cite{federer} and with part $(b)$ of Theorem \ref{an_rel}.

As was mentioned earlier in the discussion of Example \ref{stair} (and was predicted in \cite[Subsections 12.1.2 and 12.3.2]{lapidusfrank12}, based on an `approximate tube formula'), we expect that $\po(\zeta_{A,A_{1/3}})=\po(\zeta_{A,\O})$, as given by \eqref{devil_dim}, and hence, that we actually have equalities instead of inclusions in \eqref{5.5.34.4/5}, even equalities between multisets.
If so, then the `full' Cantor graph $(A,A_{1/3})$ is fractal, not critically fractal, but (strictly) subcritically fractal in the single dimension $d:=D_{CS}=\log_32$.

Clearly, both $(A,\O)$ and $(A,A_{1/3})$ should be fractal for a proper definition of fractality.
This would completely resolve the following apparent paradox: the RFD $(A,A_{1/3})$ is not ``fractal'' according to Mandelbrot's original definition of fractality given in \cite{Man},\footnote{Indeed, Mandelbrot's definition, given in \cite[p.\ 15]{Man}, can be stated as follows.
A geometric object is ``fractal'' if its Hausdorff dimension is strictly greater than (i.e., is not equal to) its topological dimension.
However, note that the Hausdorff, Minkowski and topological dimensions coincide and are equal to $1$ in the case of  (either the `full' or the `partial') Cantor graph.
If, in addition, we replaced ``Hausdorf dimension'' by (relative, upper) ``Minkowski dimension'' in the above definition and we interpreted the topological dimension in the obvious way, we would also reach the analogous conclusion for both $(A,A_{1/3})$ and $(A,\O)$, which therefore would still not be fractal according to this modified Mandelbrot definition.} even though everyone feels and expects it to be ``fractal'' simply after having glanced at the `full' Cantor graph $(A,A_{1/3})$ (the `devil's staircase' in the sense of \cite{Man}).
The same is true for the `partial' Cantor graph $(A,\O)$, for which we can now rigorously prove that it is ``fractal'' (in the sense of the theory of complex dimensions) even though it is only (strictly) subcritically fractal, which may explain, in hindsight, why some practitioners refer to it as a ``borderline fractal''.

We conclude this discussion by quoting (as in \cite[p.\ 335]{lapidusfrank12}) Mandelbrot \cite[p.\ 82]{Man} writing about the devil's staircase (the `full' Cantor graph, depicted in \cite[Plate 83, p.\ 83]{Man}):

\bigskip

\begin{center}
\begin{minipage}[c]{0.8\textwidth}
{\em One would love to call the present curve a fractal, but to achieve this goal we would have to define fractals less stringently, on the basis of notions other than $D$} [the Hausdorff dimension] {\em alone.}
\end{minipage}
\end{center}

\bigskip

Thanks to the higher-dimensional theory of complex dimensions of fractals and associated fractal tube formulas developed in this paper and in [LapRa\v Zu1--8], 
building on the corresponding theory for fractal strings developed in [Lap-vFr1--3], 
we are now tentatively close to having resolved this apparent paradox.
Furthermore, if we use the `partial' Cantor graph $(A,\O)$ as a suitable substitute for the `full' Cantor graph, viewed as the RFD $(A,A_{1/3})$, the corresponding paradox is indeed completely resolved here.
We invite the interested reader to extend the conclusions of the present example (i.e., Example \ref{stair}) from $(A,\O)$ to $(A,A_{1/3})$, and thereby, to fully prove the conjectures and statements made in \cite[Subsection~12.1.2]{lapidusfrank12} as well as here about the devil's staircase itself.

\subsection{Fractal Nests and Unbounded Geometric Chirps}\label{subsec_nestch}

In this subsection, we apply our general fractal tube formulas to several families of fractal nests (Example \ref{ex_nest}) and of (unbounded) geometric chirps (Example \ref{unb_chirp}).
Both of these families are examples of ``fractal'' sets which are {\em not} self-similar or, more generally, `self-alike' in any sense.

\begin{example}({\em Fractal nests}).\label{ex_nest}
We let $\mathcal{L}=(\ell_j)_{j\geq 1}$ be a bounded fractal string and, as before, let $A_{\mathcal{L}}=\{a_k:k\in\eN\}\subset\eR$, with $a_k:=\sum_{j\geq k}\ell_j$ for each $k\geq 1$.
Furthermore, consider now $A_{\mathcal{L}}$ as a subset of the $x_1$-axis in $\eR^2$ and let $A$ be the planar set obtained by rotating $A_{\mathcal{L}}$ around the origin; i.e., $A$ is a union of concentric circles of radii $a_k$ and center at the origin which we call the {\em fractal nest of center type} generated by the fracctal string $\mathcal{L}$; see Figure \ref{nest_center}.
For $\d>\ell_1/2$, the distance zeta function of $A$ is given by
\begin{equation}
\zeta_{A}(s)=\frac{2^{2-s}\pi}{s-1}\sum_{j=1}^{\ty}\ell_j^{s-1}(a_j+a_{j+1})+\frac{2\pi\d^{s}}{s}+\frac{2\pi a_1\d^{s-1}}{s-1};
\end{equation}
see \cite[Example 3.5.1]{fzf}.
The last two terms in the above formula correspond to the annulus $a_1<r<a_1+\d$ and we will neglect them; that is, we will consider only the relative distance zeta function $\zeta_{A,\O}$, with $\O:=B_{a_1}(0)$.\footnote{Here, $B_r(x)$ denotes the open ball of radius $r$ with center at $x$.}
Furthermore, since $a_{j+1}=a_j-\ell_j$ for each $j\geq 1$, we have
\begin{equation}
\begin{aligned}
\zeta_{A,\O}(s)&=\frac{2^{2-s}\pi}{s-1}\sum_{k=1}^{\ty}\ell_j^{s-1}(2a_j-\ell_j)
=\frac{2^{3-s}\pi}{s-1}\zeta_1(s)-\frac{2^{2-s}\pi}{s-1}\zeta_{\mathcal{L}}(s),
\end{aligned}
\end{equation}
where we have denoted by $\zeta_1$ the first of the two sums appearing after the first equality and where $\zeta_{\mathcal{L}}$ is the geometric zeta function of the fractal string $\mathcal{L}$.

\medskip

\begin{figure}[t]
\begin{center}
\includegraphics[trim=0cm 2.7cm 0cm 0cm,clip=true,width=8cm]{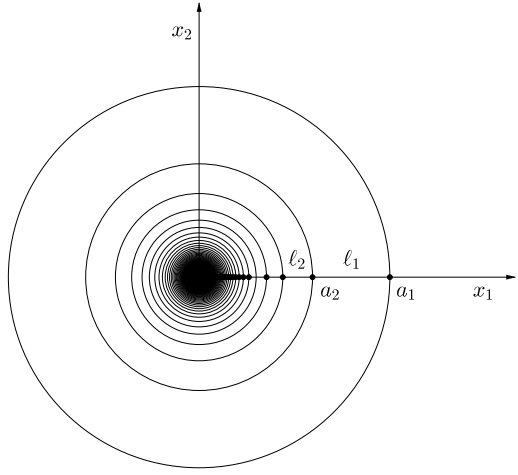}
\end{center}
\vskip-.5cm
\caption{The fractal nest of center type in the plane $\eR^2$ generated by the fractal string $\mathcal{L}=(\ell_j)_{j\ge1}$.
Note that for every $k\ge1$ we have $\ell_k=a_k-a_{k+1}$, where $a_k:=\sum_{j\ge k}\ell_j$. Furthermore, we have $A_{\mathcal{L}}:=\{a_k:k\ge1\}$.}\label{nest_center}
\end{figure}

Let us next consider a special case of the fractal nest above; that is, the relative fractal drum $(A_a,\O)$ corresponding to the $a$-string $\mathcal{L}:=\mathcal{L}_a$, with $a>0$; so that $\ell_j:=j^{-a}-(j+1)^{-a}$ for all $j\geq 1$ and hence, $a_j=j^{-a}$ for every $j\geq 1$.
In this case, we have that
\begin{equation}\label{zeta_nest}
\zeta_{A_a,\O}(s)=\frac{2^{3-s}\pi}{s-1}\sum_{j=1}^{\ty}j^{-a}\ell_j^{s-1}-\frac{2^{2-s}\pi}{s-1}\zeta_{\mathcal{L}}(s).
\end{equation}
Since the geometric zeta function $\zeta_{\mathcal{L}}=\zeta_{\mathcal{L}_a}$ has already been analyzed in Example \ref{ex_a} (based on the results of \cite[Subsection~6.5.1]{lapidusfrank12}), we will now do the same for the zeta function $\zeta_1$ by means of a technique analogous to the one used in the proof of \cite[Theorem 6.21]{lapidusfrank12}.
Here, $\zeta_1(s)$ is initially defined by the following Dirichlet series (still with $\ell_j:=j^{-a}-(j+1)^{-a}$, for all $j\geq 1$):
\begin{equation}\label{5.5.65.1/2..}
\zeta_1(s)=\sum_{j=1}^{\ty}j^{-a}\ell_j^{s-1},
\end{equation}
for all $s\in\Ce$ with $\re s$ sufficiently large.
Hence, we have $\zeta_1(s)=\zeta_{\mathcal{L},-a}(s-1)$, in the notation of the next theorem, and $\zeta_1(s)=\zeta_{\mathcal{L},-a,1}(s)$, in the notation of Corollary \ref{kor_zet} following it.
\end{example}

\begin{theorem}\label{anest}
Let $a>0$, $b\in\eR$, and let $\mathcal{L}=\mathcal{L}_a$ be the $a$-string with lengths $\ell_j$ given by \eqref{lj}; i.e., $\ell_j=j^{-a}-(j+1)^{-a}$ for all $j\geq 1$.
Then, the Dirichlet series $\zeta_{\mathcal{L},b}(s):=\sum_{j=1}^{\ty}j^{b}\ell_j^{s}$ $($defined initially for all $s\in\Ce$ with $\re s$ sufficiently large$)$  has a meromorphic continuation to all of $\Ce$. The poles of $\zeta_{\mathcal{L},b}$ are located at 
\begin{equation}
C:=D(\zeta_{\mathcal{L},b})=\frac{b+1}{a+1}
\end{equation}
and in $($a subset of$)$ $\big\{\frac{b-m}{a+1}:m\in\eN_0\big\}\setminus\{0\}$, and they are all simple.\footnote{Here, as usual, we let $\eN_0:=\eN\cup\{0\}$.}
In particular, we have the following inclusions$:$\footnote{For `generic' values of $a$ and $b$, the second inclusion in \eqref{5.5.67.1/2} should be an equality while for `most' values of those parameters, $\po(\zeta_{\mathcal{L},b})$ should at least contain an infinite subset of $\{\frac{b-m}{a+1}:m\in\eN_0\}$. However, this informal comment will not be needed in the sequel and the underlying conjecture has not been proved.\label{foot_zeta_b}}
\begin{equation}\label{5.5.67.1/2}
\begin{aligned}
\left\{\frac{b+1}{a+1}\right\}&\subseteq\po(\zeta_{\mathcal{L},b}):=\po(\zeta_{\mathcal{L},b},\Ce)\subseteq\left\{\frac{b+1}{a+1}\right\}\cup\left(\left\{\frac{b-m}{a+1}:m\in\eN_0\right\}\setminus\{0\}\right).
\end{aligned}
\end{equation}
Furthermore, the residue of $\zeta_{\mathcal{L},b}$ at $C=\frac{b+1}{a+1}$ is equal to $\frac{a^C}{a+1}$; so that $\frac{b+1}{a+1}$ is always a $($necessarily simple$)$ pole of $\zeta_{\mathcal{L},b}$.

Moreover, for any screen $\bm{S}_\sigma$  chosen to be a vertical line $\{\re s=\sigma\}$, with $\sigma\in\eR\setminus\po(\zeta_{\mathcal{L},b})$, the zeta function $\zeta_{\mathcal{L},b}$ satisfies the languidity conditions {\bf L1} and {\bf L2}, with $\kappa=\frac{1}{2}+b-(a+1)\sigma$ if $\sigma\leq\frac{b}{a+1}$ and $\kappa=\frac{1}{2}(1+b-(a+1)\sigma)$ if $\sigma\in\big[\frac{b}{a+1},\frac{b+1}{a+1}\big]$.

Finally, we have that $\zeta_{\mathcal{L},b}(0)=\zeta(-b)$ for all $b\in\eR\setminus\{-1\}$, where $\zeta$ is the Riemann zeta function.
\end{theorem}

\begin{proof}
We begin by computing the first term of an asymptotic expansion of $\ell_j$:
\begin{equation}
\ell_j=j^{-a}-(j+1)^{-a}=a\int_j^{j+1}x^{-a-1}\di x=aj^{-a-1}+H(j),
\end{equation}
where $j\geq 1$ and $H(j):=a\int_j^{j+1}(x^{-a-1}-j^{-a-1})\di x$.
We next introduce a new variable $t:=x/j-1$ and let
\begin{equation}\label{hj}
h_j:=a^{-1}j^{a+1}H(j)=j\int_0^{1/j}\big((1+t)^{-a-1}-1\big)\di t.
\end{equation}
Note that $h_j=O(1/j)$ as $j\to\ty$.
By now choosing an integer $M\geq 0$, we have
\begin{equation}
\begin{aligned}
j^b\ell_j^s&=j^b\big(aj^{-a-1}(1+h_j)\big)^s\\
&=a^sj^{b-s(a+1)}\left(\sum_{n=0}^{M}\binom{s}{n}h_j^n+O\left(\frac{(|s|+1)^{M+1}}{j^{M+1}}\right)\right)\ \textrm{as}\ j\to\ty,
\end{aligned}
\end{equation}
where we have let
\begin{equation}
\binom{s}{n}:=\frac{(s-n+1)_n}{n!},\quad\textrm{for all}\ s\in\Ce\textrm{ and }n\in\eN_0.
\end{equation}
(Clearly, $\binom{s}{n}$ is a natural generalization of the usual binomial coefficient to an arbitrary value of the parameter $s\in\Ce$.)
We thus obtain the following identity:
\begin{equation}\label{zhj}
\zeta_{\mathcal{L},b}(s)=\sum_{n=0}^{M}a^s\binom{s}{n}\sum_{j=0}^{\ty}h_j^nj^{b-s(a+1)}+f(s),
\end{equation}
where $f(s)$ is defined and holomorphic on the open half-plane $\{\re s>\frac{b-M}{a+1}\}$.
Furthermore, the first term (i.e., the term corresponding to $n=0$ in the above sum) is equal to $a^s\zeta((a+1)s-b)$, where $\zeta$ is the Riemann zeta function, and thus has a single, simple pole at $s=C:=D(\zeta_{\mathcal{L},b})=\frac{b+1}{a+1}$.\footnote{See, e.g., \cite{Titch2} or \cite{Edw} for the relevant properties of the Riemann zeta function.
Recall, in particular, that $\zeta$ has a meromorphic continuation to all of $\Ce$ with a single, simple pole at $s=1$ (with residue $1$) and that it is initially defined by the Dirichlet series $\zeta(s)=\sum_{j=1}^{\ty}j^{-s}$ for all $s\in\Ce$ with $\re s>1$.}
In order to compute the residue of $a^s\zeta((a+1)s-b)$ at $s=\frac{b+1}{a+1}$, we use the fact that the principal part of the Riemann zeta function at $s=1$ is equal to $1/(s-1)$ and consequently,
\begin{equation}
\lim_{s\to C}(s-C)a^s\zeta((a+1)s-b)=\lim_{s\to C}a^s\frac{s-C}{(a+1)s-b-1}=\frac{a^{\frac{b+1}{a+1}}}{a+1}.
\end{equation}
A well-known result about the growth of the Riemann zeta function along vertical lines (see, e.g., \cite[Section~9.2]{Edw}) implies that the first term in \eqref{zhj} grows along the vertical lines $\{\re s=\sigma\}$, with $\sigma\in\eR$, as $(|t|+1)^{\frac{1}{2}+b-\sigma(a+1)}$ if $\sigma<\frac{b}{a+1}$, as $(|t|+1)^{\frac{1}{2}(b+1-(a+1)\sigma)}$ for $\sigma\in\big[\frac{b}{a+1},\frac{b+1}{a+1}\big]$, and is bounded from above by a constant (possibly depending on $\sigma$) if $\sigma>\frac{b+1}{a+1}$.

It now remains to analyze the functions
\begin{equation}\label{hjn}
\sum_{j=1}^{\ty}h_j^nj^{b-(a+1)s},
\end{equation}
for each $n\geq 1$.

Let us fix $M\in\eN_0$, for now.
Then, the asymptotic expansion $(1+t)^{-a-1}=\sum_{m=0}^{M}\binom{-a-1}{m}t^m+O(t^{M+1})$ as $t\to0^+$, together with \eqref{hj}, yields
\begin{equation}
\begin{aligned}
h_j&=j\int_0^{1/j}\sum_{m=1}^{M}\binom{-a-1}{m}t^m\di t+O(j^{-M-1})\\
&=-\frac{1}{a}\sum_{m=1}^{M}\binom{-a}{m+1}j^{-m}+O(j^{-M-1})\quad\textrm{as}\ j\to+\ty.
\end{aligned}
\end{equation}
We proceed by taking the $n$-th power of the above expansion to obtain an asymptotic expansion for $h_j^n$ and substitute this into \eqref{hjn}.
This enables us to express each of the functions in \eqref{hjn} as a sum of constant multiples of $\zeta(m+(a+1)s-b)$, for $n\leq m\leq M$, and of a remainder term of order $O(j^{-M-1})$.
Since $\zeta(m+(a+1)s-b)$ has a simple pole at $s=\frac{b+1-m}{a+1}$ and in view of \eqref{zhj}, we conclude that $\zeta_{\mathcal{L},b}(s)$ has a meromorphic continuation to the open right half-plane $\{\re s>\frac{b+1-M}{1+a}\}$, with simple poles at $s=\frac{b+1-m}{1+a}$ for $m=0,1,2,\ldots,M$.
To be more specific, some of these potential poles of $\zeta_{\mathcal{L},b}$ may not actually be poles (due to cancellations), depending on the choice of the parameters $a$ and $b$.
(See, however, the unproven assertion in footnote \ref{foot_zeta_b}.)
Furthermore, $0$ is never a pole of $\zeta_{{\mathcal{L},b}}$, since by looking at \eqref{zhj} we can see that it is canceled by the factor $\binom{s}{m}$ for $m\geq 1$.  
Moreover, since $M$ is arbitrary, we conclude that $\zeta_{\mathcal{L},b}$ has a meromorphic continuation to all of $\Ce$.
Next, note that for each integer $m\geq 1$, the growth of $\zeta(m+(a+1)s-b)$ is dominated by the growth of the first term $a^s\zeta((a+1)s-b)$ and therefore, we have proved the statement about the languidity of $\zeta_{\mathcal{L},b}$.

Finally, the last statement of the theorem follows from an application of the principle of analytic continuation since we deduce directly from the definition of $\zeta_{\mathcal{L},b}$ that $\zeta_{\mathcal{L},b}(0)=\zeta(-b)$ for all $b\in\{\re s<-1\}$.
\end{proof}

In order to complete the present discussion of the example of the fractal nests, as well as in preparation for the example of the unbounded geometric chirps (Example \ref{unb_chirp} below), we will need the following simple consequence of the above theorem.

\begin{corollary}\label{kor_zet}
Let $a>0$, $b\in\eR$, $\tau\in\eR$ and let $\mathcal{L}:=\mathcal{L}_a$ be the $a$-string with lengths $\ell_j$ given by \eqref{lj}.
Then, the Dirichlet series $\zeta_{\mathcal{L},b,\tau}(s):=\sum_{j=1}^{\ty}j^{b}\ell_j^{s-\tau}$ $($initially defined for all $s\in\Ce$ with $\re s$ sufficiently large$)$ has a meromorphic continuation to all of $\Ce$. The poles of $\zeta_{\mathcal{L},b,\tau}$ are located at 
\begin{equation}
D(\zeta_{\mathcal{L},b,\tau})=\frac{b+1}{a+1}+\tau
\end{equation}
and in $($a subset of$)$ $\big\{\frac{b-m}{a+1}+\tau:m\in\eN_0\big\}\setminus\{\tau\}$, and they are all simple.
In particular, we have the following inclusions$:$\footnote{A comment entirely analogous to the one made in footnote \ref{foot_zeta_b} on page \pageref{foot_zeta_b} holds relative to `generic' (or else `most') values of the parameters $a$, $b$ and $\tau$.
(Recall that $\mathcal{L}=\mathcal{L}_a$, so that $\zeta_{\mathcal{L},b,\tau}$ depends on $a$, $b$ and $\tau$.)}
\begin{equation}\label{5.5.67.1/22}
\begin{aligned}
\left\{\frac{b+1}{a+1}+\tau\right\}&\subseteq\po(\zeta_{\mathcal{L},b,\tau}):=\po(\zeta_{\mathcal{L},b,\tau},\Ce)\\
&\subseteq\left\{\frac{b+1}{a+1}+\tau\right\}\cup\left(\left\{\frac{b-m}{a+1}+\tau:m\in\eN_0\right\}\setminus\{\tau\}\right).
\end{aligned}
\end{equation}
Furthermore, the residue of $\zeta_{\mathcal{L},b,\tau}$ at $\frac{b+1}{a+1}+\tau$ is equal to $\frac{a^{({b+1})/({a+1})}}{a+1}$; so that $D(\zeta_{\mathcal{L},b,\tau})=\frac{b+1}{a+1}+\tau$ is always a $($necessarily simple$)$ pole of $\zeta_{\mathcal{L},b,\tau}$.

Moreover, for any screen $\bm{S}_\sigma$ chosen to be a vertical line $\{\re s=\sigma\}$, with $\sigma\in\eR\setminus\po(\zeta_{\mathcal{L},b,\tau})$, the zeta function $\zeta_{\mathcal{L},b,\tau}$ satisfies the languidity conditions {\bf L1} and {\bf L2}, with $\kappa=\frac{1}{2}+b-(a+1)\sigma$ if $\sigma\leq\frac{b}{a+1}+\tau$ and $\kappa=\frac{1}{2}(1+b-(a+1)\sigma)$ if $\sigma\in\big[\frac{b}{a+1}+\tau,\frac{b+1}{a+1}+\tau\big]$.

Finally, we have that $\zeta_{\mathcal{L},b,\tau}(\tau)=\zeta(-b)$ for all $b\in\eR\setminus\{-1\}$.
\end{corollary}

\begin{proof}
Since $\zeta_{\mathcal{L},b,\tau}(s)=\zeta_{\mathcal{L},b}(s-\tau)$, this an immediate consequence of Theorem \ref{anest}.
\end{proof}

Let us now return to Example \ref{ex_nest}, where the distance zeta function of $({A_a},\O)$ is given by \eqref{zeta_nest}; see also \eqref{5.5.65.1/2..} and the brief discussion following it.
We therefore deduce from Corollary \ref{kor_zet} and the discussion of $\zeta_{\mathcal{L}}=\zeta_{\mathcal{L}_a}$ in Example \ref{ex_nest}, combined with an application of the principle of analytic continuation, that $\zeta_{A_a,\O}$ is meromorphic on all of $\Ce$ and is given for all $s\in\Ce$ by
\begin{equation}\label{za1}
\zeta_{A_a,\O}(s)=\frac{2^{3-s}\pi}{s-1}\zeta_{\mathcal{L},-a,1}(s)-\frac{2^{2-s}\pi}{s-1}\zeta_{\mathcal{L}}(s).
\end{equation}
Moreover, 
the set of complex dimensions of $({A_a},\O)$ satisfies the inclusion
\begin{equation}
\begin{aligned}
\po({\zeta}_{A_a,\O})&:=\po({\zeta}_{A_a,\O},\Ce)\subseteq\left\{1,\frac{2}{a+1},\frac{1}{a+1}\right\}\cup\left\{-\frac{m}{a+1}:m\in\eN\right\}.
\end{aligned}
\end{equation}
Provided $a\neq1$, all of the above (potential) complex dimension are simple and if $a=1$ the complex dimension $\omega=1$ has multiplicity 2.
Furthermore, we are certain that $\frac{2}{a+1}$ is always a complex dimension of $({A_a},\O)$ since it is never canceled, as a pole.
Namely, by letting $D:=\frac{2}{a+1}$, we have for all positive $a\neq 1$ that
\begin{equation}
\res\left({\zeta}_{A_a,\O},D\right)=\frac{2^{2-D}D\pi}{D-1}a^{D-1}.
\end{equation}
We will see shortly that it will then follow from the fractal tube formula for $(A_a,\O)$ that 
if $a\in(0,1)$, $\dim_B({A_a},\O)=D(\zeta_{A_a,\O})=D$ and $(A_a,\O)$ is Minkowski measurable with Minkowski content given by
\begin{equation}\label{mink_a<1}
\mathcal{M}^D({A_a},\O)=\frac{2^{2-D}D\pi}{(2-D)(D-1)}a^{D-1}.
\end{equation}

Furthermore, 
it will also follow that if $a>1$, we have that $\dim_B({A_a},\O)=1$ and the corresponding residue is given by
\begin{equation}
\res({\zeta}_{A_a,\O},1)=4\pi\zeta_{\mathcal{L},-a,1}(1)-2\pi\zeta_{\mathcal{L}}(1)=4\pi\zeta(a)-2\pi
\end{equation}
Therefore, still for $a>1$, the RFD $(A_a,\O)$ is Minkowski measurable with Minkowski content given by
\begin{equation}\label{mink_a>1}
\mathcal{M}^1({A_a},\O)=4\pi\zeta(a)-2\pi;
\end{equation}
note that $\mathcal{M}^1({A_a},\O)$ is positive since $\zeta(a)>1$ for $a>1$; so that $2\pi<\mathcal{M}^1({A_a},\O)<\ty.$

In the critical case when $a=1$, we have that $s=1$ is a pole of second order (i.e., of multiplicity two) of ${\zeta}_{A_1,\O}(s)$ and since it is a simple pole of $\zeta_{\mathcal{L},-1,1}$,  we deduce from \eqref{za1} that 
\begin{equation}\label{res-1}
\res({\zeta}_{A_1,\O},1)=4\pi\zeta_{\mathcal{L},-1,1}[1]_0-2\pi,
\end{equation}
where for each $m\in\Ze$, $\zeta_{\mathcal{L},-1,1}[\omega]_m$\label{mthcoeff} indicates the $m$-th coefficient in the Laurent series expansion of $\zeta_{\mathcal{L},-1,1}$ around $s=\omega$.
We conclude that in this case (i.e., when $a=1$), by Theorem \ref{pole1} (and part $(b)$ of Theorem \ref{an_rel}), the RFD $({A_1},\O)$ must be Minkowski degenerate with $\dim_B({A_1},\O)=D(\zeta_{A_a,\O})=1$.\footnote{Actually, it can also be shown directly that $\mathcal{M}^1({A_1},\O)$ exists in this case and is equal to $+\ty$.}
We can also compute the coefficient corresponding to $(s-1)^{-2}$ in the Laurent expansion of ${\zeta}_{A_1,\O}$ around $s=1$, by using Corollary \ref{kor_zet}:
\begin{equation}\label{res-2}
{\zeta}_{A_1,\O}[1]_{-2}=4\pi\res(\zeta_{\mathcal{L},-1,1},1)=2\pi.
\end{equation}

Assume now that $a\neq 1$.
For $M\in\eN\cup\{0\}$, as before, we choose the screen $\bm{S}_M$ to be some vertical line between $-\frac{M+1}{1+a}$ and $-\frac{M+2}{1+a}$, and let $\bm{W}_M$ be the corresponding window.
After having applied Theorem~\ref{dist_tube_formula_d}, we then obtain the following asymptotic distributional formula for the tube function $V(t):=|(A_{a})_t\cap\O|$, as $t\to0^+$:
\begin{equation}\label{5.5.53.1/2}
\begin{aligned}
V(t)&=\frac{2^{2-D}D\pi}{(2-D)(D-1)}a^{D-1}t^{2-D}+\big(4\pi\zeta(a)-2\pi\big)t+\frac{\res\left(\zeta_{A_a,\O},\frac{1}{a+1}\right)t^{2-\frac{1}{a+1}}}{2-\frac{1}{a+1}}\\ 
&\phantom{=}+\sum_{m=1}^M\frac{\res\left(\zeta_{A_a,\O},-\frac{m}{a+1}\right)t^{2+\frac{m}{a+1}}}{2+\frac{m}{a+1}}+O\big(t^{2+\frac{M+1}{a+1}}\big)\quad\textrm{as}\quad t\to0^+,
\end{aligned}
\vspace{5pt}
\end{equation}
where the sum is interpreted as being equal to $0$ if $M=1$. 
By choosing as a screen a vertical line $\{\re s=\sigma\}$, with $\sigma>-\frac{1}{2(a+1)}$, we obtain a pointwise fractal tube formula with a pointwise error term of order $O(t^{2-\sigma})$; indeed, in light of Corollary \ref{kor_zet}, we have that $\kappa_d<0$ and hence, we can apply part $(i)$ of Theorem \ref{pointwise_thm_d}.
This pointwise formula is still given by \eqref{5.5.53.1/2} but now interpreted pointwise and valid for all $t>0$.
It is actually initially valid for all $t\in(0,\d)$ but since $\d>\ell_1/2$ may be taken arbitrary large, we conclude that it is valid for all $t>0$.
Of course, we actually do not know much about the above error term when $t$ is not close to zero, which is unimportant since we are not interested in the value of $V(t)=|(A_{a})_t\cap\O|$ for large $t$.
(Note also that clearly, $|(A_{a})_t\cap\O|=|\O|=|B_1(0)|=\pi$ for $t$ sufficiently large.)
From the obtained pointwise formula, it now follows (as was claimed) that for $a\neq 1$, $(A_a,\O)$ is Minkowski measurable with $\dim_B(A_a,\O)=\max\{1,D\}$ and with Minkowski content given by \eqref{mink_a<1} if $a<1$ or \eqref{mink_a>1} if $a>1$.   

Let us next consider the critical case when $a=1$. Choose a screen given by the vertical line $\{\re s=\sigma\}$, with $\sigma\in(-3/4,-1/2)$; we then obtain the following pointwise fractal tube formula with error term:
\begin{equation}
\begin{aligned}
V(t)&=\res\left(\frac{t^{2-s}}{2-s}\zeta_{A_1,\O}(s),1\right)+\frac{2}{3}{\res\left(\zeta_{A_1,\O},\frac{1}{2}\right)t^{\frac{3}{2}}}\\
&\phantom{=}+\frac{2}{5}{\res\left(\zeta_{A_1,\O},-\frac{1}{2}\right)t^{\frac{5}{2}}}+O(t^{2-\sigma})\ \textrm{as}\ t\to0^+.\\
\end{aligned}
\end{equation}
We expand the function $t^{2-s}/(2-s)$ into a Taylor series around $s=1$, as folows:
\begin{equation}\label{taylor}
\frac{t^{2-s}}{2-s}=t\sum_{n=0}^{\ty}(s-1)^n\sum_{k=0}^{n}\frac{(-1)^{n-k}(\log t^{-1})^k}{k!(n-k)!}.
\end{equation}
We then deduce from \eqref{res-1} and \eqref{res-2} that
\begin{equation}
\res\left(\frac{t^{2-s}}{2-s}\zeta_{A_1,\O}(s),1\right)=2\pi t\log t^{-1}+4\pi t(\zeta_{\mathcal{L},-1,1}[1]_0-1);
\end{equation}
so that (still pointwise)
\begin{equation}
\begin{aligned}
V(t)&=2\pi t\log t^{-1}+4\pi t(\zeta_{\mathcal{L},-1,1}[1]_0-1)+o(t)\quad \textrm{as}\ t\to0^+.
\end{aligned}
\end{equation}
The above tube formula is in agreement with the fact that $(A_1,\O)$ is Minkowski degenerate but it is also clear that one  can choose the function $h(t):=\log t^{-1}$, for all $t\in(0,1)$, as an appropriate gauge function (see \cite[Section 2.5]{refds},\cite{ftf_b} or \cite[Subsection 6.1.1]{fzf} for an introduction to gauge functions; see also \cite{lapidushe}).
More precisely, one then has that $\mathcal{M}^1(A_1,\O,h)$, the {\em gauge relative Minkowski content} of $(A_1,\O)$, is well defined and
\begin{equation}
\mathcal{M}^1(A_1,\O,h)=\lim_{t\to0^+}\frac{|(A_1)_t\cap\O|}{t\,h(t)}=2\pi.
\end{equation}
In particular, the RFD $(A,\O)$ is $h$-Minkowski measurable.

\begin{example}({\em Unbounded geometric chirps}).\label{unb_chirp}
In this example, we consider and study a type of unbounded geometric chirp depicted in Figure \ref{pres_4}.
A standard {\em geometric $(\a,\b)$-chirp}, with positive parameters $\a$ and $\b$, is a simple geometric approximation of the graph of the function $f(x)=x^{\a}\sin(\pi x^{-\b})$, for all $x\in(0,1)$.
(See \cite[Example 4.4.1 and Proposition 4.4.2]{fzf}.)

\begin{figure}[t]
\begin{center}
\includegraphics[width=12cm]{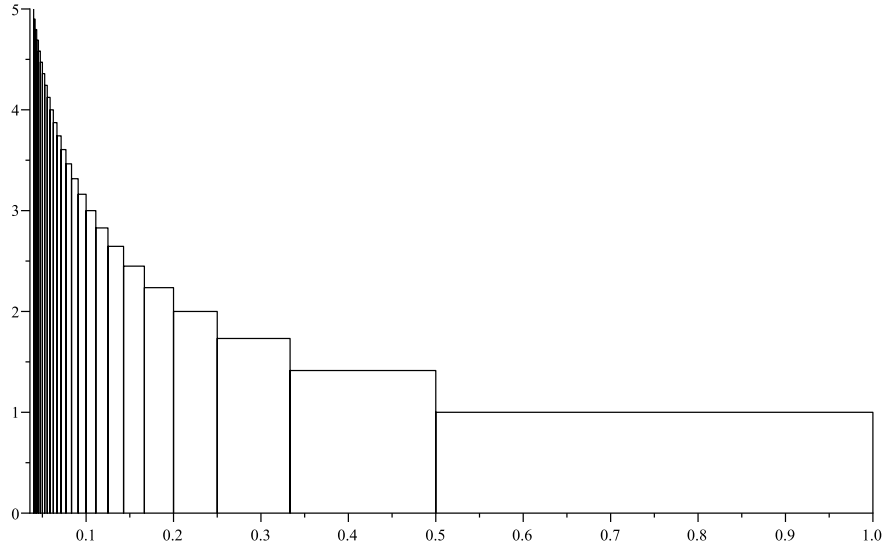}
\caption{The unbounded $(\a,\b)$-geometric chirp from Example \ref{unb_chirp}, which approximates the function $x\mapsto x^{\a}\sin(\pi x^{-\b})$ for $x\in(0,1)$. Here, $\a=-1/2$ and $\b=1$. In the corresponding RFD $(A,\O)$, the set $A$ is defined as the union of the vertical segments while the set $\O$ is defined as the union of the (open) rectangles.}
\label{pres_4}
\end{center}
\end{figure}

By choosing parameters $-1<\a<0<\beta$, we obtain an example of an unbounded chirp function $f$ which we approximate by the unbounded geometric $(\a,\b)$-chirp.
More specifically, let $A_{\a,\b}$ be the union of vertical segments with abscissae $x=j^{-1/\b}$ and of lengths $j^{-\a/\b}$, for every $j\in\eN$.
Furthermore, define $\O$ as a union of the open rectangles $R_j$ for $j\in\eN$, where $R_j$ has a base of length $j^{-1/\b}-(j+1)^{-1/\b}$ and height $j^{-\a/\b}$.
The relative distance zeta function of $(A,\O)$ can be easily computed by the interested reader (see also \cite[Example 4.4.1]{fzf}) and is given by
\begin{equation}\label{zetaab}
\begin{aligned}
\zeta_{A_{\a,\b},\O}(s)&=\frac{2^{2-s}}{(s-1)}\sum_{j=1}^{\ty}j^{-\a/\b}\left(j^{-1/\b}-(j+1)^{-1/\b}\right)^{s-1}=\frac{2^{2-s}}{(s-1)}\zeta_{\mathcal{L},-\a/\b,1}(s),
\end{aligned}
\vspace{6pt}
\end{equation}
where $\mathcal{L}$ is the $\b^{-1}$-string.
In light of Corollary \ref{kor_zet}, we conclude that $\zeta_{A_{\a,\b},\O}(s)$ has a meromorphic continuation to all of $\Ce$ and
\begin{equation}
\po({\zeta}_{A_{\a,\b},\O})\subseteq\left\{1,2-\frac{1+\a}{1+\b}\right\}\cup\left\{D_m:m\in\eN\right\},
\end{equation}
where $D_m:=2-\frac{1+\a+m\b}{1+\b}$.
Let $D:=2-\frac{1+\a}{1+\b}$.
Also, by the same corollary and from \eqref{zetaab} we have that both $1$ and $D$ are simple poles of $\zeta_{A_{\a,\b},\O}$.
Furthermore, we have that $D>1$ and, consequently, from the tube formula \eqref{chirp_tube} just below, it follows that $\dim_B(A_{\a,\b},\O)=D$ and that the RFD $(A_{\a,\b},\O)$ is Minkowski measurable with Minkowski content given by
\begin{equation}
\begin{aligned}
\mathcal{M}^D(A_{\a,\b},\O)=\frac{2^{2-D}}{(2\! -\! D)(D\!-\! 1)}\frac{\b^{\frac{1+\a}{1+\b}}}{1\!+\! \b}=\frac{(2\b)^{2-D}}{(2\!-\! D)(D\!-\! 1)(1\!+\! \b)}.
\end{aligned}
\end{equation}
Moreover, the residue at $s=1$ is given by
\begin{equation}
\res({\zeta}_{A_{\a,\b},\O},1)=2\zeta_{\mathcal{L},-\a/\b,1}(1)=2\,\zeta\Big(\frac{\a}{\b}\Big).
\end{equation}
It follows that $s=1$ is indeed a simple pole of ${\zeta}_{A_{\a,\b},\O}(s)$.

Similarly as in Example \ref{ex_nest}, for $M\in\eN\cup\{0\}$, we choose the screen $\bm{S}_M$ to be a vertical line $\{\re s=\sigma\}$, for some real number $\sigma$ lying strictly between $2-\frac{1+\a+(M+1)\b}{1+\b}$ and $2-\frac{1+\a+(M+2)\b}{1+\b}$, and let $\bm{W}_M$ be the corresponding window.
From Theorem~\ref{dist_tube_formula_d}, we then obtain the following asymptotic distributional formula for the tube function $V(t):=|(A_{\a,\b})_t\cap\O|$:
\begin{equation}\label{chirp_tube}
\begin{aligned}
V(t)&=\frac{(2\b t)^{2-D}}{(2-D)(D-1)(1+\b)}+\frac{t^{2-D_1}\res({\zeta}_{A_{\a,\b},\O},D_1)}{2-D_1}+2t\,\zeta\Big(\frac{\a}{\b}\Big)\\
&\phantom{=}+\sum_{m=2}^{M}\frac{t^{2-D_m}\res({\zeta}_{A_{\a,\b},\O},D_m)}{2-D_m}+O(t^{2-D_{M+1}})\quad\textrm{as}\quad t\to0^+.
\end{aligned}
\end{equation}
Note that the second noninteger complex dimension, namely, $D_1=1-\frac{\a}{1+\b}$, is also greater than 1.
Finally, by choosing as a screen a vertical line to the right of $-\frac{2\a+\b}{1+\b}$, we actually obtain a pointwise formula still given by \eqref{chirp_tube} above; indeed, we then have $\kappa_d<0$, so that we can apply part $(i)$ of Theorem \ref{pointwise_thm_d}.
\end{example}

\subsection{Tube Formulas for Self-Similar Sprays}\label{subsec_self_similar_sp}

We conclude this paper by explaining how the results obtained here may also be applied to recover and significantly extend, as well as place within a general conceptual framework, the tube formulas for  self-similar sprays generated by an arbitrary open set $G\subset\eR^N$ of finite $N$-dimensional Lebesgue measure.
(See, especially, [LapPe2--3] 
extended to a significantly more general setting in \cite{lappewi1}, along with the exposition of those results given in \cite[Section~13.1]{lapidusfrank12}; see also \cite{DeKoOzUr} for another, but related, proof of some of those results.)

Recall that a self-similar spray (with a single generator $G$, assumed to be bounded and open) is defined as a collection $(G_k)_{k\in\eN}$ of pairwise disjoint (bounded) open sets $G_k\subset\eR^N$, with $G_0:=G$ and such that for each $k\in\eN$, $G_k$ is a scaled copy of $G$ by some factor $\lambda_k>0$.
(We let $\lambda_0:=1$.)
The associated scaling sequence $(\lambda_k)_{k\in\eN}$ is obtained from a ratio list $\{r_1,r_2,\ldots,r_J\}$, with $0<r_j<1$ for each $j=1,\ldots,J$ and such that $\sum_{j=1}^{J}r_j^N<1$, by considering all possible words built out of the scaling ratios $r_j$.
Here, $J\geq 2$ and the scaling ratios $r_1,\ldots,r_J$ are repeated according to their multiplicities

Let us next assume that $(A,\O)$ is the self-similar spray considered as a relative fractal drum  and defined as $A:=\partial(\sqcup_{k=0}^{\ty} G_k)$ and $\O:=\sqcup_{k=0}^{\ty} G_k$, with $\ov{\dim}_B(\partial G,G)<N$.
Then, by \cite[Theorem 3.36]{refds} or  \cite[Theorem 4.2.17]{fzf}, we have the following key formula, called a {\em factorization formula}, for its associated distance zeta function $\zeta_{A,\O}$, expressed in terms of the distance zeta function of the boundary of the generator (relative to the generator), $\zeta_{\pa G,G}$, and the scaling ratios $\{r_j\}_{j=1}^{J}$:
\begin{equation}\label{5.5.102..}
\zeta_{A,\O}(s)=\frac{\zeta_{\partial G,G}(s)}{1-\sum_{j=1}^{J}r_j^s}.
\end{equation}
(See also Remark \ref{5.5.20.1/2R} below.)
It now suffices to assume that the relative distance zeta function $\zeta_{\partial G,G}$ of the generating relative fractal drum $(\partial G,G)$ satisfies suitable languidity conditions in order to apply (at level $k=0$) the fractal tube formulas of Sections \ref{sec_point}--\ref{distance_tube} and to obtain a pointwise or distributional formula, with or without error term, for the `inner' volume of $\sqcup_{k=0}^{\ty} G_k$:\footnote{Here and throughout the rest of this subsection, we use the notation $V_{A,\O}$, consistent with the statement of a pointwise tube formula.
In the case of the distributional tube formulas, we should use instead the notation $\mathcal{V}_{A,\O}$.
(And analogously for the error term $R_{A,\O}(t)$ in \eqref{spray_formula}, which should then be denoted by $\mathcal{R}_{A,\O}(t)$, in the distributional case.)
For notational simplicity, however, we will not do so in this discussion.}
\begin{equation}\label{spray_formula}
\begin{aligned}
V_{A,\O}(t)&:=|A_t\cap\O|\\
&\phantom{:}=\!\!\!\sum_{\omega\in(\mathfrak{D}\cap \bm{W})\cup\po(\zeta_{\partial G,G},\bm{W})}\!\!\!\res\left(\frac{t^{N-s}\zeta_{\partial G,G}(s)}{(N-s)\Big(1-\sum_{j=1}^Jr_j^s\Big)},\omega\right)+ R_{A,\O}(t),
\end{aligned}
\end{equation}
where $\mathfrak{D}$ denotes the set of solutions in $\Ce$ of $\sum_{j=1}^Jr_j^s=1$, the complexified Moran equation, and $R_{A,\O}:=R_{A,\O}^{[0]}$ is a pointwise or distributional error term (or else $R_{A,\O}(t)\equiv 0$ and $\bm{W}:=\Ce$, in the case of an exact tube formula, provided $\zeta_{\pa G,G}$ is strongly $d$-languid), depending on the $d$-languidity growth conditions satisfied by $\zeta_{\partial G,G}$.

In the $d$-languid (but not necessarily strongly $d$-languid) case, $R_{A,\O}=R_{A,\O}^{[0]}$ satisfies the following (pointwise or distributional) error estimate (at level $k=0$):
\begin{equation}\label{5.5.105.1/4E}
R_{A,\O}(t)=O(t^{N-\sup S})\q\textrm{as}\q t\to 0^+,
\end{equation}
where $S$ is the screen associated to the window $\bm W$.

\begin{remark}\label{5.5.20.1/2R}
Observe that we can rewrite Equation \eqref{5.5.102..} as follows:
\begin{equation}\label{5.5.105.1/2E}
\zeta_{A,\O}(s)=\zeta_{\mathfrak{S}}(s)\cdot\zeta_{\pa G,G}(s),
\end{equation}
where the geometric zeta function $\zeta_{\mathfrak{S}}$ of the associated self-similar string (with scaling ratios $\{r_j\}_{j=1}^J$ and a single gap length, equal to one, in the terminology of \cite[Chapters 2 and 3]{lapidusfrank12}) is meromorphic in all of $\Ce$ and given for all $s\in\Ce$ by
\begin{equation}\label{5.5.105.3/4E}
\zeta_{\mathfrak{S}}(s)=\frac{1}{1-\sum_{j=1}^Jr_j^s}.
\end{equation}
In general, given a connected open set $U\subseteq\Ce$, $\zeta_{A,\O}$ is meromorphic in $U$ if and only if $\zeta_{\pa G,G}$ is; furthermore, in that case, the {\em factorization formula} \eqref{5.5.105.1/2E} (or \eqref{5.5.102..}) then holds for all $s\in U$.
We note that in the sequel and following [LapPe2--3] 
and [LapPeWi1--2], 
we will often refer to $\zeta_{\mathfrak{S}}$ as the {\em scaling zeta function} of the self-similar spray $(A,\O)$ and to its poles in $\Ce$ (composing the multiset $\mathfrak{D}$)\label{matfrD} as the {\em scaling complex dimensions} of $(A,\O)$.
We will also sometimes write $\mathfrak{D}_{\mathfrak{S}}$ instead of $\mathfrak{D}$, so that $\mathfrak{D}_{\mathfrak{S}}:=\mathfrak{D}$; hence, similarly, $\mathfrak{D}_{\mathfrak{S}}\cap W=\mathfrak{D}\cap W$, the set of {\em visible scaling complex dimensions} of $(A,\O)$, denotes the set of poles of $\zeta_{\mathfrak{S}}$ visible through the window $\bm W$.
(See Equation \eqref{spray_formula} above.)

Typically, we will work with generators such that $\zeta_{\pa G,G}$ is strongly $d$-languid and consequently, since $\zeta_{\mathfrak{S}}$ (as given by \eqref{5.5.105.3/4E}) is strongly $d$-languid (after a possible scaling by an appropriate scaling factor $\lambda_{\mathfrak{S}}>0$; see Corollary \ref{str_pointwise_formula_d} and the discussion preceding it), $\zeta_{A,\O}$ will be strongly $d$-languid (also after a possible scaling by the same scaling factor $\lambda_{\mathfrak{S}}$) and given by the factorization formula \eqref{5.5.105.1/2E} (or \eqref{5.5.102..}), for all $s\in\Ce$.
As a result, unless we need to work with a `truncated tube formula' (corresponding to a fractal tube formula with error term associated with a suitable screen $\bm S$), we will be able to obtain an {\em exact} fractal tube formula, as we will now see.
\end{remark}

Assume next that the generator $G$ is {\em monophase} (in the sense of [LapPe2--3] 
and [LapPeWi1--2]); 
 that is, the volume of its `inner' $t$-neighborhood is given by a polynomial $\sum_{i=0}^{N-1}\kappa_it^{N-i}$ for all $t\in\eR$ such that $0<t<g$.
Here, $g$ is the {\em inradius} of $G$, i.e., the supremum of the radii of all the balls which are contained in $G$.
Since then,
\begin{equation}
V_{\pa G,G}(t):=|(\partial G)_t\cap G|=\sum_{i=0}^{N-1}\kappa_it^{N-i},
\end{equation}
for $0<t<g$, we can explicitly calculate the relative tube zeta function of $G$, as follows:
\begin{equation}\label{zeta_monoph}
\widetilde{\zeta}_{\partial G,G}(s;g)=\int_0^{g}t^{s-N-1}\sum_{i=0}^{N-1}\kappa_it^{N-i}\di t=\sum_{i=0}^{N-1}\frac{\kappa_ig^{s-i}}{s-i}.
\end{equation}
It is obviously meromorphic on all of $\Ce$ and still given by \eqref{zeta_monoph} for all $s\in\Ce$.

Using the functional equation which connects the relative tube and distance zeta functions (see Equation \eqref{equ_tilde}), we now obtain the following explicit expression for the relative distance zeta function of the generator $G$:
\begin{equation}\label{ispravka}
\begin{aligned}
\zeta_{\pa G,G}(s)&:={\zeta}_{\partial G,G}(s;g)=g^{s-N}|(\pa G)_g\cap G|+(N-s)\widetilde{\zeta}_{\partial G,G}(s;g)\\
&\phantom{:}=g^{s-N}|G|+(N-s)\sum_{i=0}^{N-1}\frac{\kappa_ig^{s-i}}{s-i}=(N-s)\sum_{i=0}^{N}\frac{\kappa_ig^{s-i}}{s-i},
\end{aligned}
\end{equation}
where we have let $\kappa_N:=-|G|$.

Consequently, by substituting \eqref{ispravka} into \eqref{spray_formula}, we recover (and significantly extend as well as place within the broader framework of the theory of fractal tube formulas via fractal zeta functions) a well-known result obtained in \cite{lappe2} and more generally in \cite{lappewi1}, as well as more recently via a different (but related) technique in \cite{DeKoOzUr}:
\begin{equation}\label{spray_formula_1}
V_{A,\O}(t):=|A_t\cap\O|=\!\!\!\!\sum_{\omega\in\mathfrak{D}\cup\{0,1,\ldots,N-1\}}\!\!\!\!\res\left(t^{N-s}\frac{{\sum_{i=0}^{N}\kappa_i\frac{g^{s-i}}{s-i}}}{\Big(1-\sum_{j=1}^Jr_j^s\Big)},\omega\right).
\end{equation}
This is an exact pointwise fractal tube formula.
Indeed, after an appropriate scaling by a factor $\lambda_G>0$, ${\zeta}_{\pa G,G}$ is shown to be strongly $d$-languid with $(\kappa_d)_G:=0$ for a suitable infinite sequence of vertical lines $\{\re s=\alpha_m\}$, $m\geq 1$ with $\alpha_m\in\eR$ and $\alpha_m\to -\ty$ as $m\to \ty$.
Also, it is easy to check (after an appropriate scaling by a factor $\lambda_{\mathfrak{S}}>0$) that $\zeta_{\mathfrak{S}}(s)=(1-\sum_{j=1}^Jr_j^s)^{-1}$ is strongly $d$-languid, with $(\kappa_d)_{\mathfrak{S}}:=0$ (see \cite[Equation $(6.36)$, p.\ 195]{lapidusfrank12}).
Hence (after a suitable scaling by $\lambda:=\lambda_{A,\O}$, depending on both $\lambda_G$ and $\lambda_{\mathfrak{S}}$), we deduce from the factorization formula \eqref{5.5.102..} (or, equivalently, \eqref{5.5.105.1/2E}) that $\zeta_{A,\O}$ is strongly $d$-languid, with exponent $(\kappa_d)_{A,\O}:=0$ for this same sequence of vertical lines $\{\re s=\alpha_m\}$, $m\geq 1$.
We can therefore conclude from Theorem \ref{pointwise_thm_d} that the tube formula \eqref{spray_formula_1} is valid pointwise and without an error term in this case, for all positive $t$ sufficiently small.\footnote{For an exact interval within which the fractal tube formula is valid, one has to explicitly calculate the scaling factors $\lambda_G$ and $\lambda_{\mathfrak{S}}$, which we leave as an exercise for the interested reader.}

If needed, one can also obtain a corresponding `truncated' pointwise fractal tube formula (with error term), relative to a suitable screen.

A completely analogous reasoning can be used for the case of {\em pluriphase} generators $G$ for which the `inner' tubular volume is given as a piecewise polynomial.
(See \cite{lappe2} or \cite{lappewi1} for the corresponding precise definition.)
In a future work, we plan to investigate for which classes of generators the tube formula \eqref{spray_formula} can be applied pointwise or distributionally.
It is clearly a very large class, corresponding to essentially all of the self-similar sprays (and hence, also all of the self-similar tilings, in the sense of [Pe, LapPe2--3, LapPeWi1--2, PeWi]) 
of interest, including (in light of the results of \cite{KoRati}) those with generators that are convex polyhedra (or polytopes), under mild assumptions.

\begin{remark}\label{5.5.21.1/2R}
We point out that in \cite{lappewi1}, which (prior to the present work and that in [LapRa\v Zu1--5]) 
was the paper containing the most elaborate results concerning the fractal tube formulas for self-similar sprays (and other fractal sprays), a lot of effort was required to obtain analogous (but less general) fractal tube formulas, with or without error term.
Furthermore, the `tubular zeta functions' used in [LapPe2--3] 
and, in the more general context of \cite{lappewi1}, were introduced in an ad hoc manner.
Here, by contrast, both the fractal tube formulas and the fractal zeta functions (in the present situation, the distance zeta functions) occurring in the corresponding formulas follow naturally from the general theory developed in [LapRa\v Zu1--6],  
and in particular, in the present paperr.

We close this comment by mentioning that the interested reader can find in [LapPe2--3], 
[LapPeWi1--2], 
as well as in the exposition given in \cite[Subsection 13.1.4]{lapidusfrank12}, a number of concrete examples illustrating the fractal tube formulas for self-similar sprays (or tilings).
These examples include the Cantor tiling, the Koch tiling, the Sierpi\'nski gasket and carpet tilings, along with the pentagasket tiling (see, e.g., \cite[Example 13.33]{lapidusfrank12}), which is an interesting and natural example of a self-similar spray with multiple generators.
In all of these examples, the underlying generators of the fractal sprays (or tilings) are convex polygons and therefore, satisfy the required assumptions.
\end{remark}

The next three examples illustrate interesting phenomena that may occur in the inhomogeneous self-similar setting, in the sense of \cite{fraser} (see also \cite{bd,fraser2,bfm,os}).
We obtain here their corresponding fractal tube formulas and illustrate the interesting situations which may arise, in particular, for self-similar sprays (or tilings), or, more generally, for self-similar RFDs.
These examples enable us, in particular, to illustrate our proposed definition of (critical and subcritical) fractality (see Remark \ref{5.5.9.1/2}).
Accordingly, the sets and RFDs considered in these examples are indeed fractal, in that sense, and their fractality reflects their intrinsic geometric oscillations, as is made evident by the corresponding fractal tube formulas.

\begin{figure}[ht]
\begin{center}
\includegraphics[width=5cm]{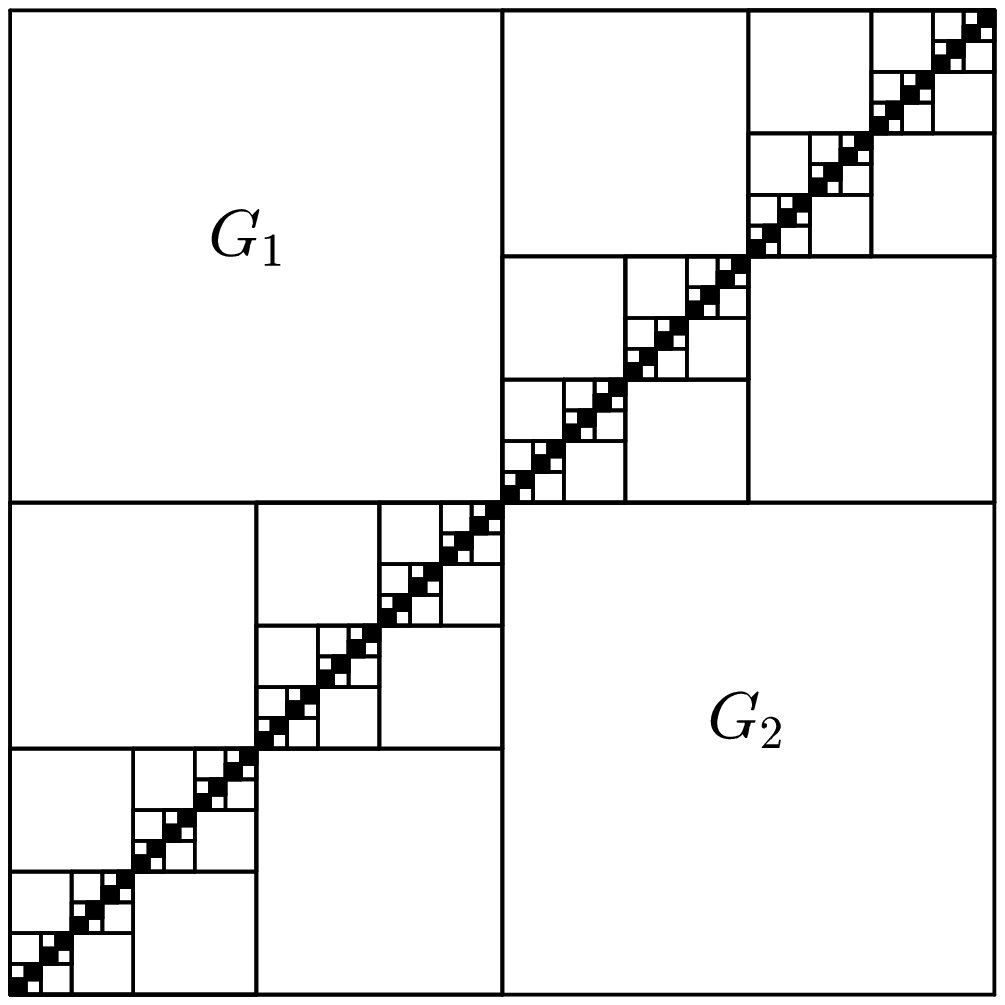}
\includegraphics[trim=0 0.8cm 0 0,clip,width=5.7cm]{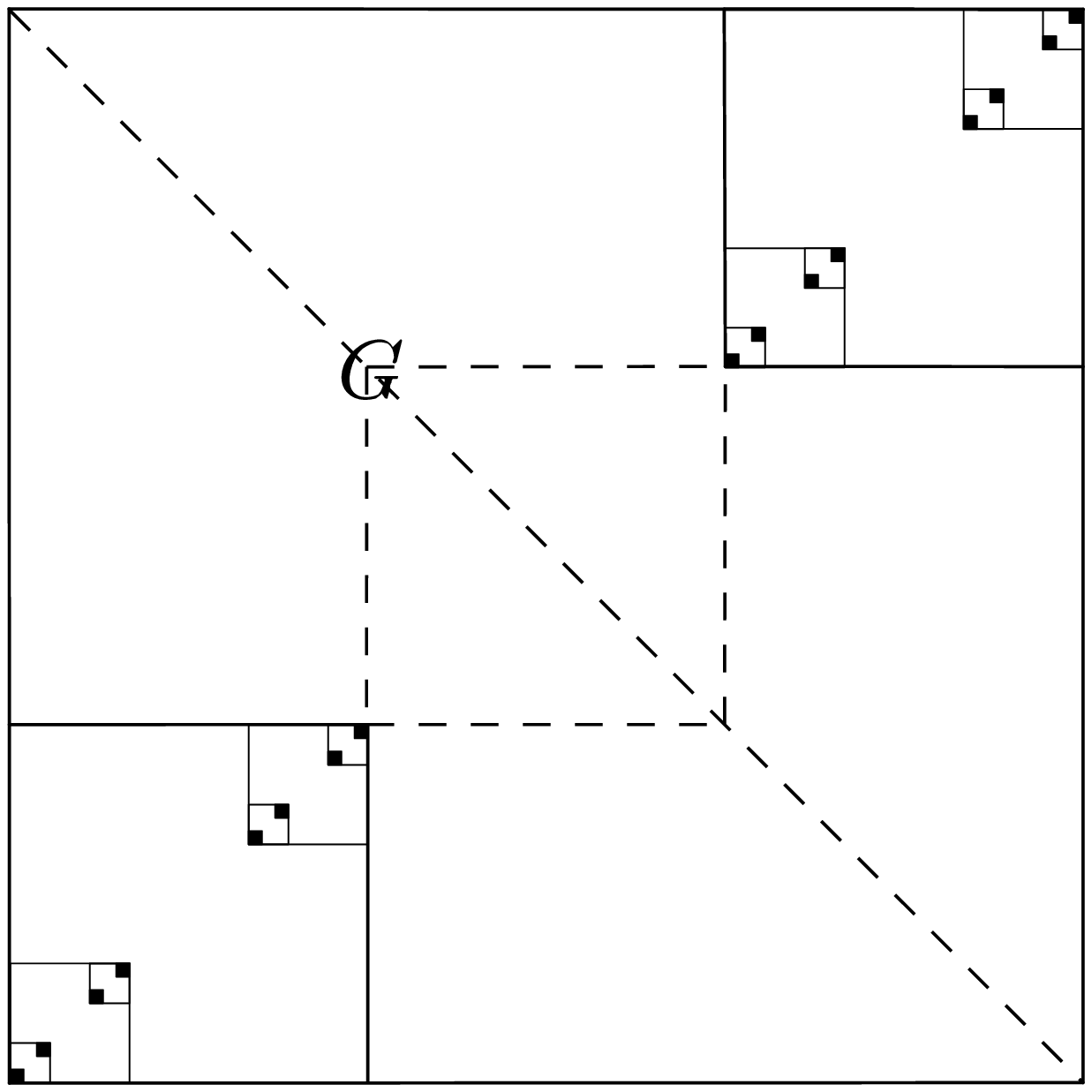}
\caption{{\bf Left:} The $1/2$-square fractal $A$ from Example \ref{1/2-tube_formula}. We start with a unit square $[0,1]^2$ and in the first step remove the open squares $G_1$ and $G_2$. In the next step, we repeat this with the remaining squares $[1/2,1]^2$ and $[0,1/2]^2$; we then continue this process ad infinitum and, by definition, $A\subset\eR^2$ is the compact set which remains behind. The first 6 iterations are depicted. Here, $G:=G_1\cup G_2$ is the single generator of the corresponding self-similar spray or RFD $(A,\O)$, where $\O=(0,1)^2$. {\bf Right:} The $1/3$-square fractal $A$ from Example \ref{1/3-tube_formula}. We start with a unit square $[0,1]^2$ and, in the first step, remove the open polygon $G$. In the next step, we repeat this with the remaining squares $[1/3,1]^2$ and $[0,1/3]^2$; we then continue this process ad infinitum and, by definition, $A\subset\eR^2$ is the compact set which remains behind. The first 4 iterations are depicted. Here, $G$ is the single generator of the corresponding self-similar spray or RFD $(A,\O)$, where $\O:=(0,1)^2$.}\label{kv_0.5}
\end{center}
\end{figure}

\begin{example}({\em Fractal tube formula for the $1/2$-square fractal}).\label{1/2-tube_formula}
Let us consider the $1/2$-square fractal $A$ from \cite[Example 3.38]{refds} and depicted in Figure \ref{kv_0.5}, left.
Its distance zeta function was obtained in \cite{refds}, where it was shown to be meromorphic on all of $\Ce$ and given by
\begin{equation}\label{dist_1/2_1}
\zeta_A(s)=\frac{2^{-s}}{s(s-1)(2^s-2)}+\frac{4}{s-1}+\frac{2\pi}{s},
\end{equation}
for every $s\in\Ce$.
In \eqref{dist_1/2_1}, without loss of generality, we have chosen $\delta:=1$.
Furthermore, as was discussed in \cite{refds}, it follows at once from \eqref{dist_1/2_1} that
\begin{equation}
\begin{aligned}
D(\zeta_A)=1,\q
\po(\zeta_A):=\po(\zeta_A,\Ce)=\{0\}\cup\left(1+\mathbf{p}\I\Ze\right)
\end{aligned}
\end{equation}
and 
\begin{equation}\label{PCsq1}
\dim_{PC} A:=\po_c(\zeta_A)=1+\mathbf{p}\I\Ze,
\end{equation}
where the oscillatory period $\mathbf{p}$ of $A$ is given by $\mathbf{p}:=\frac{2\pi}{\log 2}$ and all of the complex dimensions in $\po(\zeta_A)$ are simple, except for $\omega_0=1$ which is a double pole of $\zeta_A$.

One easily sees that $\lambda A$ is strongly $d$-languid for $\kappa_d:=-1$, any $\lambda\geq 2$ and a sequence of screens consisting of the vertical lines $\{\re s=-m\}$, $m\in\eN$, along with the constant $B_{\lambda}:=2/\lambda$ in the strong languidity condition {\bf L2'}.
Therefore, we can use Theorem~\ref{pointwise_thm_d} in order to recover the following exact pointwise fractal tube formula, valid for all $t\in(0,\min\{1/\lambda,1/2\})=(0,1/2)$:
\begin{equation}\label{racun_1/2}
\begin{aligned}
V_{A}(t):=|A_t|&=\sum_{\omega\in\po({\zeta}_{A})}\res\left(\frac{t^{2-s}}{2-s}{\zeta}_A(s),\omega\right)\\
&=\res\left(\frac{t^{2-s}}{2-s}{\zeta}_A(s),1\right)+\sum_{\omega\in\po({\zeta}_{A})\setminus\{1\}}\frac{t^{2-\omega}}{2-\omega}\res\left({\zeta}_A,\omega\right).
\end{aligned}
\end{equation}
We now let $\omega_k:=1+\I\mathbf{p}k$ for each $k\in\Ze$ and note that
\begin{equation}
\res\left({\zeta}_A,0\right)=1+2\pi\q\textrm{and}\q\res\left({\zeta}_A,\omega_k\right)=\frac{4^{-\I\mathbf{p}k}}{4\omega_k(\omega_k-1)},
\end{equation}
for every $k\in\Ze\setminus\{0\}$.

In order to compute the residue at $\omega_0=1$ in \eqref{racun_1/2}, we reason analogously as in Example \ref{ex_nest} (see Equation \eqref{taylor} and the text surrounding it) to conclude that
\begin{equation}\label{res_1}
\begin{aligned}
\res\left(\frac{t^{2-s}}{2-s}{\zeta}_A(s),1\right)&=t\sum_{n=0}^1\sum_{k=0}^n\frac{(-1)^{n-k}(\log t^{-1})^k\zeta_A[1]_{-n-1}}{k!(n-k)!}\\
&=t\,\left(\zeta_A[1]_{-1}-\zeta_{A}[1]_{-2}+\zeta_A[1]_{-2}\log t^{-1}\right).
\end{aligned}
\end{equation}
(Recall that for $q\in\Ze$, $\zeta_A[1]_{q}$ denotes the $q$-th coefficient in the Laurent expansion of $\zeta_A$ around $s=1$.)
The coefficients $\zeta_{A}[1]_{-2}$ and $\zeta_A[1]_{-1}$ are not difficult to compute and are given by 
\begin{equation}
\zeta_{A}[1]_{-2}=\frac{1}{4\log 2}\q\textrm{and}\q\zeta_A[1]_{-1}=\frac{29\log 2-2}{8\log 2},
\end{equation}
which, combined with \eqref{res_1}, yields
\begin{equation}
\res\left(\frac{t^{2-s}}{2-s}{\zeta}_A(s),1\right)=\frac{1}{4\log 2}t\log t^{-1}+\frac{29\log 2-4}{8\log 2}t.
\end{equation}
Finally, we obtain the following exact fractal tube formula for the $1/2$-square fractal $A$, valid for all $t\in(0,1/2)$:
\begin{equation}\label{1/2_tube_formula_eq}
\begin{aligned}
V_{A}(t):=|A_t|&=\frac{1}{4\log 2}t\log t^{-1}+\frac{29\log 2-4}{8\log 2}t\\
&\phantom{=}+t\sum_{k\in\Ze\setminus\{0\}}\frac{(4t)^{-\I\mathbf{p}k}}{4\omega_k(\omega_k-1)(2-\omega_k)}+\frac{1+2\pi}{2}t^2\\
&=\frac{1}{4\log 2}t\log t^{-1}+t\,G\left(\log_2(4t)^{-1}\right)+\frac{1+2\pi}{2}t^2,
\end{aligned}
\end{equation}
where $G$ is a nonconstant $1$-periodic function on $\eR$, which is bounded away from zero and infinity.
It is given by the following absolutely convergent (and hence, convergent) Fourier series:
\begin{equation}
G(x):=\frac{29\log 2-4}{8\log 2}+\frac{1}{4}\sum_{k\in\Ze\setminus\{0\}}\frac{\E^{2\pi\I kx}}{(2-\omega_k)(\omega_k-1)\omega_k},\quad\textrm{for all}\ x\in\eR. 
\end{equation}

To conclude our discussion of this example, we note that it is now clear from the fractal tube formula \eqref{1/2_tube_formula_eq} for the $1/2$-square fractal that $\dim_BA=1$ and that $A$ is Minkowski degenerate with $\mathcal{M}^1(A)=+\ty$.
On the other hand, $A$ is $h$-Minkowski measurable with $h(t):=\log t^{-1}$ (for all $t\in(0,1)$) and with $h$-Minkowski content given by
$
\mathcal{M}^1(A,h)=(4\log 2)^{-1}.
$
Finally, although $D:=\dim_BA=1$ (which is also the topological dimension of $A$) and hence, $A$ would not be considered fractal in the classical sense, we also see from \eqref{1/2_tube_formula_eq} that the nonreal complex dimensions of $A$ with real part equal to $D$ give rise to (intrinsic) geometric oscillations of order $t^{2-D}$ (or simply, $2-D$) in its fractal tube formula.
More specifically, according to our proposed definition of fractality given in Remark \ref{5.5.9.1/2}, $A$ is {\em critically fractal} in dimension $d:=D=\dim_BA=1$.
\end{example}

\begin{example}({\em Fractal tube formula for the $1/3$-square fractal}).\label{1/3-tube_formula}
Let us now consider the $1/3$-square fractal $A$ from \cite[Example 3.39]{refds} and depicted in Figure \ref{kv_0.5}, right.
Its distance zeta function was obtained in \cite{refds}, where it was shown to be meromorphic on all of $\Ce$ and given by
\begin{equation}\label{zeta_1/3_square_2}
\zeta_{A}(s)=\frac{2}{s(3^s-2)}\left(\frac{6}{s-1}+Z(s)\right)+\frac{4}{s-1}+\frac{2\pi}{s},
\end{equation}
for all $s\in\Ce$.
Here, the entire function $Z$ is given by
$
Z(s):=\int_0^{\pi/2}(\cos\varphi+\sin\varphi)^{-s}\di\varphi
$
and,  without loss of generality, we have chosen $\d:=1$.
Furthermore, as was discussed in \cite{refds}, it follows at once from \eqref{zeta_1/3_square_2} that
$
D(\zeta_A)=1
$
and
\begin{equation}\label{c_dim_A1/3_2}
\po(\zeta_A):=\po(\zeta_A,\Ce)\subseteq\{0\}\cup\left(\log_32+\mathbf{p}\I\Ze\right)\cup\{1\},
\end{equation}
where the oscillatory period $\mathbf{p}$ of $A$ is given by $\mathbf{p}:=\frac{2\pi}{\log 3}$ and all of the complex dimensions in $\po(\zeta_A)$ are simple.
In Equation \eqref{c_dim_A1/3_2}, we only have an inclusion since, at least in principle, some of the complex dimensions with real part $\log_32$ may be canceled by the zeros of $6/(s-1)+Z(s)$.
However, it can be checked numerically that there exist nonreal complex dimensions with real part $\log_32$ in $\po(\zeta_A)$. 
Furthermore, observe that we have
\begin{equation}
|Z(s)|\leq\begin{cases}
2^{-\re s/2-1}\pi&\textrm{if}\q \re s<0,\\
\pi/2&\textrm{if}\q \re s\geq 0,
\end{cases}
\end{equation}
from which we conclude that that $\lambda A$ is strongly $d$-languid for $\kappa_d:=-1$, any $\lambda\geq \sqrt{2}$ and a sequence of screens consisting of the vertical lines $\{\re s=-m\}$, $m\in\eN$, along with the constant $B_{\lambda}:=\sqrt{2}/\lambda$ in the strong languidity condition {\bf L2'}.
Therefore, we can use Theorem~\ref{pointwise_thm_d} to recover the following exact pointwise fractal tube formula, valid for all $t\in(0,\min\{1/\lambda,1/\sqrt{2}\})=(0,1/{\sqrt{2}})$:
\begin{equation}\label{racun_1/3}
\begin{aligned}
V_{A}(t):=|A_t|&=\sum_{\omega\in\po({\zeta}_{A})}\res\left(\frac{t^{2-s}}{2-s}{\zeta}_A(s),\omega\right)=\sum_{\omega\in\po({\zeta}_{A})}\frac{t^{2-s}}{2-s}\res\left({\zeta}_A,\omega\right)\\
&=16t+\frac{t^{2-\log_32}}{\log 3}\sum_{k=-\ty}^{+\ty}\frac{(3t)^{-\I\mathbf{p}k}}{\omega_k(2-\omega_k)}\left(\frac{6}{\omega_k-1}+Z(\omega_k)\right)+\frac{12+\pi}{2}t^2\\
&=16t+t^{2-\log_32}G\left(\log_3(3t)^{-1}\right)+\frac{12+\pi}{2}t^2.
\end{aligned}
\end{equation}
Here, we have used the fact that
\begin{equation}
\res\left({\zeta}_A,1\right)=16,\q\res\left({\zeta}_A,0\right)=12+\pi
\end{equation}
and
\begin{equation}
\res\left({\zeta}_A,\omega_k\right)=\frac{3^{-\I\mathbf{p}k}}{(\log 3)\omega_k}\left(\frac{6}{\omega_k-1}+Z(\omega_k)\right),
\end{equation}
where we have let $\omega_k:=\log_32+\I\mathbf{p}k$ for each $k\in\Ze$.
It can be checked numerically that $\res\left({\zeta}_A,\omega_k\right)\neq 0$ (at least) for $k=-1,0,1$ and we conjecture that this is also true for all $k\in\Ze$.\footnote{We caution the reader that we do not have a rigorous proof of this last statement.}
However, the fact that $\res\left({\zeta}_A,\omega_k\right)\neq 0$ for $k=-1,0,1$ suffices to deduce that the function $G$ in the last line of \eqref{racun_1/3} is a {\em nonconstant} $1$-periodic function on $\eR$, which is bounded away from zero and infinity and is given by the following absolutely convergent (and hence, convergent) Fourier series:
\begin{equation}
G(x):=\frac{1}{\log 3}\sum_{k=-\ty}^{+\ty}\frac{\E^{2\pi\I kx}}{(2-\omega_k)\omega_k}\left(\frac{6}{\omega_k-1}+Z(\omega_k)\right),\quad\textrm{for all}\ x\in\eR. 
\end{equation}

In conclusion, we observe that it is clear from the fractal tube formula \eqref{racun_1/3} that $\dim_BA=1$ and $A$ is Minkowski measurable, with Minkowski content given by
$
\mathcal{M}^1(A)=16.
$
Moreover, since the set $A$ is rectifiable, we have that $H^1(A)=\mathcal{M}^1(A)/2=8$, which can, of course, also be computed directly. 
On the other hand, although $D:=\dim_BA=1$ (which also coincides with the topological dimension of $A$) and thus $A$ would not be considered fractal in the classical sense, we also see from \eqref{racun_1/3} that the nonreal complex dimensions of $A$ with real part equal to $\log_32$ give rise to (intrinsic) geometric oscillations of order $t^{2-\log_32}$ (or simply, $2-\log_32$) in its fractal tube formula.
Therefore, according to our proposed definition of fractality given in Remark \ref{5.5.9.1/2}, the $1/3$-square fractal $A$ is fractal; more precisely, it is {\em strictly subcritically fractal} in dimension $d:=\log_32$.
\end{example}

\begin{example}({\em Fractal tube formula for the self-similar fractal nest}).\label{ss_nest_tube_formula}
Let us now consider the self-similar fractal nest $A$ from \cite[Example 3.40]{refds}, which is defined as a union of concentric circles in $\eR^2$ centered at the origin and of radius $a^k$ for $k\in\eN_0$, with $a\in(0,1)$ being a real parameter.
Its distance zeta function was obtained in \cite[Example 4.25]{refds}, where it was shown to be meromorphic on all of $\Ce$ and given by
\begin{equation}\label{ss_nest_zeta_2}
\zeta_{A}(s)=\frac{2^{2-s}\pi(1+a)(1-a)^{s-1}}{(s-1)(1-a^s)}+\frac{2\pi}{s-1}+\frac{2\pi}{s},
\end{equation}
for all $s\in\Ce$, where without loss of generality, we have chosen $\d:=1$.
Recall from \cite[Example 4.25]{refds} that we have
$
D(\zeta_A)=1
$
and
\begin{equation}\label{pizA_1}
\po(\zeta_A):=\po(\zeta_A,\Ce)=\mathbf{p}\I\Ze\cup\{1\},
\end{equation}
where the oscillatory period $\mathbf{p}$ of $A$ is given by $\mathbf{p}:=\frac{2\pi}{\log a^{-1}}$ and all of the complex dimensions in $\po(\zeta_A)$ are simple.

It is now easy to check that $\lambda A$ is strongly $d$-languid with $\kappa_d:=-1$, any $\lambda\geq 2$ if $a\in(0,1/2]$ or any $\lambda\geq 2(1-a)/a$ if $a\in(1/2,1)$ and (in both cases) for a sequence of screens consisting of vertical lines $\{\re s=-m\}$, $m\in\eN$, in the strong languidity condition {\bf L2'}.
Again, we can use Theorem~\ref{pointwise_thm_d} in order to obtain the following exact pointwise fractal tube formula, valid for all $t\in(0,\min\{1/2,a/{2(1-a)}\})$:
\begin{equation}\label{racun_ss_nest}
\begin{aligned}
V_{A}(t):=|A_t|&=\sum_{\omega\in\po({\zeta}_{A})}\res\left(\frac{t^{2-s}}{2-s}{\zeta}_A(s),\omega\right)=\sum_{\omega\in\po({\zeta}_{A})}\frac{t^{2-s}}{2-s}\res\left({\zeta}_A(s),\omega\right)\\
&=\frac{4\pi}{1-a}t+\left(\pi+\frac{4\pi(1+a)}{(\log a^{-1})(1-a)}\sum_{k=-\ty}^{+\ty}\frac{\left(\frac{2t}{1-a}\right)^{-\I\mathbf{p}k}}{(\omega_k-1)(2-\omega_k)}\right)t^2\\
&=\frac{4\pi}{(1-a)}t+t^2\,G\left(\log_{a^{-1}}\left(\frac{2t}{1-a}\right)\right).
\end{aligned}
\end{equation}
Here, we have used the fact that
\begin{equation}
\res\left({\zeta}_A,1\right)=\frac{4\pi}{1-a},\q\res\left({\zeta}_A,0\right)=2\pi+\frac{4\pi(1+a)}{(\log a)(1-a)}
\end{equation}
and
\begin{equation}
\res\left({\zeta}_A,\omega_k\right)=\frac{4\pi(1+a)}{(\log a^{-1})(\omega_k-1)}\left(\frac{2}{1-a}\right)^{-\I\mathbf{p}k},
\end{equation}
where we have let $\omega_k:=\I\mathbf{p}k$ for each $k\in\Ze$.
Furthermore, the function $G$ appearing in the last line of \eqref{racun_ss_nest} is a nonconstant $1$-periodic function on $\eR$, which is bounded away from zero and infinity and is given by the following absolutely convergent (and hence, convergent) Fourier series:
\begin{equation}
G(x):=\pi+\frac{4\pi(1+a)}{(\log a^{-1})(1-a)}\sum_{k=-\ty}^{+\ty}\frac{\E^{2\pi\I kx}}{(2-\omega_k)(\omega_k-1)},\quad\textrm{for all}\ x\in\eR. 
\end{equation}

It clearly follows from the fractal tube formula \eqref{racun_ss_nest} that $\dim_BA=1$ and $A$ is Minkowski measurable with Minkowski content given by
$
\mathcal{M}^1(A)=\frac{4\pi}{1-a}.
$
Furthermore, since the set $A$ is rectifiable, we have that $H^1(A)=\mathcal{M}^1(A)/2=2\pi/(1-a)$, which, of course, can also be easily checked via a direct computation.

Finally, we conclude this example by observing that although $D:=\dim_BA=1$ (which is also the topological dimension of $A$) and $A$ would not be considered fractal in the classical sense, we also see from \eqref{racun_ss_nest} that the nonreal complex dimensions of $A$ with real part equal to $0$ give rise to (intrinsic) geometric oscillations of order $t^{2}$ (or simply $2$) in its fractal tube formula.
Consequently, according to our proposed definition of fractality given in Remark \ref{5.5.9.1/2}, the self-similar fractal nest $A$ is indeed fractal; more precisely, it is {\em strictly subcritically fractal} in dimension $d:=0$.
\end{example}

We close this section by mentioning that one could provide many further examples illustrating our tube formulas, as applied to self-similar sprays or, more generally, fractal sprays.
These examples would include the Koch curve tiling, the Sierpi\'nski gasket tiling, the pentagasket tiling and the Menger sponge tiling depicted, respectively, in Figures 6.1--6.5 of \cite{lappewi1}.
We refer to [LaPe2--3] for the corresponding exact (pointwise) fractal tube formulas.
We point out that the pentagasket tiling is of special interest because it is a natural example of a self-similar spray with multiple generators.\footnote{Of course, in the case of fractal sprays with multiple (say $Q$) generators, it suffices to apply the results of the present subsection to each of the corresponding $Q$ fractal sprays with a single generator, and then to add-up the resulting $Q$ fractal tube formulas.}

Other interesting examples include the Cantor carpet, the $U$-shaped modification of the Sierpi\'nski carpet (which has a generator which is itself ``fractal''), the binary trees, and the Apollonian packings depicted, respectively, in Figures 6.6, 6.9, 6.11 and 6.12 of \cite{lappewi1} and whose associated fractal tube formulas are provided or discussed in Subsections 6.1--6.4 of \cite{lappewi1}.

We also mention that the authors have recently obtained an explicit fractal tube formula for the Koch drum (or the Koch snowflake RFD), by using the general theory developed in this paper.
This important example (which is {\em not} a fractal spray) should be discussed in a later work, as well as compared with the earlier tube formula obtained by the first author and E.\ Pearse in \cite{lappe1}, as described in \cite[Subsection 12.2.1]{lapidusfrank12}.

Finally, we point out that our methods apply naturally to fractal sprays which are not necessarily self-similar (such as the last three examples mentioned just above).
Moreover, as was alluded to in the introduction of this paper, our general (pointwise or distributional) fractal tube formulas can be extended (under suitable hypotheses) to include the case where the associated fractal zeta function have nonremovable singularities which are not poles.
Some examples of such situations are provided in [LapRa\v Zu1,4] 
but we plan to develop the corresponding systematic theory in a later work.

\end{document}